\documentclass[12pt]{iopart}
\usepackage{newtxtext}
\usepackage{color}
\usepackage{comment}
\usepackage{cite}
\usepackage[dvipdfmx,hypertexnames=false]{hyperref}
\hypersetup{colorlinks=true,linkcolor=blue,citecolor=blue}
\usepackage{amssymb}
\usepackage{amsthm}
\usepackage{iopams}
\usepackage{graphicx}

\theoremstyle{definition}
\newtheorem{Def}{Definition}[section]
\newtheorem{lem}[Def]{Lemma}
\newtheorem{thm}[Def]{Theorem}

\newtheorem{eg}[Def]{Example}
\newtheorem{rem}[Def]{Remark}

\newcommand{\od}[3][1]{
\ifnum #1=1\frac{\mathrm{d} #2}{\mathrm{d} #3}
\else\frac{\mathrm{d}^{#1} #2}{\mathrm{d} #3^{#1}}
\fi}
\newcommand{\pd}[3][1]{
\ifnum #1=1\frac{\partial #2}{\partial #3}
\else\frac{\partial^{#1} #2}{\partial #3^{#1}}
\fi}

\makeatletter
  \renewcommand{\theequation}
  {\arabic{section}.\arabic{equation}}
  \@addtoreset{equation}{section}

	\renewcommand{\thefigure}{
	\thesection.\arabic{figure}}
	\@addtoreset{figure}{section}

	\renewcommand{\thetable}{
	\thesection.\arabic{table}}
	\@addtoreset{table}{section}
\makeatother

\begin{document}

\title[Integrable full discretization of the multi-component short pulse equation]
{Integrable full discretization of the multi-component short pulse equation}
\author{Ayako Hori$^1$, Yuta Tanaka$^1$, Ken-ichi Maruno$^2$, Yasuhiro Ohta $^3$, Bao-Feng Feng $^4$}
\address{$^1$~Department of Pure and Applied Mathematics, 
School of Fundamental Science and Engineering, Waseda University, 3-4-1 Okubo, Shinjuku-ku, Tokyo 169-8555, Japan}
\address{$^2$~Department of Applied Mathematics, Faculty of Science and Engineering, Waseda University, 
3-4-1 Okubo, Shinjuku-ku, Tokyo 169-8555, Japan}
\address{$^3$~ Department of Mathematics, Kobe University, Rokko, Kobe 657-8501, Japan}
\address{$^4$~ School of Mathematical and Statistical Sciences, The University of Texas Rio Grande Valley, 
1201 West University Dr., Edinburg, Texas 78541, USA}

\eads{ayako0903@akane.waseda.jp, tanayu@suou.waseda.jp, kmaruno@waseda.jp, ohta@math.kobe-u.ac.jp, baofeng.feng@utrgv.edu}
\vspace{10pt}
\begin{indented}
\item[]{\today}
\end{indented}
\begin{abstract}
We propose a new formulation of the multi-component short pulse (MCSP) equation
that includes the coupled complex short pulse (CCSP) equation as a reduction.
Using Hirota's bilinear method, we construct its $N$-soliton solutions in Pfaffian form.
We then derive integrable semi-discrete and fully discrete analogues of the MCSP equation 
admitting Pfaffian $N$-soliton solutions.
The resulting fully discrete system provides a practical self-adaptive moving mesh scheme 
for numerical simulations. For the parameter sets considered, 
numerical simulations demonstrate excellent agreement between the numerical and exact solutions, 
confirming the robustness and high accuracy of the proposed scheme.
\end{abstract}



\begin{section}{Introduction}
\label{sec_intro}

The theory of discrete integrable systems has developed substantially since the pioneering
work of Hirota and of Ablowitz and Ladik in the mid-1970s.
Hirota introduced the bilinear formalism, which enables soliton equations to be discretized
while preserving their integrable structure and solution properties. He applied this approach
to several equations, including the KdV, mKdV, sine-Gordon, and KP equations
\cite{Hirota1,Hirota2,Hirota3,Hirota4,Hirota5,Hirota_DAGTE}. Independently, Ablowitz and
Ladik developed a discretization method based on the AKNS-type linear eigenvalue problem and
derived integrable discretizations of several equations, including the nonlinear Schr\"{o}dinger 
(NLS) equation \cite{Abl1,Abl2,Abl3,Abl4,Abl5}. Detailed historical accounts of discrete
integrable systems can be found in the books by Suris \cite{Suris} and by
Hietarinta, Joshi, and Nijhoff \cite{DS}.

A class of soliton equations is associated with Wadati--Konno--Ichikawa (WKI)-type
linear eigenvalue problems \cite{WKI1,WKI2}. Ishimori, as well as Wadati and Sogo, showed that
WKI-type soliton equations can be transformed, through a hodograph transformation, into
equations associated with AKNS-type linear eigenvalue problems
\cite{ishimori1,ishimori2,wadachi1,rogers}. Although the integrable discretization of
WKI-type soliton equations had long been considered difficult, some of the present authors
overcame this difficulty by combining Hirota's bilinear method with a discrete hodograph
transformation \cite{maruno3,maruno4,maruno5}. The resulting discrete equations naturally
provide self-adaptive moving mesh schemes, whose mesh intervals adjust automatically in
regions where the solution varies rapidly. 
The same approach has subsequently been applied
to several other soliton equations that admit 
loop solitons (which are multivalued) and cusp solitons 
(whose derivatives blow up at the peak) 
\cite{maruno1,maruno8,maruno2,maruno6,maruno7}.

The short pulse (SP) equation
\begin{eqnarray}
u_{xt}=u+\frac{1}{6}(u^{3})_{xx},
\label{SP}
\end{eqnarray}
belongs to the WKI class and has also been discretized while preserving its
integrability by this approach. Sch\"{a}fer and Wayne proposed the SP equation as an approximate model for
the propagation of ultrashort optical pulses in nonlinear media \cite{S-W}. Although
the NLS equation derived from Maxwell's equations accurately describes slowly varying
envelopes in weakly nonlinear media, its accuracy deteriorates as the pulse width
decreases. In this regime, numerical studies show that the SP equation more
accurately approximates solutions of Maxwell's equations \cite{C-J-S-W}.

The SP equation has also attracted considerable attention in the field of integrable systems.
Sakovich and Sakovich constructed a Lax pair and thereby established its integrability
\cite{Sakovich1}. They further showed that the SP equation can be transformed into the
sine-Gordon equation, with loop and breather solutions of the SP equation corresponding to
kink and breather solutions of the sine-Gordon equation, respectively \cite{Sakovich2}.
Multi-loop-soliton and multi-breather solutions were subsequently obtained \cite{Matsuno1}.
Some of the present authors also proposed an integrable discretization of the SP equation
\cite{maruno3} and constructed a self-adaptive moving mesh scheme 
with nonzero boundary values \cite{hori2}. 
A self-adaptive moving mesh scheme has also been derived from a
differential-geometric perspective \cite{geometry}.

Multi-component and complex extensions of the SP equation have been studied extensively. Several multi-component 
SP equations and their $N$-soliton solutions have been proposed
\cite{Matsuno2, CoupleSP}, together with integrable semi-discretizations and self-adaptive moving mesh schemes \cite{semidisMCSP, semidisCoupleSP,hori2}. 
Similar results have been obtained for the complex short pulse (CSP) equation \cite{maruno5, complexSP}.

From a physical viewpoint, two-component or vector extensions of the short pulse equation arise when the effects of polarization and anisotropy 
on ultrashort optical pulses are taken into account \cite{Pietrzyk2008, Sakovich2008}.
In such models, the different components represent polarization components of the optical field, while the nonlinear coupling describes their mutual interaction. 
For birefringent optical fibers, a coupled complex short pulse
equation has also been derived to describe the propagation of two interacting polarization components \cite{complexSP}.

An integrable full discretization of the multi-component SP (MCSP) equation, however, has not yet been established. 
In particular, a practical self-adaptive moving mesh scheme for the fully discrete MCSP equation remains unavailable. 
In general, the construction of integrable fully discrete schemes for multi-component or coupled nonlinear wave equations is a challenging but highly important problem, as it requires a sophisticated treatment to preserve the underlying algebraic and geometric structures simultaneously in both space and time. An important example in this direction was provided by Hirota \cite{Hirota2000_mKdV}, who constructed an integrable fully discrete analogue of the coupled modified KdV
equation using the bilinear formalism and gauge-invariant discrete bilinear operators.
This example shows how auxiliary variables associated with the $\tau$-functions can
close fully discrete nonlinear difference equations while preserving integrability and
Pfaffian multi-soliton solutions.

In the case of the short pulse type equations, a previous full
discretization of the SP equation was constructed based on Wronskian solutions \cite{maruno3},
but the resulting scheme is too complicated for practical numerical implementation and is
difficult to extend to multi-component or complex-valued systems. In this paper, we instead
construct a full discretization from Pfaffian solutions, leading to a substantially more
compact scheme suitable for numerical computation.

To this end, we introduce a novel formulation of the MCSP equation that admits Pfaffian
multi-soliton solutions and includes the coupled complex short pulse (CCSP) equation
as a natural reduction:
\begin{eqnarray}
\left\{
\begin{array}{ll}
u_{xt}^{(\mu)}=u^{(\mu)}+\displaystyle\frac{1}{2}\left(\left(\displaystyle\sum_{\alpha \in I}\displaystyle\sum_{\beta \in J}c_{\alpha\beta}u^{(\alpha)}v^{(\beta)}\right)u_{x}^{(\mu)}\right)_{x},\\
v_{xt}^{(\nu)}=v^{(\nu)}+\displaystyle\frac{1}{2}\left(\left(\displaystyle\sum_{\alpha \in I}\displaystyle\sum_{\beta \in J}c_{\alpha\beta}u^{(\alpha)}v^{(\beta)}\right)v_{x}^{(\nu)}\right)_{x}, 
\end{array}
\right.\label{MCSP}
\end{eqnarray}
where $\mu\in I$, $\nu\in J$, $I=\{1,2,\cdots,K\}$, and
$J=\{1,2,\cdots,L\}$. The coefficients $c_{\alpha\beta}$ are constants independent of
$x$ and $t$, and $K$ and $L$ are positive integers. Under an appropriate complex
conjugacy reduction, this system yields the CCSP equation
\begin{eqnarray}
u_{xt}^{(\mu)}=u^{(\mu)}+\frac{1}{2}\left(\displaystyle\sum_{\alpha\in I}|u^{(\alpha)}|^{2}u_{x}^{(\mu)}\right)_{x}, \qquad \mu\in I, \quad I=\{1,2,\cdots, K\}.
\label{CCSP}
\end{eqnarray}
Specifically, (\ref{CCSP}) follows from (\ref{MCSP}) by imposing
$I=J=\{1,2,\cdots,K\}$, $c_{\mu\nu}=\delta_{\mu\nu}$ (where $\delta_{\mu\nu}$ is the Kronecker delta), and
$v^{(\mu)}=(u^{(\mu)})^*$, where $(u^{(\mu)})^*$ denotes the complex conjugate of
$u^{(\mu)}$.
Under this complex-conjugacy reduction, the nonlinear coupling term becomes the total intensity $\sum_{\alpha\in I}|u^{(\alpha)}|^2$.
Thus, in physically relevant reductions, the component index can be interpreted as labeling interacting polarization or modal components. 
The present MCSP system may therefore be viewed as a mathematical multi-component generalization of previously studied vector and coupled short pulse models.

The remainder of this paper is organized as follows. Section \ref{sec_bilinear} derives a bilinear form of the MCSP equation 
and constructs its $N$-soliton solutions in Pfaffian form. 
Section \ref{sec_semi} develops an integrable semi-discrete analogue and its corresponding $N$-soliton solutions. 
Section \ref{sec_full} derives an integrable fully discrete analogue. Section \ref{sec_nc} presents numerical 
simulations based on the resulting self-adaptive moving mesh scheme and compares the numerical results with exact solutions. 
Section \ref{sec_con} concludes the paper.

\end{section}

\begin{section}{Bilinear equations and \texorpdfstring{$N$}{N}-soliton solutions for the MCSP equation}
\label{sec_bilinear}

In this section, we derive a bilinear form and a hodograph transformation for
the MCSP equation (\ref{MCSP}) and construct its $N$-soliton solutions in Pfaffian
form.

First, every sufficiently smooth solution of the MCSP equation (\ref{MCSP}) satisfies the conservation law
\begin{eqnarray}
\left(\frac{1}{\rho}\right)_{t}-\left(\frac{F}{2\rho}\right)_{x}=0,\label{conservationlaw}
\end{eqnarray}
where
\begin{eqnarray}
\rho=\left(1+\displaystyle\sum_{\mu\in I}\displaystyle\sum_{\nu\in J}c_{\mu\nu}u^{(\mu)}_{x}v^{(\nu)}_{x}\right)^{-\frac{1}{2}},\qquad 
F=\displaystyle\sum_{\mu\in I}\displaystyle\sum_{\nu\in J}c_{\mu\nu}u^{(\mu)}v^{(\nu)}.\label{conservationlaw2}
\end{eqnarray}
Indeed, setting
$S=\sum_{\mu\in I}\sum_{\nu\in J}c_{\mu\nu}u_x^{(\mu)}v_x^{(\nu)}$
and $q=\sqrt{1+S}=1/\rho$, the MCSP equation (\ref{MCSP}) can be expanded as
\begin{eqnarray*}
u_{xt}^{(\mu)}&=&u^{(\mu)}+\frac{1}{2}F_xu_x^{(\mu)}
+\frac{1}{2}Fu_{xx}^{(\mu)},\\
v_{xt}^{(\nu)}&=&v^{(\nu)}+\frac{1}{2}F_xv_x^{(\nu)}
+\frac{1}{2}Fv_{xx}^{(\nu)}.
\end{eqnarray*}
Thus, differentiating $S$ with respect to $t$ and substituting these equations, we obtain
\begin{eqnarray*}
S_t&=&\sum_{\mu\in I}\sum_{\nu\in J}c_{\mu\nu}
\left(u_{xt}^{(\mu)}v_x^{(\nu)}+u_x^{(\mu)}v_{xt}^{(\nu)}\right)\\
&=&\sum_{\mu\in I}\sum_{\nu\in J}c_{\mu\nu}
\left(u^{(\mu)}v_x^{(\nu)}+u_x^{(\mu)}v^{(\nu)}\right)
+F_x\sum_{\mu\in I}\sum_{\nu\in J}c_{\mu\nu}u_x^{(\mu)}v_x^{(\nu)}
\nonumber\\
&&+\frac{F}{2}\sum_{\mu\in I}\sum_{\nu\in J}c_{\mu\nu}
\left(u_{xx}^{(\mu)}v_x^{(\nu)}+u_x^{(\mu)}v_{xx}^{(\nu)}\right)\\
&=&F_x+F_xS+\frac{F}{2}S_x
=F_x(1+S)+\frac{F}{2}S_x.
\end{eqnarray*}
Moreover, differentiating $q^2=1+S$ with respect to $t$ and
$x$ gives
\begin{eqnarray*}
2qq_t=S_t,\qquad 2qq_x=S_x.
\end{eqnarray*}
Using these relations together with $1+S=q^2$, we find
\begin{eqnarray*}
q_t&=&\frac{S_t}{2q}
=\frac{1}{2q}\left(F_xq^2+\frac{F}{2}(2qq_x)\right)\\
&=&\frac{1}{2}F_xq+\frac{1}{2}Fq_x
=\left(\frac{Fq}{2}\right)_x,
\end{eqnarray*}
which proves (\ref{conservationlaw}).
The conservation law (\ref{conservationlaw}) therefore defines the hodograph transformation
\begin{eqnarray}
dX=\frac{1}{\rho} dx+\frac{F}{2\rho}dt,\qquad dT=dt,\label{hodograph}
\end{eqnarray}
Equivalently, the inverse of the hodograph transformation (\ref{hodograph}) is
given in differential form by
\begin{eqnarray}
dx=\rho\,dX-\frac{F}{2}dT,\qquad dt=dT.
\label{inversehodographdifferential}
\end{eqnarray}
Hence, for a fixed reference point $X_{*}$, integration with respect to $X$
gives
\begin{eqnarray}
x(X,T)=x_{*}(T)+\int_{X_{*}}^{X}\rho(\bar{X},T)\,d\bar{X},
\qquad t=T,
\label{inversehodograph}
\end{eqnarray}
where $x_{*}(T)=x(X_{*},T)$ satisfies
\begin{eqnarray*}
\frac{d x_{*}}{dT}=-\frac{1}{2}F(X_{*},T).
\end{eqnarray*}
The compatibility condition for (\ref{inversehodographdifferential}) is
$\rho_T+F_X/2=0$, which is derived below in
(\ref{conservationlawXT}).
The hodograph transformation also gives the derivative relations
\begin{eqnarray}
\frac{\partial}{\partial X}=\rho \frac{\partial}{\partial x},\qquad\frac{\partial}{\partial T}
=\frac{\partial }{\partial t}-\frac{F}{2}\frac{\partial }{\partial x}.\label{derivativelaw}
\end{eqnarray}
The MCSP equation (\ref{MCSP}) can be written as
\begin{eqnarray}
\left\{
\begin{array}{ll}
\partial_{x}\left(\partial_{t}-\displaystyle\frac{F}{2}\partial_{x}\right)u^{(\mu)}=u^{(\mu)},\\
\partial_{x}\left(\partial_{t}-\displaystyle\frac{F}{2}\partial_{x}\right)v^{(\nu)}=v^{(\nu)}.\\
\end{array}
\right.
\label{eq_6}
\end{eqnarray}
Using (\ref{derivativelaw}), equation (\ref{eq_6}) becomes
\begin{eqnarray}
\left\{
\begin{array}{ll}
u_{XT}^{(\mu)}=u^{(\mu)}\rho,\\
v_{XT}^{(\nu)}=v^{(\nu)}\rho.
\end{array}
\right.
\label{CND}
\end{eqnarray}
Similarly, the conservation law (\ref{conservationlaw}) becomes
\begin{eqnarray}
\rho_{T}+\frac{1}{2}\left(\sum_{\mu\in I}\sum_{\nu\in J}c_{\mu\nu}u^{(\mu)}v^{(\nu)}\right)_{X}=0.\label{conservationlawXT}
 \end{eqnarray}
Thus, in the $(X,T)$ variables, $\rho$ is the conserved density and
$\frac{1}{2}\left(\sum_{\mu\in I}\sum_{\nu\in J}c_{\mu\nu}u^{(\mu)}v^{(\nu)}\right)$
is the associated flux.

Consequently, the hodograph transformation (\ref{hodograph}) maps the MCSP equation
(\ref{MCSP}) to the multi-component coupled integrable dispersionless (MCCID) system
\begin{eqnarray}
\left\{
\begin{array}{lll}
\rho_{T}+\displaystyle\frac{1}{2}\left(\displaystyle\sum_{\mu\in I}\displaystyle\sum_{\nu\in J}c_{\mu\nu}u^{(\mu)}v^{(\nu)}\right)_{X}=0,\\
u_{XT}^{(\mu)}=u^{(\mu)}\rho,\\
v_{XT}^{(\nu)}=v^{(\nu)}\rho.
\end{array}
\right.
\label{dispersionlessMCSP}
\end{eqnarray}

Next, we introduce the dependent-variable transformation
\begin{eqnarray}
\fl{u^{(\mu)}=\frac{g^{(\mu)}}{f}\quad(\mu\in I),\qquad v^{(\nu)}=\frac{h^{(\nu)}}{f}\quad(\nu\in J),\qquad\rho=1-2(\log{f})_{XT}.}\label{dependenttransformation}
\end{eqnarray}
Under this transformation, the MCCID system (\ref{dispersionlessMCSP}) is
equivalent to the following bilinear system:
\begin{eqnarray}
\left\{
\begin{array}{ll}
D_{T}^{2}f\cdot f=\displaystyle\frac{1}{2}\displaystyle\sum_{\mu\in I}\displaystyle\sum_{\nu\in J}c_{\mu\nu}g^{(\mu)}h^{(\nu)},\\
D_{X}D_{T}f\cdot g^{(\mu)}=fg^{(\mu)},\\
D_{X}D_{T}f\cdot h^{(\nu)}=fh^{(\nu)}.
\end{array}
\right.
\label{bilinearMCSP}
\end{eqnarray}
Here, the Hirota $D$-operators are defined by
\begin{eqnarray*}
D_X^m D_T^n a \cdot b = \left. \left(\frac{\partial}{\partial X}
 - \frac{\partial}{\partial X'}\right)^m \left(\frac{\partial}{\partial T} - \frac{\partial}{\partial T'}\right)^n a(X,T)b(X',T') \right|_{X=X', T=T'}.
\end{eqnarray*}

For completeness, the equivalence above includes a gauge choice. Substitution of
$\rho=1-2(\log f)_{XT}$ into the first equation of
(\ref{dispersionlessMCSP}) gives
\begin{eqnarray*}
\partial_X\left[-2(\log f)_{TT}+\frac{F}{2}\right]=0.
\end{eqnarray*}
Hence the expression in brackets may initially be an arbitrary function $A(T)$.
The common gauge transformation
$(f,g^{(\mu)},h^{(\nu)})\mapsto
e^{\psi(T)}(f,g^{(\mu)},h^{(\nu)})$ leaves
$u^{(\mu)}$, $v^{(\nu)}$, and $\rho$ unchanged and preserves the form of the
last two equations in (\ref{bilinearMCSP}). Choosing $\psi_{TT}=A(T)/2$
sets $A(T)=0$, which gives the first equation in (\ref{bilinearMCSP}).

We now express the inverse hodograph transformation in terms of the
$\tau$-function. Using $\rho=1-2(\log f)_{XT}$ in
(\ref{inversehodograph}), we obtain
\begin{eqnarray*}
x(X,T)&=&x_{*}(T)+(X-X_{*})
-2\left[(\log f)_{T}(X,T)-(\log f)_{T}(X_{*},T)\right]\\
&=&X-2(\log f)_{T}(X,T)+C(T),
\end{eqnarray*}
where
$C(T)=x_{*}(T)-X_{*}+2(\log f)_{T}(X_{*},T)$.
The first equation in (\ref{bilinearMCSP}) implies
$(\log f)_{TT}=F/4$, and therefore
\begin{eqnarray*}
C_{T}=-\frac{1}{2}F(X_{*},T)
+2(\log f)_{TT}(X_{*},T)=0.
\end{eqnarray*}
Thus $C(T)$ is constant. Since this constant corresponds to a translation of
the spatial coordinate $x$, we set $C(T)=0$ without loss of generality. The
corresponding solution of the original MCSP equation is then given
parametrically by
\begin{eqnarray}
\fl{u^{(\mu)}=\frac{g^{(\mu)}}{f},\qquad
v^{(\nu)}=\frac{h^{(\nu)}}{f},\qquad
x=X-2(\log f)_{T},\qquad t=T.}
\label{parametricMCSP}
\end{eqnarray}

We next introduce the Pfaffian elements needed to establish a Pfaffian identity used below. 
Let $M$ and $N$ be integers such that $1\leq N<M$. We use the standard antisymmetry convention
${\rm Pf}(\alpha,\beta)=-{\rm Pf}(\beta,\alpha)$ and
${\rm Pf}(\alpha,\alpha)=0$; hence it is sufficient to specify one orientation
of each element. The required Pfaffian elements are defined as follows:
\begin{eqnarray}
\fl{{\rm Pf}(j,k): {\rm arbitrary}, \quad(1\leq j<k\leq M),}
\label{pfMCSP00}
\end{eqnarray}
\begin{eqnarray}
\fl{{\rm Pf}(d_{0},j): {\rm arbitrary}, \quad(1\leq j\leq M),}
\end{eqnarray}
\begin{eqnarray}
\fl{{\rm Pf}(j,k^{\prime})=\delta_{jk},\quad (1\leq j\leq M,\quad 1\leq k\leq M)}
\label{pfMCSP10}
\end{eqnarray}
\begin{eqnarray}
\fl{{\rm Pf}(k^{\prime},j^{\prime})=
\left\{
\begin{array}{ll}
\displaystyle\frac{1}{P_{j}-P_{k}}\sum_{\mu\in I}\sum_{\nu\in J} c_{\mu\nu}a_{j}^{(\mu)}b_{k}^{(\nu)},&
(1\leq j\leq N<k\leq M),\\
0,&{\rm otherwise},
\end{array}
\right.}
\end{eqnarray}
\begin{eqnarray}
\fl{{\rm Pf}(a^{(\mu)}, j^{\prime})=
\left\{
\begin{array}{ll}
a_{j}^{(\mu)},&(1\leq j\leq N),\\
0,&(N+1\leq j\leq M),
\end{array}
\right.}
\qquad\mu\in I,
\label{pfMCSP20}
\end{eqnarray}
\begin{eqnarray}
\fl{{\rm Pf}(b^{(\nu)}, k^{\prime})=
\left\{
\begin{array}{ll}
0,&(1\leq k\leq N),\\
b_{k}^{(\nu)},& (N+1\leq k\leq M),
\end{array}
\right.}
\qquad\nu\in J,
\label{pfMCSP21}
\end{eqnarray}
where $P_{j}, c_{\mu\nu}, a_{j}^{(\mu)}$ and $b_{k}^{(\nu)}$ are arbitrary constants,
subject to the condition that all denominators appearing above are nonzero.
All remaining Pfaffian elements are set to zero. 

\begin{lem}
\label{lem2.1}
For the Pfaffian elements defined in (\ref{pfMCSP00})--(\ref{pfMCSP21}), the following identity holds:
\begin{eqnarray}
\fl{\sum_{\mu\in I}\sum_{\nu\in J}c_{\mu\nu}{\rm Pf}(a^{(\mu)}, d_{0},1,\cdots, M, M^{\prime}, 
\cdots, 2^{\prime}, 1^{\prime}){\rm Pf}(b^{(\nu)}, d_{0},1,\cdots, M, M^{\prime}, \cdots, 2^{\prime}, 1^{\prime})}\nonumber\\
\fl{=\sum_{j=1}^{M}P_{j}{\rm Pf}(d_{0},1,\cdots, \hat{j},\cdots, M, M^{\prime}, 
\cdots, 2^{\prime},1^{\prime})}\nonumber\\
\fl{\qquad\times{\rm Pf}(d_{0},1,\cdots, M, M^{\prime}, \cdots, \hat{j^{\prime}}, \cdots , 2^{\prime},1^{\prime}).}
\end{eqnarray}
Here the hat $\,\hat{}\,$ indicates that the corresponding index is omitted.
\end{lem}

\begin{proof}
Since
\begin{eqnarray} 
\fl{{\rm Pf}(a^{(\mu)},d_{0},1,\cdots, M,M^{\prime},\cdots, 2^{\prime}, 1^{\prime})}\nonumber\\
\fl{\qquad=\sum_{j=1}^{N}(-1)^{j-1}{\rm Pf}(a^{(\mu)},j^{\prime}){\rm Pf}(d_{0},1,\cdots, M, M^{\prime}, \cdots, \hat{j^{\prime}}, \cdots, 2^{\prime}, 1^{\prime}),}\nonumber
\end{eqnarray}
\begin{eqnarray} 
\fl{{\rm Pf}(b^{(\nu)},d_{0},1,\cdots, M,M^{\prime},\cdots, 2^{\prime}, 1^{\prime})}\nonumber\\
\fl{\qquad=\sum_{k=N+1}^{M}(-1)^{k-1}{\rm Pf}(b^{(\nu)},k^{\prime}){\rm Pf}(d_{0},1,\cdots, M, M^{\prime}, \cdots, \hat{k^{\prime}}, \cdots, 2^{\prime}, 1^{\prime}),}\nonumber
\end{eqnarray}
we obtain
\begin{eqnarray}
\fl{\sum_{\mu\in I}\sum_{\nu\in J}c_{\mu\nu}{\rm Pf}(a^{(\mu)}, d_{0},1,\cdots, M, M^{\prime}, \cdots, 2^{\prime}, 1^{\prime})}\nonumber\\
\fl{\qquad\times{\rm Pf}(b^{(\nu)},d_{0},1,\cdots, M, M^{\prime}, \cdots, 2^{\prime}, 1^{\prime})}\nonumber\\
\fl{=\sum_{j=1}^{N}\sum_{k=N+1}^{M}(-1)^{j+k}\sum_{\mu\in I}\sum_{\nu\in J}c_{\mu\nu}a_{j}^{(\mu)}b_{k}^{(\nu)}{\rm Pf}(d_{0},1,\cdots, M, M^{\prime}, \cdots, \hat{j^{\prime}}, \cdots,  2^{\prime}, 1^{\prime})}\nonumber\\
\fl{\qquad\times{\rm Pf}(d_{0},1,\cdots, M, M^{\prime}, \cdots,  \hat{k^{\prime}}, \cdots, 2^{\prime}, 1^{\prime})}\nonumber\\
\fl{=\sum_{j=1}^{N}\sum_{k=N+1}^{M}(-1)^{j+k}(P_{j}-P_{k}){\rm Pf}(k^{\prime},j^{\prime}){\rm Pf}(d_{0},1,\cdots, M, M^{\prime}, \cdots, \hat{j^{\prime}}, \cdots,  2^{\prime}, 1^{\prime})}\nonumber\\
\fl{\qquad\times{\rm Pf}(d_{0},1,\cdots, M, M^{\prime}, \cdots,  \hat{k^{\prime}}, \cdots, 2^{\prime}, 1^{\prime})}\nonumber\\
\fl{=\sum_{j=1}^{N}P_{j}{\rm Pf}(d_{0},1,\cdots, M, M^{\prime}, \cdots, \hat{j^{\prime}}, \cdots,  2^{\prime}, 1^{\prime})}\nonumber\\
\fl{\qquad\times\sum_{k=N+1}^{M}(-1)^{j+k}{\rm Pf}(k^{\prime},j^{\prime}){\rm Pf}(d_{0},1,\cdots, M, M^{\prime}, \cdots, \hat{k^{\prime}}, \cdots,  2^{\prime}, 1^{\prime})}\nonumber\\
\fl{-\sum_{k=N+1}^{M}P_{k}{\rm Pf}(d_{0},1,\cdots, M, M^{\prime}, \cdots, \hat{k^{\prime}}, \cdots,  2^{\prime}, 1^{\prime})}\nonumber\\
\fl{\qquad\times\sum_{j=1}^{N}(-1)^{j+k}{\rm Pf}(k^{\prime},j^{\prime}){\rm Pf}(d_{0},1,\cdots, M, M^{\prime}, \cdots, \hat{j^{\prime}}, \cdots,  2^{\prime}, 1^{\prime}).}
\label{eq_2.16}
 \end{eqnarray}
Now, by expanding the trivially vanishing Pfaffians ${\rm Pf}(d_{0},1,\cdots ,M,M^{\prime},\cdots, j+1^{\prime}, j^{\prime},j^{\prime},j-1^{\prime},\cdots,1^{\prime})$ and ${\rm Pf}(d_{0},1,\cdots ,M,M^{\prime},\cdots, k+1^{\prime}, k^{\prime},k^{\prime},k-1^{\prime},\cdots,1^{\prime})$, we obtain 
\begin{eqnarray}
\fl{\sum_{k=N+1}^{M}(-1)^{j+k}{\rm Pf}(k^{\prime},j^{\prime}){\rm Pf}(d_{0},1,\cdots, M, M^{\prime}, \cdots, \hat{k^{\prime}}, \cdots,  2^{\prime}, 1^{\prime})}\nonumber\\
\fl{\qquad={\rm Pf}(d_{0},1,\cdots, \hat{j}, \cdots, M, M^{\prime}, \cdots, 1^{\prime}),\qquad (1\leq j \leq N),}\nonumber
\end{eqnarray}
\begin{eqnarray}
\fl{\sum_{j=1}^{N}(-1)^{j+k}{\rm Pf}(k^{\prime},j^{\prime}){\rm Pf}(d_{0},1,\cdots, M, M^{\prime}, \cdots, \hat{j^{\prime}}, \cdots,  2^{\prime}, 1^{\prime})}\nonumber\\
\fl{\qquad=-{\rm Pf}(d_{0},1,\cdots, \hat{k}, \cdots, M, M^{\prime}, \cdots, 1^{\prime}),\qquad (N+1 \leq k \leq M).}\nonumber
\end{eqnarray}
Therefore, it follows from (\ref{eq_2.16}) that
\begin{eqnarray}
\fl{\sum_{\mu\in I}\sum_{\nu\in J}c_{\mu\nu}{\rm Pf}(a^{(\mu)}, d_{0},1,\cdots, M, M^{\prime}, \cdots, 2^{\prime}, 1^{\prime}){\rm Pf}(b^{(\nu)}, d_{0},1,\cdots, M, M^{\prime}, \cdots, 2^{\prime}, 1^{\prime})}\nonumber\\
\fl{=\sum_{j=1}^{N}P_{j}{\rm Pf}(d_{0},1,\cdots, M, M^{\prime}, \cdots, \hat{j^{\prime}}, \cdots, 2^{\prime}, 1^{\prime}){\rm Pf}(d_{0},1,\cdots, \hat{j}, \cdots, M, M^{\prime}, \cdots, 1^{\prime})}\nonumber\\
\fl{\quad +\sum_{k=N+1}^{M}P_{k}{\rm Pf}(d_{0},1,\cdots, M, M^{\prime}, \cdots, \hat{k^{\prime}}, \cdots, 2^{\prime}, 1^{\prime}){\rm Pf}(d_{0},1,\cdots, \hat{k}, \cdots, M, M^{\prime}, \cdots, 1^{\prime})}\nonumber\\
\fl{=\sum_{j=1}^{M}P_{j}{\rm Pf}(d_{0},1,\cdots, \hat{j}, \cdots, M, M^{\prime}, \cdots, 2^{\prime}, 1^{\prime})}\nonumber\\
\fl{\qquad\times{\rm Pf}(d_{0},1,\cdots, M, M^{\prime}, \cdots, \hat{j^{\prime}}, \cdots, 2^{\prime}, 1^{\prime}).}
\end{eqnarray}
\end{proof}
In what follows, we abbreviate ${\rm Pf}(1, 2,\cdots, M-1, M, M^{\prime}, M-1^{\prime},\cdots, 2^{\prime}, 1^{\prime})$ as ${\rm Pf}(\cdots)$, and ${\rm Pf}(\alpha, \beta, 1, 2,\cdots, M-1, M, M^{\prime}, M-1^{\prime},\cdots, 2^{\prime}, 1^{\prime})$ as ${\rm Pf}(\alpha,\beta,\cdots)$, where $\alpha$ and $\beta$ are arbitrary
Pfaffian labels. The same notation is used throughout the semi-discrete and fully discrete formulations.

To construct Pfaffian solutions of the bilinear system (\ref{bilinearMCSP}), let $M=2N$
and introduce the following Pfaffian elements: 
\begin{eqnarray}
\fl{{\rm Pf}(j,k)=\frac{p_{j}-p_{k}}{p_{j}+p_{k}}\varphi_{j}\varphi_{k},\quad (1\leq j<k\leq M),\label{pfMCSP1}}
\end{eqnarray}
\begin{eqnarray}
\fl{{\rm Pf}(j,k^{\prime})=\delta_{jk},\quad (1\leq j\leq M,\quad 1\leq k\leq M),}
\end{eqnarray}
\begin{eqnarray}
\fl{{\rm Pf}(k^{\prime},j^{\prime})=
\left\{
\begin{array}{ll}
\displaystyle\frac{p_{j}^{2}p_{k}^{2}}{4(p_{j}^{2}-p_{k}^{2})}\sum_{\mu\in I}\sum_{\nu\in J} c_{\mu\nu}a_{j}^{(\mu)}b_{k}^{(\nu)},&
(1\leq j\leq N<k\leq M),\\
0,&{\rm otherwise},
\end{array}
\right.}
\end{eqnarray}
\begin{eqnarray}
\fl{{\rm Pf}(a^{(\mu)}, j^{\prime})=
\left\{
\begin{array}{ll}
a_{j}^{(\mu)},&(1\leq j\leq N),\\
0,&(N+1\leq j\leq M),
\end{array}
\right.}
\qquad\mu\in I,
\label{pfMCSP22}
\end{eqnarray}
\begin{eqnarray}
\fl{{\rm Pf}(b^{(\nu)}, k^{\prime})=
\left\{
\begin{array}{ll}
0,&(1\leq k\leq N),\\
b_{k}^{(\nu)},& (N+1\leq k\leq M),
\end{array}
\right.}
\qquad\nu\in J,
\label{pfMCSP23}
\end{eqnarray}
where $\varphi_{j}=B_{j}e^{p_{j}X+p_{j}^{-1}T}$, and $B_j$ and $p_j$ are constants determining the phase and wave number of the $j$-th soliton, respectively.
The constants $a_{j}^{(\mu)}$ and $b_{k}^{(\nu)}$ are arbitrary, and the parameters are
chosen so that all expressions above are well-defined.
The element involving $d_0$ is defined by
\begin{eqnarray}
{\rm Pf}(d_0,i)=\varphi_i,\qquad 1\leq i\leq M.
\label{pfMCSP24}
\end{eqnarray}
All other Pfaffian elements formed from the labels appearing in the solution
formulas above are set to zero.

\begin{thm}
\label{thm2.2}
For the Pfaffian elements defined in (\ref{pfMCSP1})--(\ref{pfMCSP24}), the functions
\begin{eqnarray}
f={\rm Pf}(1,2,\cdots, M-1, M, M^{\prime}, M-1^{\prime}, \cdots, 2^{\prime}, 1^{\prime}),\label{con_solb}\\
g^{(\mu)}={\rm Pf}(a^{(\mu)}, d_{0},1,2,\cdots, M-1, M, M^{\prime}, M-1^{\prime}, \cdots, 2^{\prime}, 1^{\prime}),\\
h^{(\nu)}={\rm Pf}(b^{(\nu)}, d_{0},1,2,\cdots, M-1, M, M^{\prime}, M-1^{\prime}, \cdots, 2^{\prime}, 1^{\prime}),
\label{con_sole}
\end{eqnarray}
with $\mu\in I$ and $\nu\in J$, satisfy the bilinear system (\ref{bilinearMCSP}). 
\end{thm}


\begin{proof}
We first verify the first equation of the bilinear system (\ref{bilinearMCSP}).
For this purpose, define the additional Pfaffian elements
\begin{eqnarray}
{\rm Pf}(d_{n},i)=p_{i}^{n}\varphi_{i},
\end{eqnarray}
where $1 \leq i \leq M$ and $n=-2,-1,1$. We set
${\rm Pf}(d_n,d_r)=0$ for $n,r\in\{-2,-1,0,1\}$ and set all other
Pfaffian elements involving the newly introduced labels to zero. Since
\begin{eqnarray}
\frac{\partial}{\partial X}{\rm Pf}(j,k)=(p_{j}-p_{k})\varphi_{j}\varphi_{k}={\rm Pf}(d_{0},d_{1},j,k),\nonumber\\
\frac{\partial}{\partial T}{\rm Pf}(j,k)=(p_{k}^{-1}-p_{j}^{-1})\varphi_{j}\varphi_{k}={\rm Pf}(d_{-1},d_{0},j,k),\nonumber\\
\frac{\partial^{2}}{\partial T^{2}}{\rm Pf}(j,k)=(p_{k}^{-2}-p_{j}^{-2})\varphi_{j}\varphi_{k}={\rm Pf}(d_{-2},d_{0},j,k),\nonumber\\
\frac{\partial^{2}}{\partial X \partial T}{\rm Pf}(j,k)=(p_{j}p_{k}^{-1}-p_{k}p_{j}^{-1})\varphi_{j}\varphi_{k}={\rm Pf}(d_{-1},d_{1},j,k),\nonumber
\end{eqnarray}
we have 
\begin{eqnarray}
\frac{\partial f}{\partial X}={\rm Pf}(d_{0},d_{1},\cdots),\qquad
\frac{\partial f}{\partial T}={\rm Pf}(d_{-1},d_{0},\cdots),\nonumber\\
\frac{\partial^{2} f}{\partial T^{2}}={\rm Pf}(d_{-2},d_{0},\cdots),\qquad
\frac{\partial^{2} f}{\partial X\partial T}={\rm Pf}(d_{-1},d_{1},\cdots).
\end{eqnarray}
Applying Lemma \ref{lem2.1} with $P_{j}=-4p_{j}^{-2}$, we obtain
\begin{eqnarray}
\fl{\sum_{\mu\in I}\sum_{\nu\in J}c_{\mu\nu}{\rm Pf}(a^{(\mu)}, d_{0},1,\cdots, M, M^{\prime}, \cdots, 2^{\prime}, 1^{\prime})}\nonumber\\
\fl{\qquad\times{\rm Pf}(b^{(\nu)}, d_{0},1,\cdots, M, M^{\prime}, \cdots, 2^{\prime}, 1^{\prime})}\nonumber\\
\fl{=-\sum_{j=1}^{M}4p_{j}^{-2}{\rm Pf}(d_{0},1,\cdots, \hat{j},\cdots, M, M^{\prime}, \cdots, 2^{\prime},1^{\prime})}\nonumber\\
\fl{\qquad\times{\rm Pf}(d_{0},1,\cdots, M, M^{\prime}, \cdots, \hat{j^{\prime}}, \cdots, 2^{\prime},1^{\prime}).}\label{eq_thm1}
\end{eqnarray}
Since ${\rm Pf}(d_0,d_0,\cdots)=0$, we have
\begin{eqnarray}
\fl{\frac{\partial^{2}f}{\partial T^{2}}\times 0-\left(\frac{\partial f}{\partial T}\right)^2}\nonumber\\
\fl{\qquad={\rm Pf}(d_{-2},d_{0},\cdots){\rm Pf}(d_{0},d_{0},\cdots)-{\rm Pf}(d_{-1},d_{0},\cdots){\rm Pf}(d_{-1},d_{0},\cdots)}\nonumber\\
\fl{=\sum_{j=1}^{M}(-1)^{j}{\rm Pf}(d_{-2},j){\rm Pf}(d_{0},1,\cdots,\hat{j},\cdots, M, M^{\prime}, \cdots 1^{\prime})}\nonumber\\
\fl{\qquad\times\sum_{k=1}^{M}(-1)^{k}{\rm Pf}(d_{0},k){\rm Pf}(d_{0},1,\cdots,\hat{k},\cdots M, M^{\prime}, \cdots, 1^{\prime})}\nonumber\\
\fl{-\sum_{j=1}^{M}(-1)^{j}{\rm Pf}(d_{-1},j){\rm Pf}(d_{0},1,\cdots,\hat{j},\cdots, M, M^{\prime}, \cdots 1^{\prime})}\nonumber\\
\fl{\qquad\times\sum_{k=1}^{M}(-1)^{k}{\rm Pf}(d_{-1},k){\rm Pf}(d_{0},1,\cdots,\hat{k},\cdots M, M^{\prime}, \cdots, 1^{\prime})}\nonumber\\
\fl{=\sum_{j,k=1}^{M}(-1)^{j+k}(p_{j}^{-2}-p_{j}^{-1}p_{k}^{-1})\varphi_{j}\varphi_{k}}\nonumber\\
\fl{\qquad\times{\rm Pf}(d_{0},1,\cdots,\hat{j},\cdots, M, M^{\prime}, \cdots 1^{\prime}){\rm Pf}(d_{0},1,\cdots,\hat{k},\cdots, M, M^{\prime}, \cdots 1^{\prime})}\nonumber\\
\fl{=\sum_{j,k=1}^{M}(-1)^{j+k+1}(p_{j}^{-2}+p_{j}^{-1}p_{k}^{-1}){\rm Pf}(j,k)}\nonumber\\
\fl{\qquad\times{\rm Pf}(d_{0}, 1,\cdots,\hat{j},\cdots, M, M^{\prime}, \cdots, 1^{\prime}){\rm Pf}(d_{0}, 1,\cdots,\hat{k},\cdots, M, M^{\prime}, \cdots, 1^{\prime}).}\label{eq_A55}
\end{eqnarray}
Since
\begin{eqnarray}
\fl{\quad\sum_{j,k=1}^{M}(-1)^{j+k+1}p_{j}^{-1}p_{k}^{-1}{\rm Pf}(j,k){\rm Pf}(d_{0}, 1,\cdots,\hat{j},\cdots, M, M^{\prime}, \cdots 1^{\prime})}\nonumber\\
\fl{\qquad\times{\rm Pf}(d_{0},1,\cdots,\hat{k},\cdots, M, M^{\prime}, \cdots 1^{\prime})}\nonumber\\
\fl{=\sum_{k,j=1}^{M}(-1)^{k+j+1}p_{k}^{-1}p_{j}^{-1}{\rm Pf}(k,j){\rm Pf}(d_{0},1,\cdots,\hat{k},\cdots, M, M^{\prime}, \cdots 1^{\prime})}\nonumber\\
\fl{\qquad\times{\rm Pf}(d_{0}, 1,\cdots,\hat{j},\cdots, M, M^{\prime}, \cdots 1^{\prime})}\nonumber\\
\fl{=-\sum_{j,k=1}^{M}(-1)^{j+k+1}p_{j}^{-1}p_{k}^{-1}{\rm Pf}(j,k){\rm Pf}(d_{0}, 1,\cdots,\hat{j},\cdots, M, M^{\prime}, \cdots 1^{\prime})}\nonumber\\
\fl{\qquad\times{\rm Pf}(d_{0}, 1,\cdots,\hat{k},\cdots, M, M^{\prime}, \cdots 1^{\prime}),} \nonumber
\end{eqnarray}
we obtain
\begin{eqnarray}
\fl{\sum_{j,k=1}^{M}(-1)^{j+k+1}p_{j}^{-1}p_{k}^{-1}{\rm Pf}(j,k){\rm Pf}(d_{0}, 1,\cdots,\hat{j},\cdots, M, M^{\prime}, \cdots 1^{\prime})}\nonumber\\
\fl{\qquad\times{\rm Pf}(d_{0}, 1,\cdots,\hat{k},\cdots, M, M^{\prime}, \cdots 1^{\prime})=0.}\label{eq_A66}
\end{eqnarray}
Substituting (\ref{eq_A66}) into (\ref{eq_A55}), it follows that
\begin{eqnarray}
\fl{-\frac{\partial f}{\partial T}\frac{\partial f}{\partial T}}\nonumber\\
\fl{\quad =\sum_{j,k=1}^{M}(-1)^{j+k+1}p_{j}^{-2}{\rm Pf}(j,k){\rm Pf}(d_{0}, 1, \cdots,\hat{j},\cdots, M, M^{\prime}, \cdots 1^{\prime})}\nonumber\\
\fl{\qquad\times{\rm Pf}(d_{0}, 1,\cdots,\hat{k},\cdots, M, M^{\prime}, \cdots 1^{\prime})}\nonumber\\
\fl{\quad =\sum_{j=1}^{M}(-1)^{j+1}p_{j}^{-2}{\rm Pf}(d_{0}, 1,\cdots,\hat{j},\cdots, M, M^{\prime}, \cdots 1^{\prime})}\nonumber\\
\fl{\qquad\times\left[\sum_{k=1}^{M}(-1)^{k}{\rm Pf}(j,k){\rm Pf}(d_{0}, 1, \cdots,\hat{k},\cdots, M, M^{\prime}, \cdots 1^{\prime})\right].}\label{eq_A77}
\end{eqnarray}
Now, by expanding the trivially vanishing Pfaffian ${\rm Pf}(j,d_{0},1,\cdots, j, \cdots, M, M^{\prime}, \cdots, 1^{\prime})$, we obtain
\begin{eqnarray}
\fl{\sum_{k=1}^{M}(-1)^{k}{\rm Pf}(j,k){\rm Pf}(d_{0}, 1,\cdots,\hat{k},\cdots M, M^{\prime},\cdots 1^{\prime})}\nonumber\\
\fl{\quad ={\rm Pf}(d_{0},j){\rm Pf}(1, \cdots, M, M^{\prime}, \cdots, 1^{\prime})
+(-1)^{j}{\rm Pf}(d_{0}, 1,\cdots,M,M^{\prime},\cdots,\hat{j^{\prime}},\cdots ,1^{\prime}).}
\label{eq_A88}
\end{eqnarray}
Substituting (\ref{eq_A88}) into (\ref{eq_A77}), we have
\begin{eqnarray}
\fl{-\frac{\partial f}{\partial T}\frac{\partial f}{\partial T}}\nonumber\\
\fl{\quad =\sum_{j=1}^{M}(-1)^{j+1}p_{j}^{-2}{\rm Pf}(d_{0}, 1,\cdots,\hat{j},\cdots, M, M^{\prime}, \cdots 1^{\prime})}\nonumber\\
\fl{\qquad\times\left[{\rm Pf}(d_{0},j){\rm Pf}(1, \cdots, M, M^{\prime}, \cdots, 1^{\prime})
+(-1)^{j}{\rm Pf}(d_{0}, 1,\cdots,M,M^{\prime},\cdots,\hat{j^{\prime}},\cdots ,1^{\prime})\right]}\nonumber\\
\fl{\quad =\sum_{j=1}^{M}(-1)^{j+1}p_{j}^{-2}{\rm Pf}(d_{0},j){\rm Pf}(d_{0}, 1, \cdots,\hat{j},\cdots, M, M^{\prime}, \cdots 1^{\prime})}\nonumber\\
\fl{\qquad\times{\rm Pf}(1, \cdots, M, M^{\prime}, \cdots, 1^{\prime})}\nonumber\\
\fl{\qquad -\sum_{j=1}^{M}p_{j}^{-2}{\rm Pf}(d_{0}, 1,\cdots,\hat{j},\cdots, M, M^{\prime}, \cdots 1^{\prime}){\rm Pf}(d_{0}, 1,\cdots,M,M^{\prime},\cdots,\hat{j^{\prime}},\cdots ,1^{\prime})}\nonumber\\
\fl{\quad =-\sum_{j=1}^{M}(-1)^{j}{\rm Pf}(d_{-2},j){\rm Pf}(d_{0}, 1,\cdots,\hat{j},\cdots, M, M^{\prime}, \cdots 1^{\prime})}\nonumber\\
\fl{\qquad\times{\rm Pf}(1, \cdots, M, M^{\prime}, \cdots, 1^{\prime})}\nonumber\\
\fl{\qquad -\sum_{j=1}^{M}p_{j}^{-2}{\rm Pf}(d_{0}, 1,\cdots,\hat{j},\cdots, M, M^{\prime}, \cdots 1^{\prime})
{\rm Pf}(d_{0}, 1,\cdots,M,M^{\prime},\cdots,\hat{j^{\prime}},\cdots ,1^{\prime})}\nonumber\\
\fl{\quad =-{\rm Pf}(d_{-2},d_{0},\cdots){\rm Pf}(1, \cdots, M, M^{\prime}, \cdots, 1^{\prime})}\nonumber\\
\fl{\qquad -\sum_{j=1}^{M}p_{j}^{-2}{\rm Pf}(d_{0},1, \cdots,\hat{j},\cdots, M, M^{\prime}, \cdots 1^{\prime})
{\rm Pf}(d_{0}, 1,\cdots,M,M^{\prime},\cdots,\hat{j^{\prime}},\cdots ,1^{\prime})}\nonumber\\
\fl{\quad =-\frac{\partial^{2} f}{\partial T^{2}}f+\frac{1}{4}\displaystyle\sum_{\mu\in I}\sum_{\nu\in J}c_{\mu\nu}{\rm Pf}(a^{(\mu)},d_{0},\cdots){\rm Pf}(b^{(\nu)},d_{0},\cdots)}\nonumber\\
\fl{\quad =-\frac{\partial^{2} f}{\partial T^{2}}f+\frac{1}{4}\displaystyle\sum_{\mu\in I}\sum_{\nu\in J}c_{\mu\nu}g^{(\mu)}h^{(\nu)},}\nonumber
\end{eqnarray}
where (\ref{eq_thm1}) was used in the last step. Hence,
\begin{eqnarray}
\frac{\partial^{2}f}{\partial T^{2}}f-\frac{\partial f}{\partial T}\frac{\partial f}{\partial T}=\frac{1}{4}\displaystyle\sum_{\mu\in I}\sum_{\nu\in J}c_{\mu\nu}g^{(\mu)}h^{(\nu)},\nonumber
\end{eqnarray}
which is the first equation of (\ref{bilinearMCSP}).

We next verify the second equation of (\ref{bilinearMCSP}).
\begin{eqnarray}
\fl{\displaystyle\frac{\partial g^{(\mu)}}{\partial X}=\partial_{X}{\rm Pf}(a^{(\mu)},d_{0},\cdots)}\nonumber\\
\fl{\quad =-\partial_{X}\left(\sum_{j=1}^{M}(-1)^{j}{\rm Pf}(d_{0},j){\rm Pf}(a^{(\mu)},1,\cdots,\hat{j},\cdots, M, M^{\prime}, \cdots, 1^{\prime})\right)}\nonumber\\
\fl{\quad =-\sum_{j=1}^{M}(-1)^{j}\partial_{X}{\rm Pf}(d_{0},j){\rm Pf}(a^{(\mu)}, 1,\cdots,\hat{j},\cdots, M, M^{\prime}, \cdots, 1^{\prime})}\nonumber\\
-\sum_{j=1}^{M}(-1)^{j}{\rm Pf}(d_{0},j)\partial_{X}{\rm Pf}(a^{(\mu)},1,\cdots,\hat{j},\cdots, M, M^{\prime}, \cdots, 1^{\prime})\nonumber\\
\fl{\quad =-\sum_{j=1}^{M}(-1)^{j}{\rm Pf}(d_{1},j){\rm Pf}(a^{(\mu)}, 1,\cdots,\hat{j},\cdots, M, M^{\prime}, \cdots, 1^{\prime})}\nonumber\\
\fl{\qquad-\sum_{j=1}^{M}(-1)^{j}{\rm Pf}(d_{0},j){\rm Pf}(d_{0}, d_{1},a^{(\mu)},1,\cdots,\hat{j},\cdots, M, M^{\prime}, \cdots, 1^{\prime})}\nonumber\\
\fl{\quad ={\rm Pf}(a^{(\mu)},d_{1},\cdots).}
\end{eqnarray}
Similarly, we obtain
\begin{eqnarray}
\fl{\displaystyle\frac{\partial g^{(\mu)}}{\partial T}=\partial_{T}{\rm Pf}(a^{(\mu)},d_{0},\cdots)}\nonumber\\
\fl{\quad =-\partial_{T}\left(\sum_{j=1}^{M}(-1)^{j}{\rm Pf}(d_{0},j){\rm Pf}(a^{(\mu)},1,\cdots,\hat{j},\cdots, M, M^{\prime}, \cdots, 1^{\prime})\right)}\nonumber\\
\fl{\quad =-\sum_{j=1}^{M}(-1)^{j}\partial_{T}{\rm Pf}(d_{0},j){\rm Pf}(a^{(\mu)},1,\cdots,\hat{j},\cdots, M, M^{\prime}, \cdots, 1^{\prime})}\nonumber\\
-\sum_{j=1}^{M}(-1)^{j}{\rm Pf}(d_{0},j)\partial_{T}{\rm Pf}(a^{(\mu)},1,\cdots,\hat{j},\cdots, M, M^{\prime}, \cdots, 1^{\prime})\nonumber\\
\fl{\quad =-\sum_{j=1}^{M}(-1)^{j}{\rm Pf}(d_{-1},j){\rm Pf}(a^{(\mu)},1,\cdots,\hat{j},\cdots, M, M^{\prime}, \cdots, 1^{\prime})}\nonumber\\
\fl{\qquad-\sum_{j=1}^{M}(-1)^{j}{\rm Pf}(d_{0},j){\rm Pf}(d_{-1},d_{0},a^{(\mu)},1,\cdots,\hat{j},\cdots, M, M^{\prime}, \cdots, 1^{\prime})}\nonumber\\
\fl{\quad ={\rm Pf}(a^{(\mu)},d_{-1},\cdots).}
\end{eqnarray}
Furthermore, 
\begin{eqnarray}
\fl{\displaystyle\frac{\partial^{2} g^{(\mu)}}{\partial X\partial T}=-\partial_{X}{\rm Pf}(d_{-1},a^{(\mu)},\cdots, M, M^{\prime}, \cdots, 1^{\prime})}\nonumber\\
\fl{\quad =-\partial_{X}\sum_{j=1}^{M}(-1)^{j}{\rm Pf}(d_{-1},j){\rm Pf}(a^{(\mu)},1,\cdots, \hat{j}, \cdots, M, M^{\prime}, \cdots, 1^{\prime})}\nonumber\\
\fl{\quad =-\sum_{j=1}^{M}(-1)^{j}{\rm Pf}(d_{0},j){\rm Pf}(a^{(\mu)},1,\cdots, \hat{j}, \cdots, M, M^{\prime}, \cdots, 1^{\prime})}\nonumber\\
-\sum_{j=1}^{M}(-1)^{j}{\rm Pf}(d_{-1},j){\rm Pf}(d_{0},d_{1},a^{(\mu)},1,\cdots, \hat{j}, \cdots, M, M^{\prime}, \cdots, 1^{\prime})\nonumber\\
\fl{\quad ={\rm Pf}(a^{(\mu)},d_{0},\cdots)-{\rm Pf}(d_{-1},d_{0},d_{1},a^{(\mu)}, \cdots).}
\end{eqnarray}
A Pfaffian identity\cite{direct}
\begin{eqnarray}
\fl{{\rm Pf}(d_{-1},d_{0},d_{1},a^{(\mu)},\cdots){\rm Pf}(\cdots)={\rm Pf}(d_{-1},d_{0},\cdots){\rm Pf}(d_{1},a^{(\mu)},\cdots)}\nonumber\\
-{\rm Pf}(d_{-1},d_{1},\cdots){\rm Pf}(d_{0},a^{(\mu)},\cdots)+{\rm Pf}(d_{-1},a^{(\mu)},\cdots){\rm Pf}(d_{0},d_{1},\cdots)
\end{eqnarray}
leads to
\begin{eqnarray}
\fl{\partial_{X}\partial_{T}f\times g^{(\mu)}-\partial_{X}f \times \partial_{T}g^{(\mu)}-\partial_{T}f \times \partial_{X}g^{(\mu)}+f\times\partial_{X}\partial_{T}g^{(\mu)}=fg^{(\mu)}},
\end{eqnarray}
which is the second equation of (\ref{bilinearMCSP}). The third equation follows
by the same argument, with $g^{(\mu)}$ replaced by $h^{(\nu)}$.
\end{proof}

\begin{eg}
\label{eg2.3}
We now present explicit one- and two-soliton solutions of the MCSP equation (\ref{MCSP}).\\

\noindent
One-soliton: Setting $M=2$ and $N=1$ in (\ref{con_solb})--(\ref{con_sole})
gives the following $\tau$-functions:
\begin{eqnarray}
\fl{f={\rm Pf}(1,2,2^{\prime},1^{\prime})=1+{\rm Pf}(1,2){\rm Pf}(2^{\prime},1^{\prime})}\nonumber\\
\fl{\quad =1+\displaystyle\frac{1}{4}\left(\frac{p_{1}p_{2}}{p_{1}+p_{2}}\right)^{2}\varphi_{1}\varphi_{2}\sum_{\mu\in I}\sum_{\nu\in J}c_{\mu\nu}a_{1}^{(\mu)}b_{2}^{(\nu)},}\label{con_ex1}
\end{eqnarray}
\begin{eqnarray}
\fl{g^{(\mu)}={\rm Pf}(a^{(\mu)},d_{0},1,2,2^{\prime},1^{\prime})={\rm Pf}(a^{(\mu)},1^{\prime}){\rm Pf}(d_{0},1)=a_{1}^{(\mu)}\varphi_{1},\qquad \mu\in I},
\end{eqnarray}
\begin{eqnarray}
\fl{h^{(\nu)}={\rm Pf}(b^{(\nu)},d_{0},1,2,2^{\prime},1^{\prime})={\rm Pf}(b^{(\nu)},2^{\prime}){\rm Pf}(d_{0},2)=b_{2}^{(\nu)}\varphi_{2},\qquad \nu\in J},\label{con_ex2}
\end{eqnarray}
where $\varphi_{i}=B_{i}e^{p_{i}X+p_{i}^{-1}T}$ and $i=1,2$.\\

\noindent
Two-soliton: Setting $M=4$ and $N=2$ in (\ref{con_solb})--(\ref{con_sole})
gives the following $\tau$-functions. 
To express these $\tau$-functions in a compact form, we introduce the notation:
\begin{eqnarray}
\fl C_{ij}&:=&\sum_{\mu\in I}\sum_{\nu\in J}c_{\mu\nu}a_i^{(\mu)}b_j^{(\nu)},\qquad
R_{ij}:=\frac{p_i-p_j}{p_i+p_j},\nonumber\\
\fl \tilde{C}_{ij}&:=&\frac{p_i^2p_j^2}{4(p_i+p_j)^2}C_{ij}, \qquad \Delta:=\frac{p_1^2p_2^2p_3^2p_4^2}{16}
\left|
\begin{array}{rr}
\displaystyle\frac{C_{14}}{p_{1}^{2}-p_{4}^{2}} &
\displaystyle\frac{C_{24}}{p_{2}^{2}-p_{4}^{2}}\\[15pt]
\displaystyle\frac{C_{13}}{p_{1}^{2}-p_{3}^{2}} &
\displaystyle\frac{C_{23}}{p_{2}^{2}-p_{3}^{2}}
\end{array}
\right|.
\label{con_compact_coefficients}
\end{eqnarray}
\begin{eqnarray}
\fl f&=&{\rm Pf}(1,2,3,4,4^{\prime},3^{\prime},2^{\prime},1^{\prime})\nonumber\\
\fl &=&1+\tilde{C}_{14}\varphi_1\varphi_4+\tilde{C}_{13}\varphi_1\varphi_3
+\tilde{C}_{24}\varphi_2\varphi_4+\tilde{C}_{23}\varphi_2\varphi_3\nonumber\\
\fl &&+\Delta\left(\prod_{1\leq i<j\leq4}R_{ij}\right)
\varphi_1\varphi_2\varphi_3\varphi_4,\label{con_ex3}
\end{eqnarray}
\begin{eqnarray}
\fl g^{(\mu)}&=&{\rm Pf}(a^{(\mu)},d_{0},1,2,3,4,4^{\prime},3^{\prime},2^{\prime},1^{\prime})\nonumber\\
\fl &=&a_{1}^{(\mu)}\left(\varphi_1+R_{12}R_{14}\tilde{C}_{24}\varphi_1\varphi_2\varphi_4
+R_{12}R_{13}\tilde{C}_{23}\varphi_1\varphi_2\varphi_3\right)\nonumber\\
\fl &&+a_{2}^{(\mu)}\left(\varphi_2-R_{12}R_{24}\tilde{C}_{14}\varphi_1\varphi_2\varphi_4
-R_{12}R_{23}\tilde{C}_{13}\varphi_1\varphi_2\varphi_3\right),\qquad \mu\in I,
\end{eqnarray}
\begin{eqnarray}
\fl h^{(\nu)}&=&{\rm Pf}(b^{(\nu)},d_{0},1,2,3,4,4^{\prime},3^{\prime},2^{\prime},1^{\prime})\nonumber\\
\fl &=&b_{3}^{(\nu)}\left(\varphi_3-R_{23}R_{34}\tilde{C}_{24}\varphi_2\varphi_3\varphi_4
-R_{13}R_{34}\tilde{C}_{14}\varphi_1\varphi_3\varphi_4\right)\nonumber\\
\fl &&+b_{4}^{(\nu)}\left(\varphi_4+R_{24}R_{34}\tilde{C}_{23}\varphi_2\varphi_3\varphi_4
+R_{14}R_{34}\tilde{C}_{13}\varphi_1\varphi_3\varphi_4\right),\qquad \nu\in J.\label{con_ex4}
\end{eqnarray}
where $\varphi_{i}=B_{i}e^{p_{i}X+p_{i}^{-1}T}$ and $i=1,2,3,4$.\\
\end{eg}

\begin{rem}
\label{rem2.4}
The $N$-soliton solutions of the $K$-component CSP equation follow from those
of the $2K$-component SP equation under the reduction
\begin{eqnarray}
\fl{I=J=\{1,2,\cdots,K\},\quad M=2N,\quad c_{\mu\nu}=\delta_{\mu\nu},}\nonumber\\
\fl{p_{i}=p_{2N+1-i}^{*},\quad B_{i}=B_{2N+1-i}^{*},\quad a_{i}^{(\mu)}=(b_{2N+1-i}^{(\mu)})^*, \qquad1\leq i\leq N,\quad 1 \leq \mu \leq K.}
\end{eqnarray}
Examples of the corresponding CCSP solutions are presented in Section \ref{sec_nc}.
\end{rem}
\end{section}

\begin{section}{Semi-discretization of the MCSP equation}
\label{sec_semi}
We take the following system as a semi-discrete analogue of the bilinear
system (\ref{bilinearMCSP}):
\begin{eqnarray}
\left\{
\begin{array}{ll}
 D_{T}^{2}f_{l}\cdot f_{l}=\displaystyle\frac{1}{2}\displaystyle\sum_{\mu\in I}\sum_{\nu\in J}c_{\mu\nu}g_{l}^{(\mu)}h_{l}^{(\nu)},\\
\displaystyle\frac{1}{a}D_{T}(f_{l+1}\cdot g_{l}^{(\mu)}-f_{l}\cdot g_{l+1}^{(\mu)})=f_{l+1}g_{l}^{(\mu)}+f_{l}g_{l+1}^{(\mu)},\\
\\
\displaystyle\frac{1}{a}D_{T}(f_{l+1}\cdot h_{l}^{(\nu)}-f_{l}\cdot h_{l+1}^{(\nu)})=f_{l+1}h_{l}^{(\nu)}+f_{l}h_{l+1}^{(\nu)},
\end{array}
\right.
\label{disbilinear}
\end{eqnarray}
where $l\in\mathbb{Z}$ labels the uniform lattice
$X_l=2al$ in the hodograph coordinate $X$, $a>0$, and $2a$ is the
corresponding lattice spacing. The physical mesh interval in the $x$ coordinate
will be defined below and is generally nonuniform.

To construct Pfaffian solutions of the bilinear system (\ref{disbilinear}), let $M=2N$ and introduce the following Pfaffian elements:
\begin{eqnarray}
\fl{{\rm Pf}(j,k)_{l}=\frac{p_{j}-p_{k}}{p_{j}+p_{k}}\varphi_{j}(l)\varphi_{k}(l),\quad(1\leq j<k\leq M),}
\label{sdpfMCSP0}
\end{eqnarray}
\begin{eqnarray}
\fl{{\rm Pf}(j,k^{\prime})_{l}=\delta_{jk},\quad (1\leq j\leq M,\quad 1\leq k\leq M),\label{sdpfMCSP1}}
\end{eqnarray}
\begin{eqnarray}
\fl{{\rm Pf}(k^{\prime},j^{\prime})_{l}=
\left\{
\begin{array}{ll}
\displaystyle\frac{p_{j}^{2}p_{k}^{2}}{4(p_{j}^{2}-p_{k}^{2})}\sum_{\mu\in I}\sum_{\nu\in J} c_{\mu\nu}a_{j}^{(\mu)}b_{k}^{(\nu)},& (1\leq j\leq N<k\leq M),\\
0,&{\rm otherwise},
\end{array}
\right.}
\end{eqnarray}
\begin{eqnarray}
\fl{{\rm Pf}(a^{(\mu)}, j^{\prime})_{l}=
\left\{
\begin{array}{ll}
a_{j}^{(\mu)},&(1\leq j\leq N),\\
0,&(N+1\leq j\leq M),
\end{array}
\right.}
\qquad\mu\in I,
\label{sdpfMCSP3}
\end{eqnarray}
\begin{eqnarray}
\fl{{\rm Pf}(b^{(\nu)}, k^{\prime})_{l}=
\left\{
\begin{array}{ll}
0,&(1\leq k\leq N),\\
b_{k}^{(\nu)},&(N+1\leq k\leq M),
\end{array}
\right.}
\qquad\nu\in J,
\label{sdpfMCSP2}
\end{eqnarray}
where
$\varphi_{i}(l)=B_{i}\left(\displaystyle\frac{1+ap_{i}}{1-ap_{i}}\right)^{l}e^{p_{i}^{-1}T}$,
and $B_i$ and $p_i$ are the phase and spectral parameters associated with the
$i$-th exponential factor, respectively. The element involving $d_0$ is
defined by
\begin{eqnarray}
{\rm Pf}(d_0,i)_l=\varphi_i(l),\qquad 1\leq i\leq M.
\label{sdpfMCSP4}
\end{eqnarray}
All other Pfaffian elements formed from the labels appearing in the solution
formulas above are set to zero. The functions $\varphi_j$ satisfy the linear dispersion relation
\begin{eqnarray} 
\displaystyle\frac{\varphi_{j}(l+1)-\varphi_{j}(l)}{a}=p_{j}(\varphi_{j}(l+1)+\varphi_{j}(l)).\label{lde}
\end{eqnarray}
The constants $a_{j}^{(\mu)}$ and $b_{k}^{(\nu)}$ are arbitrary, and the parameters are
chosen so that all expressions above are well-defined.

\begin{thm}
\label{thm3.1}
For the Pfaffian elements defined in (\ref{sdpfMCSP0})--(\ref{sdpfMCSP4}), the functions 
\begin{eqnarray}
f_{l}={\rm Pf}(1,2,\cdots, M-1, M, M^{\prime}, M-1^{\prime}, \cdots, 2^{\prime}, 1^{\prime})_{l},\label{semidis_solb}\\
g_{l}^{(\mu)}={\rm Pf}(a^{(\mu)}, d_{0},1,2,\cdots, M-1, M, M^{\prime}, M-1^{\prime}, \cdots, 2^{\prime}, 1^{\prime})_{l},\\
h_{l}^{(\nu)}={\rm Pf}(b^{(\nu)}, d_{0},1,2,\cdots, M-1, M, M^{\prime}, M-1^{\prime}, \cdots, 2^{\prime}, 1^{\prime})_{l},
\label{semidis_sole}
\end{eqnarray}
with $\mu\in I$ and $\nu\in J$,  satisfy the semi-discrete bilinear system (\ref{disbilinear}).
\end{thm}

\begin{proof}
The first equation of (\ref{disbilinear}) follows as in the continuous case.
We therefore verify the second equation. For this purpose, introduce the
auxiliary Pfaffian label $s_{l}^{+}$, associated with the forward shift
$l\mapsto l+1$, and define the following additional Pfaffian elements:
\begin{eqnarray}
{\rm Pf}(d_{-1},i)_{l}=p_{i}^{-1}\varphi_{i}(l),\qquad{\rm Pf}(d_{0},s_{l}^{+})_{l}=1,\nonumber\\
{\rm Pf}(s_{l}^{+},i)_{l}=\varphi_{i}(l+1),\qquad{\rm Pf}(s_{l}^{+},d_{-1})_{l}=a,
\end{eqnarray}
where $1\leq i \leq M$. We set ${\rm Pf}(d_{-1},d_0)_l=0$ and set all
other Pfaffian elements involving the newly introduced labels $d_{-1}$ and
$s^l$ to zero.
As in the continuous case, we obtain 
\begin{eqnarray}
\fl{\partial_{T}f_{l}={\rm Pf}(d_{-1},d_{0},\cdots)_{l},\qquad\partial_{T}g_{l}^{(\mu)}={\rm Pf}(a^{(\mu)}, d_{-1},\cdots)_{l}.}\nonumber
\end{eqnarray}
Since
\begin{eqnarray}
\fl{{\rm Pf}(j,k)_{l+1}={\rm Pf}(j,k)_{l}+\varphi_{j}(l+1)\varphi_{k}(l)-\varphi_{j}(l)\varphi_{k}(l+1)}\nonumber\\
\fl{\quad ={\rm Pf}(d_{0},s_{l}^{+})_{l}{\rm Pf}(j,k)_{l}-{\rm Pf}(d_{0},j)_{l}{\rm Pf}(s_{l}^{+},k)_{l}+{\rm Pf}(d_{0},k)_{l}{\rm Pf}(s_{l}^{+},j)_{l}}\nonumber\\
\fl{\quad ={\rm Pf}(d_{0},s_{l}^{+},j,k)_{l},}\nonumber\\
\fl{(\partial_{T}-a){\rm Pf}(j,k)_{l+1}=\partial_{T}{\rm Pf}(j,k)_{l+1}-a{\rm Pf}(j,k)_{l+1}}\nonumber\\
\fl{\quad =p_{k}^{-1}\varphi_{j}(l+1)\varphi_{k}(l+1)-p_{j}^{-1}\varphi_{j}(l+1)\varphi_{k}(l+1)}\nonumber\\
\fl{\qquad -a\left({\rm Pf}(j,k)_{l}+\varphi_{j}(l+1)\varphi_{k}(l)-\varphi_{j}(l)\varphi_{k}(l+1)\right)}\nonumber\\
\fl{\quad =-a{\rm Pf}(j,k)_{l}+\varphi_{j}(l+1)\left(p_{k}^{-1}\varphi_{k}(l+1)-a\varphi_{k}(l)\right)
-\varphi_{k}(l+1)\left(p_{j}^{-1}\varphi_{j}(l+1)-a\varphi_{j}(l)\right)}\nonumber\\
\fl{\quad =-a{\rm Pf}(j,k)_{l}+\varphi_{j}(l+1)\left(a\varphi_{k}(l+1)+p_{k}^{-1}\varphi_{k}(l)\right)
-\varphi_{k}(l+1)\left(a\varphi_{j}(l+1)+p_{j}^{-1}\varphi_{j}(l)\right)}\nonumber\\
\fl{\quad =-a{\rm Pf}(j,k)_{l}+p_{k}^{-1}\varphi_{j}(l+1)\varphi_{k}(l)-p_{j}^{-1}\varphi_{j}(l)\varphi_{k}(l+1)}\nonumber\\
\fl{\quad =-{\rm Pf}(s_{l}^{+},d_{-1})_{l}{\rm Pf}(j,k)_{l}+{\rm Pf}(s_{l}^{+},j)_{l}{\rm Pf}(d_{-1},k)_{l}
-{\rm Pf}(d_{-1},j)_{l}{\rm Pf}(s_{l}^{+},k)_{l}}\nonumber\\
\fl{\quad =-{\rm Pf}(s_{l}^{+},d_{-1},j,k)_{l},}\nonumber
\end{eqnarray}
we obtain
\begin{eqnarray}
\fl{f_{l+1}={\rm Pf}(d_{0},s_{l}^{+},\cdots)_{l},\qquad(\partial_{T}-a)f_{l+1}={\rm Pf}(d_{-1},s_{l}^{+},\cdots)_{l}.}\nonumber
\end{eqnarray}
where the linear dispersion relation (\ref{lde}) was used. Moreover,
\begin{eqnarray}
\fl{g_{l+1}^{(\mu)}=-\sum_{j=1}^{M}(-1)^{j}{\rm Pf}(d_{0},j)_{l+1}{\rm Pf}(a^{(\mu)},1,\cdots,\hat{j},\cdots)_{l+1}}\nonumber\\
\fl{\qquad =-\sum_{j=1}^{M}(-1)^{j}{\rm Pf}(s_{l}^{+},j)_{l}{\rm Pf}(d_{0},s_{l}^{+},a^{(\mu)},1,\cdots,\hat{j},\cdots)_{l}}\nonumber\\
\fl{\qquad =-{\rm Pf}(s_{l}^{+},a^{(\mu)},\cdots)_{l},}\nonumber
\end{eqnarray}
by expanding the trivially vanishing Pfaffian ${\rm Pf}(s_{l}^{+},a^{(\mu)},d_{0},s_{l}^{+}, 1, \cdots, M, M^{\prime}, \cdots, 1^{\prime})_{l}$.
A simple expression for $(\partial_{T}-a)g_{l+1}^{(\mu)}$ is obtained as follows:
\begin{eqnarray}
\fl{(\partial_{T}-a)g_{l+1}^{(\mu)}}\nonumber\\
\fl{\quad =-\sum_{j=1}^{M}(-1)^{j}(\partial_{T}-a){\rm Pf}(s_{l}^{+},j)_{l}{\rm Pf}(a^{(\mu)},1,\cdots,\hat{j},\cdots)_{l}}\nonumber\\
\fl{\qquad-\sum_{j=1}^{M}(-1)^{j}{\rm Pf}(s_{l}^{+},j)_{l}\partial_{T}{\rm Pf}(a^{(\mu)}, 1,\cdots,\hat{j},\cdots)_{l}}\nonumber\\
\fl{\quad =-\sum_{j=1}^{M}(-1)^{j}({\rm Pf}(d_{-1},j)_{l}+a{\rm Pf}(d_{0},j)_{l}){\rm Pf}(a^{(\mu)}, 1,\cdots,\hat{j},\cdots)_{l}}\nonumber\\
\fl{\qquad-\sum_{j=1}^{M}(-1)^{j}{\rm Pf}(s_{l}^{+},j)_{l}{\rm Pf}(d_{-1},d_{0},a^{(\mu)}, 1,\cdots,\hat{j},\cdots)_{l}}\nonumber\\
\fl{\quad =-{\rm Pf}(d_{-1},a^{(\mu)},\cdots)_{l}-a {\rm Pf}(d_{0},a^{(\mu)},\cdots)_{l}}\nonumber\\
\fl{\qquad-\sum_{j=1}^{M}(-1)^{j}{\rm Pf}(s_{l}^{+},j)_{l}{\rm Pf}(d_{-1},d_{0},a^{(\mu)},1,\cdots,\hat{j},\cdots)_{l}}\nonumber\\
\fl{\quad =-{\rm Pf}(s_{l}^{+},d_{-1},d_{0},a^{(\mu)},\cdots)_{l}.}\nonumber
\end{eqnarray}

A Pfaffian identity\cite{direct}
\begin{eqnarray}
\fl{-{\rm Pf}(s_{l}^{+},d_{-1},d_{0},a^{(\mu)},\cdots)_{l}{\rm Pf}(\cdots)_{l}+{\rm Pf}(s_{l}^{+},d_{-1},\cdots)_{l}{\rm Pf}(d_{0},a^{(\mu)},\cdots)_{l}}\nonumber\\
\fl{\qquad-{\rm Pf}(s_{l}^{+},d_{0},\cdots)_{l}{\rm Pf}(d_{-1},a^{(\mu)},\cdots)_{l}
+{\rm Pf}(s_{l}^{+},a^{(\mu)},\cdots)_{l}{\rm Pf}(d_{-1},d_{0},\cdots)_{l}=0,}\nonumber
\end{eqnarray}
gives
\begin{eqnarray}
\fl{(\partial_{T}-a)g_{l+1}^{(\mu)}\times f_{l}=g_{l+1}^{(\mu)}\times \partial_{T}f_{l}
-(\partial_{T}-a)f_{l+1}\times g_{l}^{(\mu)}+\partial_{T}g_{l}^{(\mu)}\times f_{l+1},}\nonumber
\end{eqnarray}
which is the second equation of (\ref{disbilinear}). The third equation follows
by the same argument, with $g_l^{(\mu)}$ replaced by $h_l^{(\nu)}$.
\end{proof}

To transform the bilinear system (\ref{disbilinear}) into a semi-discrete system in the physical variables, we introduce the dependent-variable transformation
\begin{eqnarray}
u_{l}^{(\mu)}=\frac{g_{l}^{(\mu)}}{f_{l}},\qquad v_{l}^{(\nu)}=\frac{h_{l}^{(\nu)}}{f_{l}},\qquad
\rho_{l}=1-\displaystyle\frac{1}{a}\left(\log{\frac{f_{l+1}}{f_{l}}}\right)_{T}.
\label{disdependenttransformation}
\end{eqnarray}
From derivative relation (\ref{derivativelaw}), we obtain
\begin{eqnarray}
\frac{\partial x}{\partial X}=\rho.
\label{xXrelation}
\end{eqnarray}
For fixed $T$, integrating (\ref{xXrelation}), with respect to $X$ gives the integral form of the hodograph transformation (\ref{hodograph})
\begin{eqnarray}
x=\int\rho(X,T)dX=x_{0}(T)+\int_{0}^{X}\rho(\bar{X},T)d\bar{X},
\label{int_hodograph}
\end{eqnarray}
where $x_{0}(T)=x(0,T)$ is an integration constant with respect to $X$. Discretizing (\ref{int_hodograph}) yields the discrete hodograph transformation 
\begin{eqnarray}
x_{l}=x_{0}(T)+\sum_{m=0}^{l-1}2a\rho_{m}\qquad(l\geq1),
\label{dishodographMCSP}
\end{eqnarray}
where $x_{l}=x_{l}(T)$ denotes the $l$-th mesh point and $\rho_{l}$ denotes the semi-discrete density at $X_{l}=2al$.
The mesh point $x_{0}(T)$ is allowed to evolve and is assumed to satisfy the consistency condition 
\begin{eqnarray}
\frac{dx_{0}(T)}{dT}=-\displaystyle\frac{1}{2}\sum_{\mu\in I}\displaystyle\sum_{\nu\in J}c_{\mu\nu}u_{0}^{(\mu)}v_{0}^{(\nu)}.
\label{semi-con-con}
\end{eqnarray}
The derivation of this consistency condition is given in \cite{hori2}.
The discrete hodograph transformation (\ref{dishodographMCSP}) can also be written locally in terms of the mesh interval as
\begin{eqnarray}
\delta_l:=x_{l+1}-x_l=2a\rho_l.
\label{delta_semidis}
\end{eqnarray}
The variables $u^{(\mu)}_{l}$ and $v^{(\nu)}_{l}$ are regarded as the field values at the physical mesh point $x_{l}(T)$.

\begin{thm}
\label{thm3.2}
Let $(f_l,g_l^{(\mu)},h_l^{(\nu)})$ be a solution of the
semi-discrete bilinear system (\ref{disbilinear}), and define
$u_l^{(\mu)}$, $v_l^{(\nu)}$, $\rho_l$, $x_l$, and $\delta_l$
by (\ref{disdependenttransformation}), (\ref{dishodographMCSP}) and (\ref{delta_semidis}).
Assume that $x_0(T)$ satisfies (\ref{semi-con-con}).
Then the variables $u_l^{(\mu)}$, $v_l^{(\nu)}$, and $x_l$
satisfy the following semi-discrete MCSP system:
\begin{eqnarray}
\left\{
\begin{array}{lll}
\displaystyle\frac{d}{dT}(u_{l+1}^{(\mu)}-u_{l}^{(\mu)})=\displaystyle\frac{\delta_{l}}{2}(u_{l+1}^{(\mu)}+u_{l}^{(\mu)}), \\
\\
\displaystyle\frac{d}{dT}(v_{l+1}^{(\nu)}-v_{l}^{(\nu)})=\displaystyle\frac{\delta_{l}}{2}(v_{l+1}^{(\nu)}+v_{l}^{(\nu)}), \\
\\
\displaystyle\frac{d x_{l}}{d T}=-\displaystyle\frac{1}{2}\sum_{\mu\in I}\sum_{\nu\in J}c_{\mu\nu}u_{l}^{(\mu)}v_{l}^{(\nu)}.
\end{array}
\right.
\label{semiMCSP}
\end{eqnarray}
Here $\mu\in I, \nu\in J$, and all $T$-derivatives are taken with the lattice label $l$ fixed.
\end{thm}

\begin{proof}
Dividing the second and third equations of (\ref{disbilinear}) by
$f_{l+1}f_l$ gives
\begin{eqnarray}
\left(\frac{g_{l+1}^{(\mu)}}{f_{l+1}}-\frac{g_{l}^{(\mu)}}{f_{l}}\right)_{T}=\left(a-\left(\log{\frac{f_{l+1}}{f_{l}}}\right)_{T}\right)\left(\frac{g_{l+1}^{(\mu)}}{f_{l+1}}+\frac{g_{l}^{(\mu)}}{f_{l}}\right),
\label{bi1}
\end{eqnarray}
and
\begin{eqnarray}
\left(\frac{h_{l+1}^{(\nu)}}{f_{l+1}}-\frac{h_{l}^{(\nu)}}{f_{l}}\right)_{T}=\left(a-\left(\log{\frac{f_{l+1}}{f_{l}}}\right)_{T}\right)\left(\frac{h_{l+1}^{(\nu)}}{f_{l+1}}+\frac{h_{l}^{(\nu)}}{f_{l}}\right).
\label{bi2}
\end{eqnarray}
Applying the dependent-variable transformation (\ref{disdependenttransformation}) to (\ref{bi1}) and (\ref{bi2}), we obtain
\begin{eqnarray}
\frac{d}{dT}(u_{l+1}^{(\mu)}-u_{l}^{(\mu)})=a\rho_{l}(u_{l+1}^{(\mu)}+u_{l}^{(\mu)}),\label{MCCID1}
\end{eqnarray}
and
\begin{eqnarray}
\frac{d}{dT}(v_{l+1}^{(\nu)}-v_{l}^{(\nu)})=a\rho_{l}(v_{l+1}^{(\nu)}+v_{l}^{(\nu)}).\label{MCCID2}
\end{eqnarray}
Since $\delta_{l}=2a\rho_{l}$, equations (\ref{MCCID1}) and (\ref{MCCID2}) yield the first two equations of (\ref{semiMCSP}).

It remains to derive the evolution equation for the mesh points.
Using (\ref{disdependenttransformation}), the first equation of
(\ref{disbilinear}) yields
\begin{eqnarray}
(\log{f_{l}})_{TT}=\displaystyle\frac{1}{4}\sum_{\mu\in I}\sum_{\nu\in J}c_{\mu\nu}u_{l}^{(\mu)}v_{l}^{(\nu)}.\label{eq_33}
\end{eqnarray}
Combining (\ref{eq_33}) with (\ref{disdependenttransformation}) yields
 \begin{eqnarray}
\frac{d}{dT}\rho_{l}=-\displaystyle\frac{1}{4a}\sum_{\mu\in I}\sum_{\nu\in J}c_{\mu\nu}(u_{l+1}^{(\mu)}v_{l+1}^{(\nu)}-u_{l}^{(\mu)}v_{l}^{(\nu)}).
\label{rho_evolution}
\end{eqnarray}
For convenience, set
\begin{eqnarray}
F_l=\sum_{\mu\in I}\sum_{\nu\in J}c_{\mu\nu}
u_l^{(\mu)}v_l^{(\nu)}.
\end{eqnarray}
Then (\ref{rho_evolution}) can be written as
\begin{eqnarray}
\frac{d\rho_l}{dT}=-\frac{1}{4a}\left(F_{l+1}-F_l\right).
\label{rho_F}
\end{eqnarray}

Differentiating the discrete hodograph transformation (\ref{dishodographMCSP}) with respect to $T$ and using (\ref{semi-con-con}) and (\ref{rho_F}), we obtain
\begin{eqnarray}
\frac{dx_{l}}{dT}&=&\frac{dx_{0}}{dT}+2a\sum_{m=0}^{l-1}\frac{d\rho_{m}}{dT}
\nonumber\\
&=&-\frac{1}{2}F_{0}-\frac{1}{2}\sum_{m=0}^{l-1}\left(F_{m+1}-F_{m}\right)\nonumber\\
&=&-\frac{1}{2}F_{l}\nonumber\\
&=&-\frac{1}{2}\sum_{\mu\in I}\sum_{\nu\in J}c_{\mu\nu}u_{l}^{(\mu)}v_{l}^{(\nu)}.
\label{x_evolution}
\end{eqnarray}
This is the third equation of (\ref{semiMCSP}).
\end{proof}

The mesh points can written directly in terms of the $\tau$-function. Substituting (\ref{disdependenttransformation}) into (\ref{dishodographMCSP}) gives
\begin{eqnarray}
x_{l}=x_{0}(T)+2al-2\left(\log\frac{f_{l}}{f_{0}}\right)_{T}.
\label{x_tau_general}
\end{eqnarray}
At $l=0$, equation (\ref{eq_33}) together with (\ref{semi-con-con}) gives
\begin{eqnarray}
x_{0,T}=-2(\log f_{0})_{TT}.
\end{eqnarray}
Hence, up to a time-independent translation of $x$, we may choose
\begin{eqnarray}
x_{0}(T)=-2(\log f_{0})_{T}.
\end{eqnarray}
Equation (\ref{x_tau_general}) then reduces to
\begin{eqnarray}
x_{l}=2al-2(\log f_{l})_{T}.
\label{semidisparametricMCSP}
\end{eqnarray}
Therefore, the Pfaffian solution of the semi-discrete MCSP system (\ref{semiMCSP}) is given parametrically by
\begin{eqnarray}
u_{l}^{(\mu)}=\frac{g_{l}^{(\mu)}}{f_{l}},\qquad v_{l}^{(\nu)}=\frac{h_{l}^{(\nu)}}{f_{l}},\qquad x_{l}=
2al-2(\log f_{l})_{T}.
\label{semi_pf_solution}
\end{eqnarray}

\end{section}


\begin{section}{\texorpdfstring{Full}{Full} discretization of the MCSP equation}
\label{sec_full}

We take the following system as a fully discrete analogue of the bilinear
system (\ref{bilinearMCSP}):
\begin{eqnarray}
\fl{
\left\{ 
\begin{array}{ll}
f_{l,m+1}f_{l,m-1}-f_{l,m}^{2}=b^{2}\displaystyle\sum_{\mu\in I}\displaystyle\sum_{\nu\in J}c_{\mu\nu}g_{l,m}^{(\mu)}h_{l,m}^{(\nu)},\\
(1-ab)(g^{(\mu)}_{l+1,m+1}f_{l,m}+g_{l,m}^{(\mu)}f_{l+1,m+1})=(1+ab)(g_{l+1,m}^{(\mu)}f_{l,m+1}+g^{(\mu)}_{l,m+1}f_{l+1,m}),\\
(1-ab)(h^{(\nu)}_{l+1,m+1}f_{l,m}+h_{l,m}^{(\nu)}f_{l+1,m+1})=(1+ab)(h_{l+1,m}^{(\nu)}f_{l,m+1}+h^{(\nu)}_{l,m+1}f_{l+1,m}),\\
\end{array}
\right.}
\label{fullbilinear}
\end{eqnarray}
where $l,m\in\mathbb{Z}$ label the uniform lattice
$X_l=2al$, $T_m=2bm$ in the hodograph coordinates, and $a,b>0$.
Thus $2a$ and $2b$ are the lattice spacings in $X$ and $T$, respectively;
the physical mesh in the $x$ coordinate is generally nonuniform. We assume
throughout this section that all denominators appearing below are nonzero. In
particular, $1-ab\neq0$ and the relevant $\tau$-functions do not vanish when
the dependent-variable transformations are used.

To construct Pfaffian solutions of the fully discrete bilinear system (\ref{fullbilinear}), let $M=2N$ and introduce the following Pfaffian elements: 
\begin{eqnarray}
\fl{{\rm Pf}(j,k)_{l,m}=\frac{p_{j}-p_{k}}{p_{j}+p_{k}}\varphi_{j}(l,m)\varphi_{k}(l,m),\quad(1\leq j<k\leq M),}\label{fdpfMCSP1}
\end{eqnarray}
\begin{eqnarray}
\fl{{\rm Pf}(j,k^{\prime})_{l,m}=\delta_{jk},\quad (1\leq j\leq M,\quad 1\leq k\leq M),}
\end{eqnarray}
\begin{eqnarray}
\fl{{\rm Pf}(k^{\prime},j^{\prime})_{l,m}=
\left\{
\begin{array}{ll}
\begin{array}{c}
\displaystyle\frac{(p_{j}^{2}-b^{2})(p_{k}^{2}-b^{2})}{4(p_{j}^{2}-p_{k}^{2})}\\
\displaystyle{}\times\sum_{\mu\in I}\sum_{\nu\in J} c_{\mu\nu}a_{j}^{(\mu)}b_{k}^{(\nu)}
\end{array},&(1\leq j\leq N<k\leq M),\\
0,&{\rm otherwise},
\end{array}
\right.}
\end{eqnarray}
\begin{eqnarray}
\fl{{\rm Pf}(a^{(\mu)}, j^{\prime})_{l,m}=
\left\{
\begin{array}{ll}
a_{j}^{(\mu)},&(1\leq j\leq N),\\
0,&(N+1\leq j\leq M),
\end{array}
\right.}
\qquad\mu\in I,
\label{fdpfMCSP3}
\end{eqnarray}
\begin{eqnarray}
\fl{{\rm Pf}(b^{(\nu)}, k^{\prime})_{l,m}=
\left\{
\begin{array}{ll}
0,& (1\leq k\leq N),\\
b_{k}^{(\nu)},&(N+1\leq k\leq M),
\end{array}
\right.}
\qquad\nu\in J.
\label{fdpfMCSP2}
\end{eqnarray}
where $\varphi_{j}(l,m)=B_{j}\left(\displaystyle\frac{1+ap_{j}}{1-ap_{j}}\right)^{l}\left(\displaystyle\frac{p_{j}+b}{p_{j}-b}\right)^{m}$, and $B_{j}$ and $p_{j}$ are the phase and spectral parameters associated with the $j$-th exponential factor, respectively. The constants $c_{\mu\nu}$, $a_{j}^{(\mu)}$ and $b_{k}^{(\nu)}$ are arbitrary, and the parameters are chosen so that all expressions above are well-defined.
The element involving $d_0$ is defined by
\begin{eqnarray}
{\rm Pf}(d_0,i)_{l,m}=\varphi_i(l,m),\qquad 1\leq i\leq M.
\label{fdpfMCSP4}
\end{eqnarray}
All other Pfaffian elements formed from the labels appearing in the solution
formulas above are set to zero.
In particular, $1\pm ap_j\neq0$, $p_j\pm b\neq0$, and all denominators in
the Pfaffian elements are assumed to be nonzero.

\begin{thm}
\label{thm4.1}
For the Pfaffian elements defined in (\ref{fdpfMCSP1})--(\ref{fdpfMCSP4}), the functions
\begin{eqnarray}
f_{l,m}={\rm Pf}(1,2,\cdots, M-1, M, M^{\prime}, M-1^{\prime}, \cdots, 2^{\prime}, 1^{\prime})_{l,m},\label{fulldis_solb}\\
g_{l,m}^{(\mu)}={\rm Pf}(a^{(\mu)}, d_{0},1,2,\cdots, M-1, M, M^{\prime}, M-1^{\prime}, \cdots, 2^{\prime}, 1^{\prime})_{l,m},\\
h_{l,m}^{(\nu)}={\rm Pf}(b^{(\nu)}, d_{0},1,2,\cdots, M-1, M, M^{\prime}, M-1^{\prime}, \cdots, 2^{\prime}, 1^{\prime})_{l,m},
\label{fulldis_sole}
\end{eqnarray}
where $\mu\in I$ and $\nu\in J$, satisfy the fully discrete bilinear system (\ref{fullbilinear}).
\end{thm}

\begin{proof}

The proof is divided into two steps. We first verify the first bilinear equation. 
We then verify the second and third bilinear equations by using mixed lattice-shift representations and a Pfaffian identity. The detailed Pfaffian expansions used below are collected in Appendix A.
\\
\\
\textit{Step 1: Verification of the first bilinear equation.}
Define the additional Pfaffian elements
\begin{eqnarray}
\fl{{\rm Pf}(s_{l}^{+},d_{0})_{l,m}={\rm Pf}(s_{l}^{+},d_{0})_{l,m+1}=-1,\qquad {\rm Pf}(s_{m}^{+}, d_{0})_{l,m}={\rm Pf}(s_{m}^{-}, d_{0})_{l,m}=1,}\nonumber\\
\fl{{\rm Pf}(s_{l}^{+},i)_{l,m}=\varphi_{i}(l+1,m),\quad{\rm Pf}(s_{m}^{+},i)_{l,m}=\varphi_{i}(l,m+1),}\nonumber\\
\fl{{\rm Pf}(s_{m}^{-},i)_{l,m}=\varphi_{i}(l,m-1),}
\end{eqnarray}
where $1\leq i \leq M$. Here, $s_{l}^{+}$ is the Pfaffian label associated with
the forward shift in $l$, whereas $s_{m}^{+}$ and $s_{m}^{-}$ are distinct Pfaffian
labels associated with the forward and backward shifts in $m$, respectively.
All other Pfaffian elements involving these newly introduced labels are set
to zero. 
Applying Lemma \ref{lem2.1} with $P_{j}=\displaystyle\frac{-4}{p_{j}^{2}-b^{2}}$, we obtain
\begin{eqnarray}
\fl{\displaystyle\sum_{\mu\in I}\displaystyle\sum_{\nu\in J}c_{\mu\nu}{\rm Pf}(a^{(\mu)}, d_{0},1,\cdots, M, M^{\prime}, \cdots, 2^{\prime}, 1^{\prime})_{l,m}}\nonumber\\
\fl{\qquad\times{\rm Pf}(b^{(\nu)}, d_{0},1,\cdots, M, M^{\prime}, \cdots, 2^{\prime}, 1^{\prime})_{l,m}}\nonumber\\
\fl{=\sum_{j=1}^{M}\displaystyle\frac{-4}{p_{j}^{2}-b^{2}}{\rm Pf}(d_{0},1,\cdots, \hat{j},\cdots, M, M^{\prime}, \cdots, 2^{\prime},1^{\prime})_{l,m}}\nonumber\\
\fl{\qquad\times{\rm Pf}(d_{0},1,\cdots, M, M^{\prime}, \cdots, \hat{j^{\prime}}, \cdots, 2^{\prime},1^{\prime})_{l,m}.}\label{eq_thm2}
\end{eqnarray}
The $m$-shift formulas for $f_{l,m}$ are
\begin{eqnarray}
\fl{f_{l,m+1}={\rm Pf}(s_{m}^{+},d_{0},\cdots)_{l,m},\qquad f_{l,m-1}={\rm Pf}(s_{m}^{-},d_{0},\cdots)_{l,m}.}
\label{fd_shift relation}
\end{eqnarray}
Using the shift relations (\ref{fd_shift relation}), the left-hand side of the first bilinear equation in (\ref{fullbilinear}) can be written as
\begin{eqnarray}
\fl{f_{l,m+1}f_{l,m-1}-f_{l,m}^{2}=b^{2}\displaystyle\sum_{j=1}^{M}\displaystyle\frac{-4}{p_{j}^{2}-b^{2}}{\rm Pf}(d_{0},\cdots,\hat{j},\cdots)_{l,m}{\rm Pf}(d_{0},\cdots,\hat{j^{\prime}},\cdots)_{l,m}}.
\label{fd_first_relation}
\end{eqnarray}
The detailed derivations of (\ref{fd_shift relation}) and (\ref{fd_first_relation}) are provided in Appendix A.1. Combining (\ref{eq_thm2}) and (\ref{fd_first_relation}), we obtain
\begin{eqnarray} 
\fl{f_{l,m+1}f_{l,m-1}-f_{l,m}^{2}=b^{2}\displaystyle\sum_{\mu\in I}\sum_{\nu\in J}c_{\mu\nu}{\rm Pf}(a^{(\mu)},d_{0},\cdots)_{l,m}{\rm Pf}(b^{(\nu)},d_{0},\cdots)_{l,m}},
\end{eqnarray}
which is the first equation of (\ref{fullbilinear}).
\\
\\
\textit{Step 2: Verification of the second and third bilinear equations.}
The $l$- and $m$-shift formulas for $f_{l,m}$ and $g^{(\mu)}_{l,m}$ are
\begin{eqnarray}
f_{l,m+1}={\rm Pf}(s_{m}^{+},d_{0},\cdots)_{l,m},\qquad f_{l+1,m}={\rm Pf}(d_{0},s_{l}^{+}\cdots)_{l,m},\nonumber\\
g^{(\mu)}_{l,m+1}=-{\rm Pf}(s_{m}^{+},a^{(\mu)},\cdots)_{l,m},\qquad g^{(\mu)}_{l+1,m}=-{\rm Pf}(s_{l}^{+},a^{(\mu)},\cdots)_{l,m}.
\label{fg_shift}
\end{eqnarray}
From these formulas, we obtain the following mixed $l$- and $m$-shift relations:
\begin{eqnarray}
\frac{1-ab}{1+ab}(f_{l+1,m+1}-f_{l,m})=-{\rm Pf}(s_{l}^{+},s_{m}^{+},\cdots)_{l,m},
\label{fg_shift1}
\end{eqnarray}
and
\begin{eqnarray}
\frac{1-ab}{1+ab}(g^{(\mu)}_{l+1,m+1}+g^{(\mu)}_{l,m})=-{\rm Pf}(s_{m}^{+}, d_{0}, s_{l}^{+}, a^{(\mu)}, \cdots)_{l,m}.
\label{fg_shift2}
\end{eqnarray}
The detailed derivations of (\ref{fg_shift}), (\ref{fg_shift1}) and (\ref{fg_shift2}) are provided in Appendix A.2.
Combing (\ref{fg_shift})--(\ref{fg_shift2}) with a Pfaffian identity\cite{direct}, we obtain
\begin{eqnarray}
\fl{\frac{1-ab}{1+ab}(g_{l+1,m+1}^{(\mu)}+g_{l,m}^{(\mu)})f_{l,m}+\frac{1-ab}{1+ab}(f_{l+1,m+1}-f_{l,m})g^{(\mu)}_{l,m}}\nonumber\\
\fl{\qquad-g^{(\mu)}_{l+1,m}f_{l,m+1}-g^{(\mu)}_{l,m+1}f_{l+1,m}}\nonumber\\
\fl{\quad =-{\rm Pf}(s_{m}^{+},d_{0},s_{l}^{+},a^{(\mu)},\cdots)_{l,m}{\rm Pf}(\cdots)_{l,m}
-{\rm Pf}(s_{l}^{+},s_{m}^{+},\cdots)_{l,m}{\rm Pf}(a^{(\mu)},d_{0},\cdots)_{l,m}}\nonumber\\
\fl{\qquad+{\rm Pf}(s_{l}^{+},a^{(\mu)},\cdots)_{l,m}{\rm Pf}(s_{m}^{+},d_{0},\cdots)_{l,m}
+{\rm Pf}(s_{m}^{+},a^{(\mu)},\cdots)_{l,m}{\rm Pf}(d_{0},s_{l}^{+},\cdots)_{l,m}}\nonumber\\
\fl{\quad=0.}
\label{bi_first}
\end{eqnarray}
Multiplying (\ref{bi_first}) by $1+ab$ yields precisely the second bilinear equation of (\ref{fullbilinear}).
The third equation follows by the same argument, with $g_{l,m}^{(\mu)}$ replaced by $h_{l,m}^{(\nu)}$.
 \end{proof}

\begin{thm}
\label{thm4.2}
Under the dependent-variable transformation
\begin{eqnarray}
u_{l,m}^{(\mu)}=\displaystyle\frac{g_{l,m}^{(\mu)}}{f_{l,m}},\quad v_{l,m}^{(\nu)}=\displaystyle\frac{h_{l,m}^{(\nu)}}{f_{l,m}},\quad \Gamma_{l,m}=\displaystyle\frac{f_{l-1,m+1}f_{l,m}}{f_{l-1,m}f_{l,m+1}},
\label{fulltrans}
\end{eqnarray}
the bilinear system (\ref{fullbilinear}) yields the following difference
equations:
\begin{eqnarray}
\left\{
\begin{array}{lll}
\Gamma_{l,m}=\displaystyle\frac{1+b^{2}\displaystyle\sum_{\mu\in I}\displaystyle\sum_{\nu\in J}c_{\mu\nu}u_{l-1,m}^{(\mu)}v_{l-1,m}^{(\nu)}}{1+b^{2}\displaystyle\sum_{\mu\in I}\sum_{\nu\in J}c_{\mu\nu}u_{l,m}^{(\mu)}v_{l,m}^{(\nu)}}\Gamma_{l,m-1},\\
\\
u_{l+1,m+1}^{(\mu)}+u_{l,m}^{(\mu)}=\displaystyle\frac{1+ab}{1-ab} \Gamma_{l+1,m}(u_{l+1,m}^{(\mu)}+u_{l,m+1}^{(\mu)}),\\
\\
v_{l+1,m+1}^{(\nu)}+v_{l,m}^{(\nu)}=\displaystyle\frac{1+ab}{1-ab} \Gamma_{l+1,m}(v_{l+1,m}^{(\nu)}+v_{l,m+1}^{(\nu)})
\end{array}
\right.
\label{fulldisper}
\end{eqnarray}
\end{thm}

\begin{proof}
Dividing the first equation of (\ref{fullbilinear}) by $f_{l,m}^{2}$ gives
\begin{eqnarray}
\displaystyle\frac{f_{l,m+1}f_{l,m-1}}{f_{l,m}^{2}}-1=b^{2}\sum_{\mu\in I}\sum_{\nu\in J}c_{\mu\nu}\frac{g_{l,m}^{(\mu)}}{f_{l,m}}\frac{h_{l,m}^{(\nu)}}{f_{l,m}}.\label{eq_4.32}
\end{eqnarray}
It follows that
\begin{eqnarray}
\displaystyle\frac{\Gamma_{l,m}}{\Gamma_{l,m-1}}=\displaystyle\frac{1+b^{2}\displaystyle\sum_{\mu\in I}\displaystyle\sum_{\nu\in J}c_{\mu\nu}\frac{g_{l-1,m}^{(\mu)}}{f_{l-1,m}}\frac{h_{l-1,m}^{(\nu)}}{f_{l-1,m}}}{1+b^{2}\displaystyle\sum_{\mu\in I}\displaystyle\sum_{\nu\in J}c_{\mu\nu}\frac{g_{l,m}^{(\mu)}}{f_{l,m}}\frac{h_{l,m}^{(\nu)}}{f_{l,m}}},
\end{eqnarray}
which gives the first equation of (\ref{fulldisper}) under the transformation
(\ref{fulltrans}).
Similarly, dividing the second equation of (\ref{fullbilinear}) by
$f_{l+1,m+1}f_{l,m}$ gives
\begin{eqnarray}
\displaystyle\frac{1+ab}{1-ab}\Gamma_{l+1,m}(u_{l+1,m}^{(\mu)}+u_{l,m+1}^{(\mu)})=u_{l+1,m+1}^{(\mu)}+u_{l,m}^{(\mu)},
\end{eqnarray}
which gives the second equation of (\ref{fulldisper}). The third equation follows
similarly.
\end{proof}
To transform the difference equations (\ref{fulldisper}) into a fully discrete system in the physical variables, we first recall the time evolution of the physical mesh points in the semi-discrete case. The parametric representation (\ref{semidisparametricMCSP}) gives
\begin{eqnarray}
\displaystyle\frac{dx_{l}}{dT}=-2(\log{f_{l}})_{TT}.\label{eq_4.37}
\end{eqnarray}
We introduce the fully discrete analogue of (\ref{eq_4.37}) as
\begin{eqnarray}
\displaystyle\frac{x_{l,m+1}-x_{l,m}}{2b}=\frac{1}{2b^{2}}\left(1-\frac{f_{l,m+1}f_{l,m-1}}{f_{l,m}^{2}}\right), 
\label{fulldis_x}
\end{eqnarray}
where $x_{l,m}$ denotes the $l$-th mesh point at the $m$-th time level. 
Indeed, setting $T=2bm$ and $f_{l,m}=f_l(T)$, one finds that, in the limit $b\to0$,
\begin{eqnarray*}
\displaystyle\frac{1}{2b^2}\left(1-\frac{f_{l,m+1}f_{l,m-1}}{f_{l,m}^{2}}\right)
&=&-2(\log f_l)_{TT}+O(b^2), 
\end{eqnarray*}
so that (\ref{fulldis_x}) reduces to (\ref{eq_4.37}) at leading order. 

The variables $u^{(\mu)}_{l,m}$ and $v^{(\nu)}_{l,m}$ are regarded as the field values at the physical mesh point $x_{l, m}$.

\begin{thm}
\label{thm4.3}
Let $(f_{l,m},g_{l,m}^{(\mu)},h_{l,m}^{(\nu)})$ satisfy the fully discrete bilinear system (\ref{fullbilinear}). Define $u_{l,m}^{(\mu)}$, $v_{l,m}^{(\nu)}$, and $\Gamma_{l,m}$ by (\ref{fulltrans}), and let $x_{l,m}$ be determined by (\ref{fulldis_x}). Then the variables $u_{l,m}^{(\mu)}$, $v_{l,m}^{(\nu)}$, and $x_{l,m}$ satisfy the following fully discrete MCSP system:
\begin{eqnarray}
\left\{
\begin{array}{ll}
\displaystyle\frac{u_{l,m+1}^{(\mu)}+u_{l-1,m}^{(\mu)}}{u_{l,m}^{(\mu)}+u_{l-1,m+1}^{(\mu)}}=\displaystyle\frac{1-b(x_{l-1,m+1}-x_{l-1,m})}{1-b(x_{l,m+1}-x_{l,m})} \displaystyle\frac{u_{l,m}^{(\mu)}+u_{l-1,m-1}^{(\mu)}}{u_{l,m-1}^{(\mu)}+u_{l-1,m}^{(\mu)}},\\
\displaystyle\frac{v_{l,m+1}^{(\nu)}+v_{l-1,m}^{(\nu)}}{v_{l,m}^{(\nu)}+v_{l-1,m+1}^{(\nu)}}=\displaystyle\frac{1-b(x_{l-1,m+1}-x_{l-1,m})}{1-b(x_{l,m+1}-x_{l,m})}\displaystyle\frac{v_{l,m}^{(\nu)}+v_{l-1,m-1}^{(\nu)}}{v_{l,m-1}^{(\nu)}+v_{l-1,m}^{(\nu)}},\\
\displaystyle\frac{x_{l,m+1}-x_{l,m}}{2b}=-\displaystyle\frac{1}{2}\sum_{\mu\in I}\sum_{\nu\in J}c_{\mu\nu}u_{l,m}^{(\mu)}v_{l,m}^{(\nu)}.
\end{array}
\right.
\label{fulldis}
\end{eqnarray}
All denominators appearing above are assumed to be nonzero.
\end{thm}

\begin{proof}
We first derive the evolution equation for the physical mesh points.
From (\ref{eq_4.32}) and the dependent-variable transformation (\ref{fulltrans}), we have
\begin{eqnarray}
\frac{1}{2b^2}\left(1-\frac{f_{l,m+1}f_{l,m-1}}{f_{l,m}^2}\right)=-\frac{1}{2}\sum_{\mu\in I}\sum_{\nu\in J}
c_{\mu\nu}u_{l,m}^{(\mu)}v_{l,m}^{(\nu)}.
\end{eqnarray}
Combining this relation with (\ref{fulldis_x}) gives
\begin{eqnarray}
\frac{x_{l,m+1}-x_{l,m}}{2b}=-\frac{1}{2}\sum_{\mu\in I}\sum_{\nu\in J}
c_{\mu\nu}u_{l,m}^{(\mu)}v_{l,m}^{(\nu)},
\end{eqnarray}
which is the third equation of (\ref{fulldis}).

We next derive the first equation. Using the first and second equations of (\ref{fulldisper}) together with the third equation of (\ref{fulldis}), we obtain
\begin{eqnarray}
\displaystyle\frac{u_{l,m+1}^{(\mu)}+u_{l-1,m}^{(\mu)}}{u_{l,m}^{(\mu)}+u_{l-1,m+1}^{(\mu)}}=\displaystyle\frac{1-b(x_{l-1,m+1}-x_{l-1,m})}{1-b(x_{l,m+1}-x_{l,m})} \displaystyle\frac{u_{l,m}^{(\mu)}+u_{l-1,m-1}^{(\mu)}}{u_{l,m-1}^{(\mu)}+u_{l-1,m}^{(\mu)}},
\end{eqnarray}
which is the first equation of (\ref{fulldis}). The second equation follows
in the same way.
\end{proof}
Once the initial mesh points $x_{l,0}$ are prescribed, the third equation of (\ref{fulldis}) determines the subsequent evolution of the mesh points. 
Therefore, for $m\geq1$, the mesh points are reconstructed as
\begin{eqnarray}
\fl{x_{l,m}=x_{l,0}-b\sum_{n=0}^{m-1}\sum_{\mu\in I}\sum_{\nu\in J}
c_{\mu\nu}u_{l,n}^{(\mu)}v_{l,n}^{(\nu)}=x_{l,0}+\frac{1}{b}\sum_{n=0}^{m-1}
\left(1-\frac{f_{l,n+1}f_{l,n-1}}{f_{l,n}^{2}}\right).}
\label{fulldisxsolution}
\end{eqnarray}
For negative $m$, the corresponding oriented sum is used.

\begin{rem}
\label{rem4.4}
Set $T=2bm$ and suppose that the lattice variables admit smooth interpolants
in $T$. As $b\rightarrow0$, we have
\begin{equation}
\begin{array}{l}
x_{l,m+1}-x_{l,m}=2bx_{l,T}+2b^2x_{l,TT}+O(b^3),\\
u^{(\mu)}_{l,m\pm1}=u^{(\mu)}_l\pm 2b\,u^{(\mu)}_{l,T}
+2b^2u^{(\mu)}_{l,TT}+O(b^3),
\end{array}
\end{equation}
and similarly for $v^{(\nu)}$. The third equation of (\ref{fulldis}) then
reduces to $x_{l,T}=-\frac{1}{2}\sum_{\mu,\nu}c_{\mu\nu}
u_l^{(\mu)}v_l^{(\nu)}$. The first two equations give
\begin{eqnarray*}
\frac{d}{dT}\left[
\frac{u_{l,T}^{(\mu)}-u_{l-1,T}^{(\mu)}}
{u_l^{(\mu)}+u_{l-1}^{(\mu)}}
-\frac{x_l-x_{l-1}}{2}\right]=0,\\
\frac{d}{dT}\left[
\frac{v_{l,T}^{(\nu)}-v_{l-1,T}^{(\nu)}}
{v_l^{(\nu)}+v_{l-1}^{(\nu)}}
-\frac{x_l-x_{l-1}}{2}\right]=0.
\end{eqnarray*}
Consequently, if the initial data satisfy the semi-discrete compatibility
conditions obtained by setting both expressions in brackets to zero, the
limit is the semi-discrete MCSP equation (\ref{semiMCSP}). Taking further
the limit $a\rightarrow0$ recovers the continuous MCSP equation (\ref{MCSP}).
\end{rem}

\begin{eg}
\label{eg4.5}
We now present explicit one- and two-soliton solutions of the fully discrete
MCSP equation (\ref{fulldis}). The fields are given by
$u_{l,m}^{(\mu)}=g_{l,m}^{(\mu)}/f_{l,m}$ and
$v_{l,m}^{(\nu)}=h_{l,m}^{(\nu)}/f_{l,m}$, while the mesh points are recovered
from (\ref{fulldisxsolution}) after the initial mesh $x_{l,0}$ is prescribed.\\

\noindent
One-soliton: Setting $M=2$ and $N=1$ in
(\ref{fulldis_solb})--(\ref{fulldis_sole}) gives the following $\tau$-functions:
\begin{eqnarray}
\fl{f_{l,m}={\rm Pf}(1,2,2^{\prime},1^{\prime})_{l,m}=1+\displaystyle\frac{(p_{1}^{2}-b^{2})(p_{2}^{2}-b^{2})}{4(p_{1}+p_{2})^{2}}\varphi_{1}\varphi_{2}\sum_{\mu\in I}\sum_{\nu\in J}c_{\mu\nu}a_{1}^{(\mu)}b_{2}^{(\nu)},}
\end{eqnarray}
\begin{eqnarray}
\fl{g^{(\mu)}_{l,m}={\rm Pf}(a^{(\mu)},d_{0},1,2,2^{\prime},1^{\prime})_{l,m}=a_{1}^{(\mu)}\varphi_{1},\qquad \mu\in I,}
\end{eqnarray}
\begin{eqnarray}
\fl{h^{(\nu)}_{l,m}={\rm Pf}(b^{(\nu)},d_{0},1,2,2^{\prime},1^{\prime})_{l,m}=b_{2}^{(\nu)}\varphi_{2},\qquad \nu\in J,}
\end{eqnarray}
where $\varphi_{i}=B_{i}\left(\displaystyle\frac{1+ap_{i}}{1-ap_{i}}\right)^{l}\left(\displaystyle\frac{p_{i}+b}{p_{i}-b}\right)^{m}$ and $i=1,2$.\\
\noindent
Two-soliton: Setting $M=4$ and $N=2$ in
(\ref{fulldis_solb})--(\ref{fulldis_sole}) gives the following $\tau$-functions. 
To express these $\tau$-functions in a compact form, we introduce the notation:
\begin{eqnarray*}
\fl C_{ij}&:=&\sum_{\mu\in I}\sum_{\nu\in J}c_{\mu\nu}a_i^{(\mu)}b_j^{(\nu)},\qquad
R_{ij}:=\frac{p_i-p_j}{p_i+p_j},\\
\fl \tilde{C}_{ij}&:=&\frac{(p_i^2-b^2)(p_j^2-b^2)}{4(p_i+p_j)^2}C_{ij},\qquad 
\Delta:=\frac{1}{16}\prod_{i=1}^{4}(p_i^2-b^2)
\left|
\begin{array}{rr}
\displaystyle\frac{C_{14}}{p_{1}^{2}-p_{4}^{2}} &
\displaystyle\frac{C_{24}}{p_{2}^{2}-p_{4}^{2}}\\[15pt]
\displaystyle\frac{C_{13}}{p_{1}^{2}-p_{3}^{2}} &
\displaystyle\frac{C_{23}}{p_{2}^{2}-p_{3}^{2}}
\end{array}
\right|.
\end{eqnarray*}
\begin{eqnarray}
\fl f_{l,m}&=&{\rm Pf}(1,2,3,4,4^{\prime},3^{\prime},2^{\prime},1^{\prime})_{l,m}\nonumber\\
\fl &=&1+\tilde{C}_{14}\varphi_1\varphi_4+\tilde{C}_{13}\varphi_1\varphi_3
+\tilde{C}_{24}\varphi_2\varphi_4+\tilde{C}_{23}\varphi_2\varphi_3\nonumber\\
\fl &&+\Delta\left(\prod_{1\leq i<j\leq4}R_{ij}\right)
\varphi_1\varphi_2\varphi_3\varphi_4.
\end{eqnarray}
\begin{eqnarray}
\fl g^{(\mu)}_{l,m}&=&{\rm Pf}(a^{(\mu)},d_{0},1,2,3,4,4^{\prime},3^{\prime},2^{\prime},1^{\prime})_{l,m}\nonumber\\
\fl &=&a_{1}^{(\mu)}\left(\varphi_1+R_{12}R_{14}\tilde{C}_{24}\varphi_1\varphi_2\varphi_4
+R_{12}R_{13}\tilde{C}_{23}\varphi_1\varphi_2\varphi_3\right)\nonumber\\
\fl &&+a_{2}^{(\mu)}\left(\varphi_2-R_{12}R_{24}\tilde{C}_{14}\varphi_1\varphi_2\varphi_4
-R_{12}R_{23}\tilde{C}_{13}\varphi_1\varphi_2\varphi_3\right),\qquad \mu\in I.
\end{eqnarray}
\begin{eqnarray}
\fl h^{(\nu)}_{l,m}&=&{\rm Pf}(b^{(\nu)},d_{0},1,2,3,4,4^{\prime},3^{\prime},2^{\prime},1^{\prime})_{l,m}\nonumber\\
\fl &=&b_{3}^{(\nu)}\left(\varphi_3-R_{23}R_{34}\tilde{C}_{24}\varphi_2\varphi_3\varphi_4
-R_{13}R_{34}\tilde{C}_{14}\varphi_1\varphi_3\varphi_4\right)\nonumber\\
\fl &&+b_{4}^{(\nu)}\left(\varphi_4+R_{24}R_{34}\tilde{C}_{23}\varphi_2\varphi_3\varphi_4
+R_{14}R_{34}\tilde{C}_{13}\varphi_1\varphi_3\varphi_4\right),\qquad \nu\in J.
\end{eqnarray}
where $\varphi_{i}=B_{i}\left(\displaystyle\frac{1+ap_{i}}{1-ap_{i}}\right)^{l}\left(\displaystyle\frac{p_{i}+b}{p_{i}-b}\right)^{m}$ and $i=1,2,3,4$.\\
\end{eg}

\begin{rem}
\label{rem4.6}
For real $a$ and $b$ and a real initial mesh, imposing $I=J$,
$c_{\mu\nu}=\delta_{\mu\nu}$, and $v^{(\mu)}=(u^{(\mu)})^*$ on the fully
discrete MCSP equation (\ref{fulldis}) yields the fully discrete CCSP equation
\begin{eqnarray}
\left\{
\begin{array}{ll}
\displaystyle\frac{u_{l,m}^{(\mu)}+u_{l-1,m-1}^{(\mu)}}{u_{l,m-1}^{(\mu)}+u_{l-1,m}^{(\mu)}}\displaystyle\frac{1-b(x_{l-1,m+1}-x_{l-1,m})}{1-b(x_{l,m+1}-x_{l,m})} =\displaystyle\frac{u_{l,m+1}^{(\mu)}+u_{l-1,m}^{(\mu)}}{u_{l,m}^{(\mu)}+u_{l-1,m+1}^{(\mu)}},\\
\displaystyle\frac{x_{l,m+1}-x_{l,m}}{2b}=-\displaystyle\frac{1}{2}\displaystyle\sum_{\mu\in I}|u_{l,m}^{(\mu)}|^{2}.
\end{array}
\right.
\label{CCSPfulldis}
\end{eqnarray}
where $(u^{(\mu)})^*$ denotes the complex conjugate of $u^{(\mu)}$.
Its $N$-soliton solutions are obtained from the Pfaffian solutions (\ref{fulldis_solb})--(\ref{fulldis_sole}) under the reduction in Remark
\ref{rem2.4}. Numerical examples are presented in Section \ref{sec_nc}.
\end{rem}
\end{section}

\begin{section}{Numerical simulations of the MCSP and CCSP equations}
\label{sec_nc}
In this section, we numerically test the fully discrete MCSP and CCSP
equations. 
The continuous exact solutions (\ref{con_ex1})--(\ref{con_ex4}) 
are used to generate the initial data for the fully discrete schemes. 
In each figure, the blue dotted curve represents the numerical solution, 
the yellow solid curve represents the exact solution, 
and the red dots indicate the distribution of the adaptive mesh points. 
In all simulations, the number of mesh intervals is $N_x=8000$, the
computational domain width is $D=80$, and the time step is $\Delta t=0.001$.
Since $2a$ and $2b$ denote the spatial and temporal lattice intervals,
respectively, the numerical parameters satisfy
\begin{eqnarray}
\Delta X = 2a=\frac{D}{N_x}, \qquad \Delta t = 2b,
\end{eqnarray}
where $a=0.005$ and $b=0.0005$. 
The computational grid is defined by 
\begin{eqnarray}
X_{l}=X_{0}+2al,\qquad X_{0}=-\frac{D}{2},\qquad T_{m}=2bm=m\Delta t,
\end{eqnarray}
where $l=0,\ldots,N_x,$ and $m=0,1,\ldots$.
The initial mesh points $x_{l,0}$ are determined from the hodograph
transformation
\begin{eqnarray}
x_{l,0}=X_l-2\frac{f_T(X_{l},0)}{f(X_{l},0)}.
\end{eqnarray}
Because the first two equations of the fully discrete system
(\ref{fulldis}) involve the previous time level $m-1$, the field variables are initialized on two time levels. 
Specifically, their values at $m=0$ and $m=1$ are obtained from the
corresponding continuous exact solutions at $T=0$ and $T=\Delta t$,
respectively:
\begin{eqnarray}
u_{l,0}^{(\mu)}=u_{\rm ex}^{(\mu)}(X_l,0),\qquad
v_{l,0}^{(\nu)}=v_{\rm ex}^{(\nu)}(X_l,0),
\nonumber\\
u_{l,1}^{(\mu)}=u_{\rm ex}^{(\mu)}(X_l,\Delta t),\qquad
v_{l,1}^{(\nu)}=v_{\rm ex}^{(\nu)}(X_l,\Delta t).
\end{eqnarray}
Here $u_{\rm ex}^{(\mu)}(X,T)$ and $v_{\rm ex}^{(\nu)}(X,T)$ denote the
continuous exact solutions given in (\ref{con_ex1})--(\ref{con_ex4}).
Once $x_{l,0}$,
$u_{l,0}^{(\mu)}$, and $v_{l,0}^{(\nu)}$ have been specified, 
the mesh points at $m=1$ can be determined 
from the discrete mesh equation. Thus we set
\begin{eqnarray}
x_{l,1}=x_{l,0}-\frac{\Delta t}{2}\sum_{\mu\in I}\sum_{\nu\in J}c_{\mu\nu}u_{l,0}^{(\mu)}v_{l,0}^{(\nu)} .
\end{eqnarray}
This choice makes the initial mesh consistent with the fully discrete
time evolution

After initializing the field variables at time levels $m=0$ and $m=1$, at each time step we
first update the mesh points at the new time level using the third equation
of (\ref{fulldis}) and then compute $u^{(\mu)}$ and $v^{(\nu)}$ from its
first two equations.
To make their linear structure explicit, let $w$ denote either
$u^{(\mu)}$ or $v^{(\nu)}$, and define the known coefficient
\begin{eqnarray}
\fl A_{l,m}[w]:=
\frac{1-b(x_{l-1,m+1}-x_{l-1,m})}
{1-b(x_{l,m+1}-x_{l,m})}
\frac{w_{l,m}+w_{l-1,m-1}}
{w_{l,m-1}+w_{l-1,m}}.
\nonumber
\end{eqnarray}
After the mesh points at time level $m+1$ have been updated, all quantities
in $A_{l,m}[w]$ are known. The corresponding equation in
(\ref{fulldis}) is then equivalent to
\begin{eqnarray}
\fl w_{l,m+1}-A_{l,m}[w]w_{l-1,m+1}
=A_{l,m}[w]w_{l,m}-w_{l-1,m}.
\nonumber
\end{eqnarray}
Thus, for each $\mu\in I$ and $\nu\in J$, the first and second equations of
(\ref{fulldis}) yield bidiagonal linear systems for $u^{(\mu)}$ and
$v^{(\nu)}$, respectively, at time level $m+1$.
Under the right-boundary conditions used here, the right-to-left recursive
calculation is precisely backward substitution for these systems, so that
there is no need to form their inverse matrices explicitly. Because the
computational domain is sufficiently large, we impose
\begin{eqnarray}
u^{(\mu)}_{N_x,m}=v^{(\nu)}_{N_x,m}=0,\qquad m \geq 2.
\end{eqnarray}
The values at the remaining grid points are then determined successively
from $l=N_x-1$ to $l=0$, and this procedure is repeated until the final
time is reached.
For a general coupling matrix $c_{\mu\nu}$, with $K=|I|$ and $L=|J|$, 
the evaluation of the mesh equation requires $O(KL N_x)$ operations per time step. 
The subsequent updates of the $K+L$ field components consist of $K+L$ bidiagonal backward substitutions, each requiring $O(N_x)$ operations. 
Hence, the total computational cost per time step is $O\bigl((KL+K+L)N_x\bigr)$.
Since the numbers of components $K$ and $L$ are fixed, this reduces to $O(N_x)$. 
Thus, the computational cost scales linearly with the number of spatial mesh intervals $N_x$. No explicit matrix inversion is required.

Under periodic boundary conditions, the coupling between the two endpoints
results in cyclic bidiagonal linear systems. The field values at the new time level can
then be calculated simultaneously, provided that the resulting coefficient
matrices are nonsingular and the denominators in (\ref{fulldis}) do not
vanish. Thus the same formulation can also be adapted to periodic boundary
conditions.

To quantify the error in each displayed profile, let $q_{l,m}$ denote the corresponding
real-valued quantity. Thus, $q_{l,m}$ is $u_{l,m}^{(\mu)}$ or $v_{l,m}^{(\nu)}$
for a real-valued solution, and is $|u_{l,m}^{(\mu)}|$, ${\rm Re}(u_{l,m}^{(\mu)})$,
or ${\rm Im}(u_{l,m}^{(\mu)})$ for a complex-valued solution. 
The corresponding exact values are evaluated at 
the same computational grid points:
\begin{eqnarray}
q^{e}_{l,m}:=q_{\rm{ex}}(X_l,T_m).
\end{eqnarray}
To examine the accuracy near the peak of the displayed profile, define the set of grid indices
\begin{eqnarray}
\mathcal{I}_m(q):=\left\{l\mid 0\leq l\leq N_x,\ 
|q^e_{l,m}|\geq0.9\max_{0\leq j\leq N_x}|q^e_{j,m}|\right\}.
\end{eqnarray}
At each time level $m$, we define the maximum relative error in the displayed
profile by
\begin{eqnarray}
{\rm err}_{m}(q):=\max_{l\in\mathcal{I}_m(q)}
\left|\frac{q_{l,m}^{e}-q_{l,m}^{n}}{q_{l,m}^{e}}\right|.
\end{eqnarray}
The value reported in each figure caption is the maximum over all time
levels used in the simulation:
\begin{eqnarray}
{\rm maxerr}(q):=\max_{2\leq m\leq N_{t}}{\rm err}_{m}(q),
\end{eqnarray}
where $N_{t}$ denotes the final time level. Here $q^n_{l,m}$ and
$q^e_{l,m}$ denote the numerical and exact values of the displayed
profile, respectively.\\

\noindent 
Two-component SP (2-SP) equation:\\
\noindent
Setting $I=J=\{1\}$ and $c_{11}=1$ in (\ref{MCSP}) reduces the MCSP equation
to the 2-SP equation
\begin{eqnarray}
\left\{
\begin{array}{ll}
u^{(1)}_{xt}=u^{(1)}+\displaystyle\frac{1}{2}\left(u^{(1)}v^{(1)}u^{(1)}_{x}\right)_{x},\\
\\
v^{(1)}_{xt}=v^{(1)}+\displaystyle\frac{1}{2}\left(u^{(1)}v^{(1)}v^{(1)}_{x}\right)_{x}.
\end{array}
\right.
\label{2SP}
\end{eqnarray}
The corresponding fully discrete analogue is
\begin{eqnarray}
\left\{
\begin{array}{lll}
\displaystyle\frac{u^{(1)}_{l,m+1}+u^{(1)}_{l-1,m}}{u^{(1)}_{l,m}+u^{(1)}_{l-1,m+1}}=\displaystyle\frac{1-b(x_{l-1,m+1}-x_{l-1,m})}{1-b(x_{l,m+1}-x_{l,m})} \displaystyle\frac{u^{(1)}_{l,m}+u^{(1)}_{l-1,m-1}}{u^{(1)}_{l,m-1}+u^{(1)}_{l-1,m}},\\
\displaystyle\frac{v^{(1)}_{l,m+1}+v^{(1)}_{l-1,m}}{v^{(1)}_{l,m}+v^{(1)}_{l-1,m+1}}=\displaystyle\frac{1-b(x_{l-1,m+1}-x_{l-1,m})}{1-b(x_{l,m+1}-x_{l,m})} \displaystyle\frac{v^{(1)}_{l,m}+v^{(1)}_{l-1,m-1}}{v^{(1)}_{l,m-1}+v^{(1)}_{l-1,m}},\\
\displaystyle\frac{x_{l,m+1}-x_{l,m}}{2b}=-\displaystyle\frac{1}{2}u^{(1)}_{l,m}v^{(1)}_{l,m}.
\end{array}
\right.
\label{2SPfulldis}
\end{eqnarray}
We numerically test (\ref{2SPfulldis}) using the one- and two-soliton
solutions obtained from (\ref{con_ex1})--(\ref{con_ex4}) by setting
$I=J=\{1\}$ and $c_{11}=1$. 
The continuous one- and two-soliton solutions used 
to generate the initial data are given below.

\noindent
One-soliton:
 \begin{eqnarray}
\fl{u^{(1)}=\frac{g^{(1)}}{f},\quad v^{(1)}=\frac{h^{(1)}}{f},\quad x=X-2(\log{f})_{T},\quad t=T,}\nonumber\\
\fl{f=1+\displaystyle\frac{1}{4}\left(\displaystyle\frac{p_{1}p_{2}}{p_{1}+p_{2}}\right)^{2}\varphi_{1}\varphi_{2}a^{(1)}_{1}b^{(1)}_{2},\quad g^{(1)}=a^{(1)}_{1}\varphi_{1},\quad h^{(1)}=b^{(1)}_{2}\varphi_{2},}\label{eq_5.4}
\end{eqnarray}
where $\varphi_{i}=B_{i}e^{p_{i}X+p_{i}^{-1}T}$ and $i=1,2$.

\noindent
Two-soliton:
Using the notation in (\ref{con_compact_coefficients}) with
$C_{ij}=a_i^{(1)}b_j^{(1)}$, 
the two-soliton solution can be written as
 \begin{eqnarray}
 \fl u^{(1)}&=&\frac{g^{(1)}}{f},\quad v^{(1)}=\frac{h^{(1)}}{f},\nonumber\\
\fl x&=&X-2(\log{f})_{T},\quad t=T,\nonumber\\
\fl f&=&1+\tilde{C}_{14}\varphi_1\varphi_4+\tilde{C}_{13}\varphi_1\varphi_3
+\tilde{C}_{24}\varphi_2\varphi_4+\tilde{C}_{23}\varphi_2\varphi_3\nonumber\\
\fl &&+\Delta\left(\prod_{1\leq i<j\leq4}R_{ij}\right)
\varphi_1\varphi_2\varphi_3\varphi_4,\nonumber\\
\fl g^{(1)}&=&a_{1}^{(1)}\left(\varphi_1+R_{12}R_{14}\tilde{C}_{24}\varphi_1\varphi_2\varphi_4
+R_{12}R_{13}\tilde{C}_{23}\varphi_1\varphi_2\varphi_3\right)\nonumber\\
\fl &&+a_{2}^{(1)}\left(\varphi_2-R_{12}R_{24}\tilde{C}_{14}\varphi_1\varphi_2\varphi_4
-R_{12}R_{23}\tilde{C}_{13}\varphi_1\varphi_2\varphi_3\right),\nonumber\\
\fl h^{(1)}&=&b_{3}^{(1)}\left(\varphi_3-R_{23}R_{34}\tilde{C}_{24}\varphi_2\varphi_3\varphi_4
-R_{13}R_{34}\tilde{C}_{14}\varphi_1\varphi_3\varphi_4\right)\nonumber\\
\fl &&+b_{4}^{(1)}\left(\varphi_4+R_{24}R_{34}\tilde{C}_{23}\varphi_2\varphi_3\varphi_4
+R_{14}R_{34}\tilde{C}_{13}\varphi_1\varphi_3\varphi_4\right),
\label{eq_5.5}
\end{eqnarray}
where $\varphi_{i}=B_{i}e^{p_{i}X+p_{i}^{-1}T}$ and $i=1,2,3,4$.

Figures \ref{2SP_1_u1} and \ref{2SP_1_v1} show the numerical
$u^{(1)}$- and $v^{(1)}$-profiles for the one-soliton solution with
$p_{1}=0.95$, $p_{2}=1.1$, $a^{(1)}_{1}=0.5$, $b^{(1)}_{2}=20$, and
$B_{1}=B_{2}={\rm exp}(-5)$. Figures \ref{2SP_2_u1} and \ref{2SP_2_v1} show
the corresponding profiles for the two-soliton solution with $p_{1}=0.95$,
$p_{2}=1$, $p_{3}=1.1$, $p_{4}=1.2$, $a^{(1)}_{1}=0.5$, $a^{(1)}_{2}=1$,
$b^{(1)}_{3}=20$, $b^{(1)}_{4}=40$, and
$B_{1}=B_{2}=B_{3}=B_{4}={\rm exp}(-5)$. In both cases, the solitons propagate
to the left along the $x$-axis.\\


\begin{figure}[htbp]
 \begin{tabular}{cc}
      \begin{minipage}[t]{0.47\hsize}
       \centering
        \includegraphics[keepaspectratio, scale=0.33]{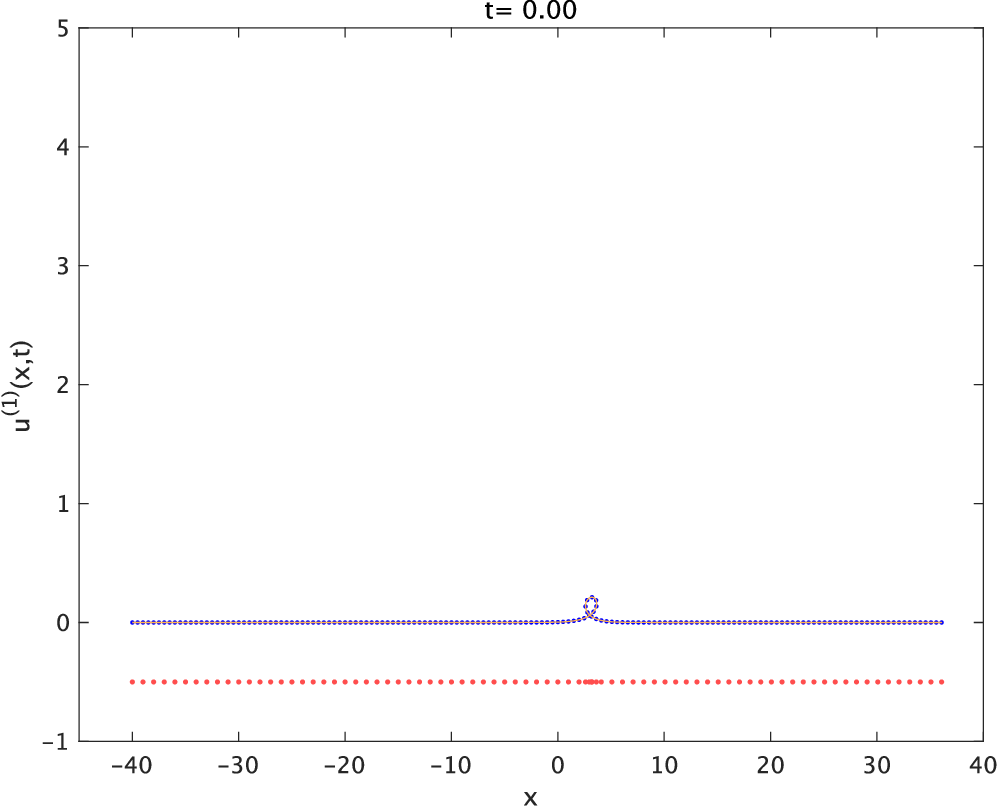}
      \end{minipage} &
      \begin{minipage}[t]{0.47\hsize}
        \centering
        \includegraphics[keepaspectratio, scale=0.33]{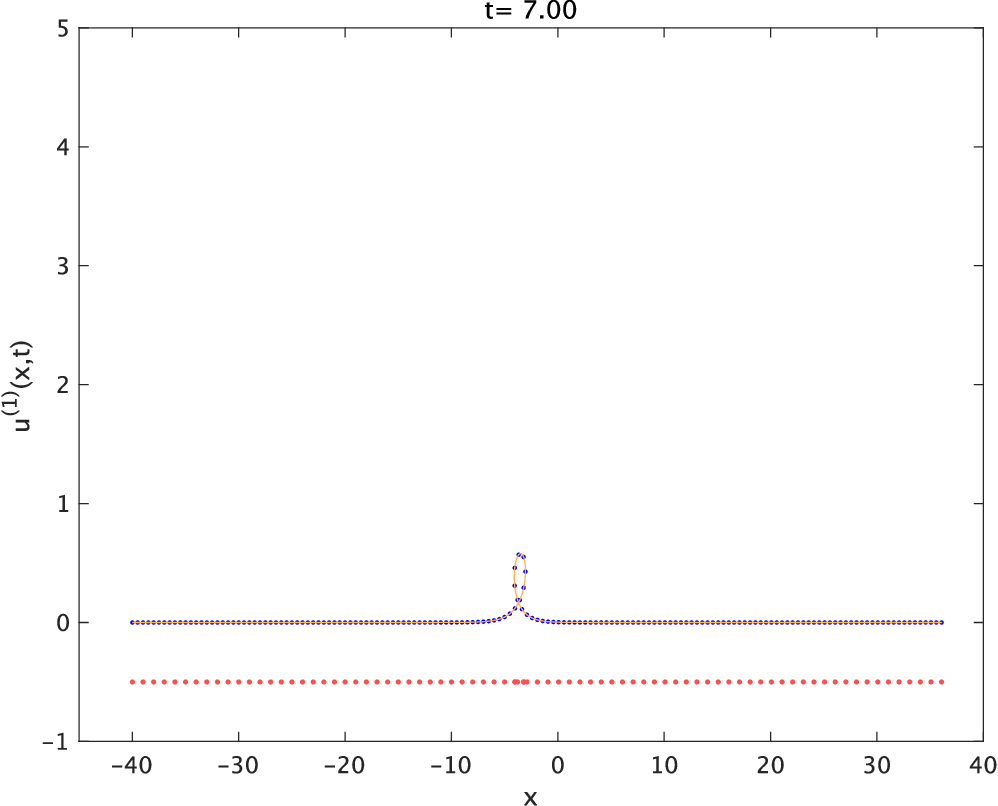}
      \end{minipage}\\ 
  
      \begin{minipage}[t]{0.47\hsize}
        \centering
        \includegraphics[keepaspectratio, scale=0.33]{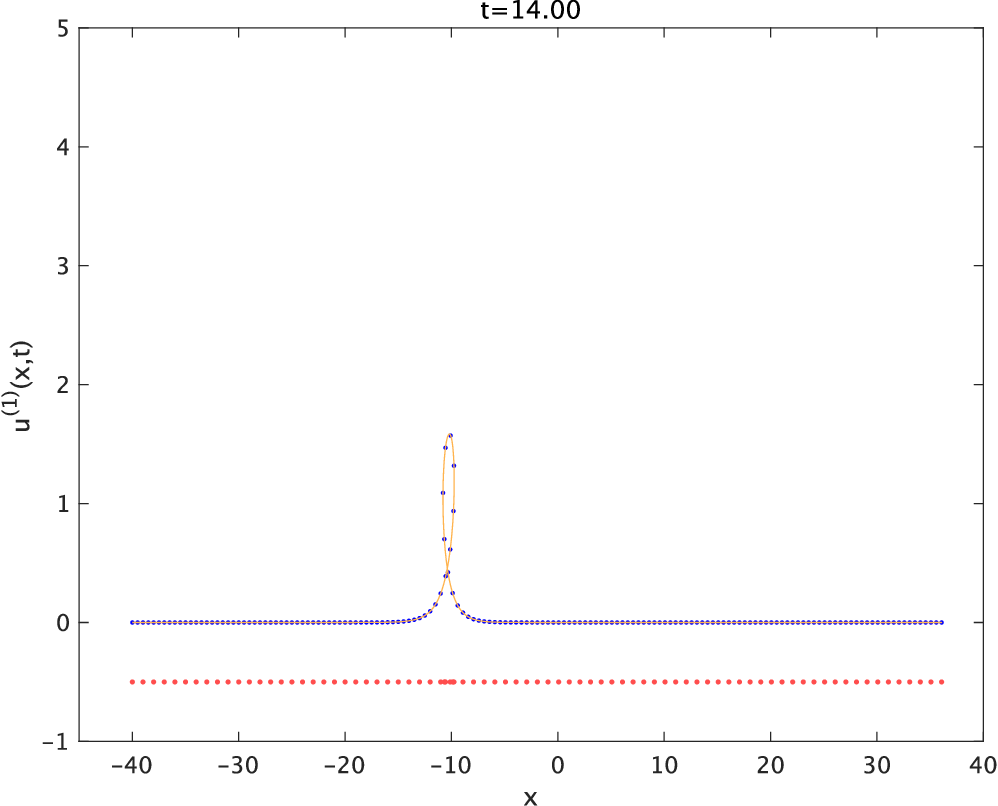}
      \end{minipage} &
      \begin{minipage}[t]{0.47\hsize}
        \centering
        \includegraphics[keepaspectratio, scale=0.33]{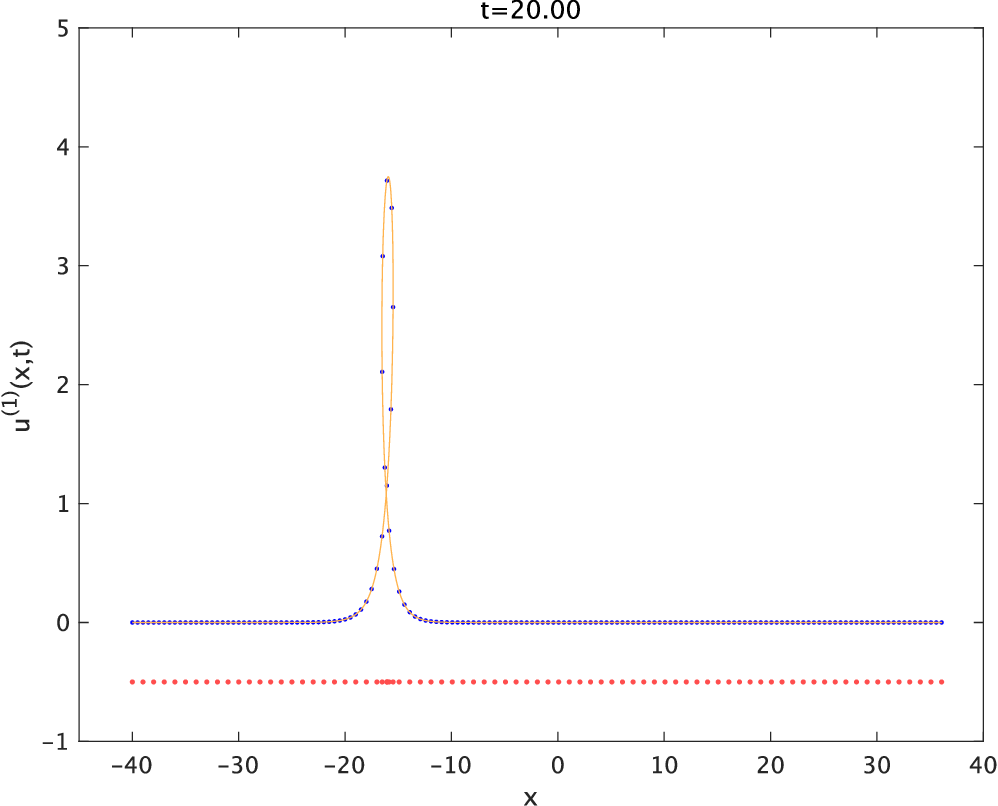}
      \end{minipage} 
       \end{tabular}
     \caption{Numerical simulation of the $u^{(1)}$-profile of 
     the one-soliton solution of the 2-SP equation. ${\rm maxerr}(u^{(1)})=9.33\times 10^{-4}$.}
              \label{2SP_1_u1}
  \end{figure}


\begin{figure}[htbp]
 \begin{tabular}{cc}
      \begin{minipage}[t]{0.47\hsize}
       \centering
        \includegraphics[keepaspectratio, scale=0.33]{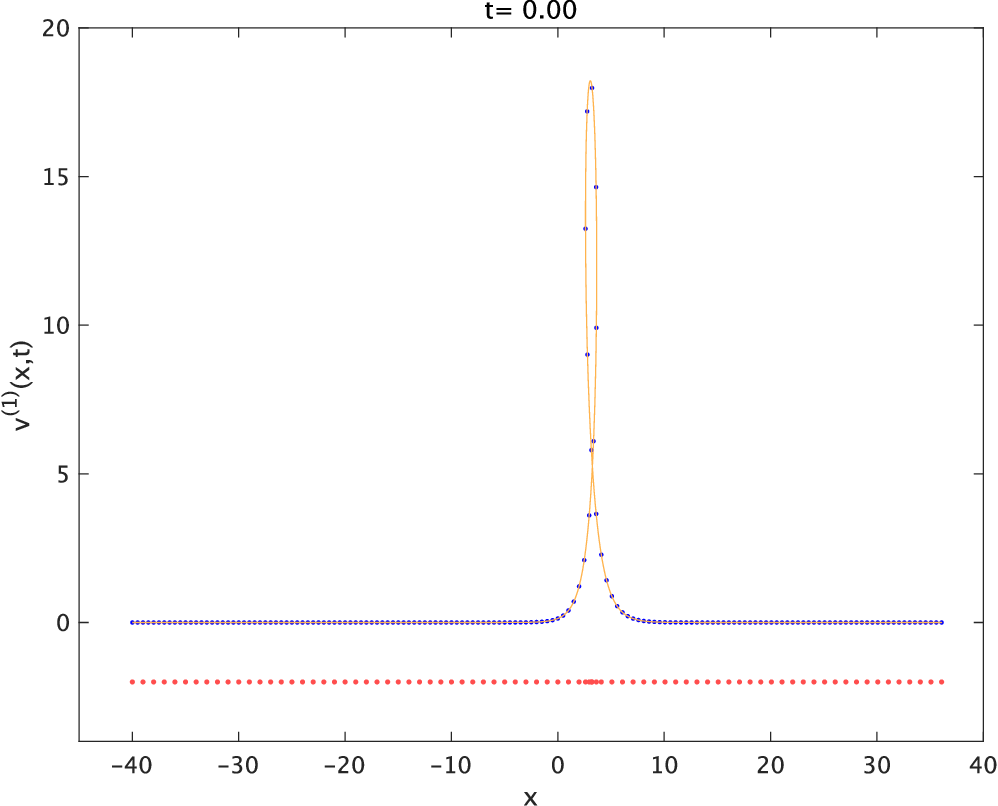}
      \end{minipage} &
      \begin{minipage}[t]{0.47\hsize}
        \centering
        \includegraphics[keepaspectratio, scale=0.33]{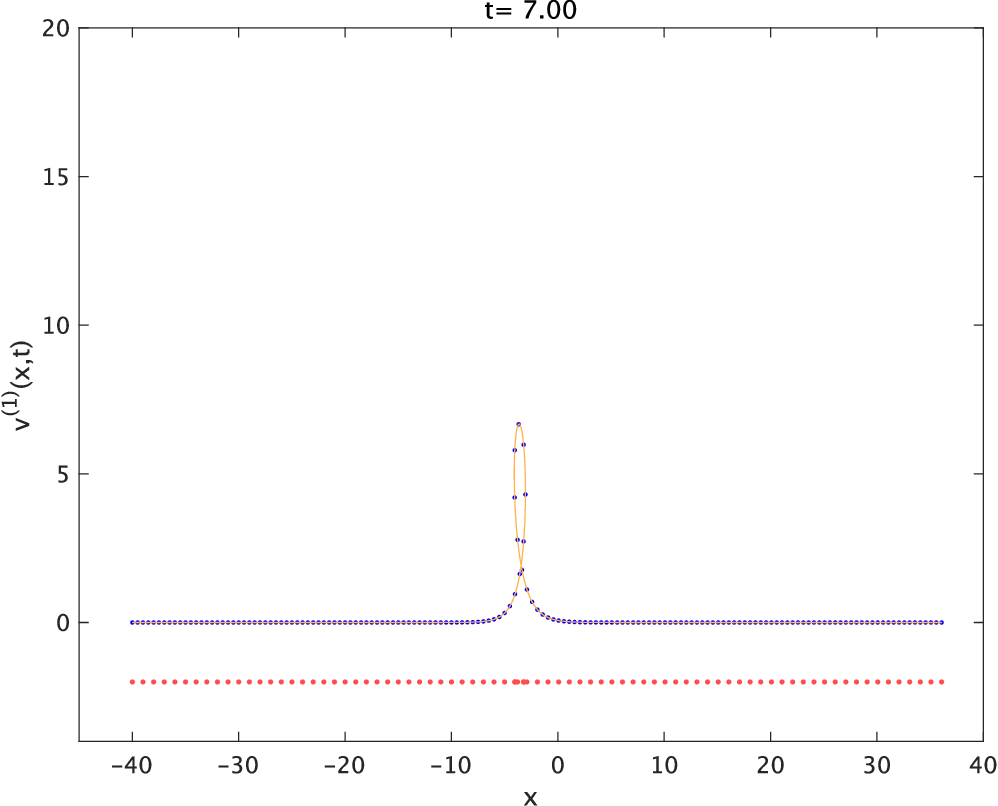}
      \end{minipage}\\ 
  
      \begin{minipage}[t]{0.47\hsize}
        \centering
        \includegraphics[keepaspectratio, scale=0.33]{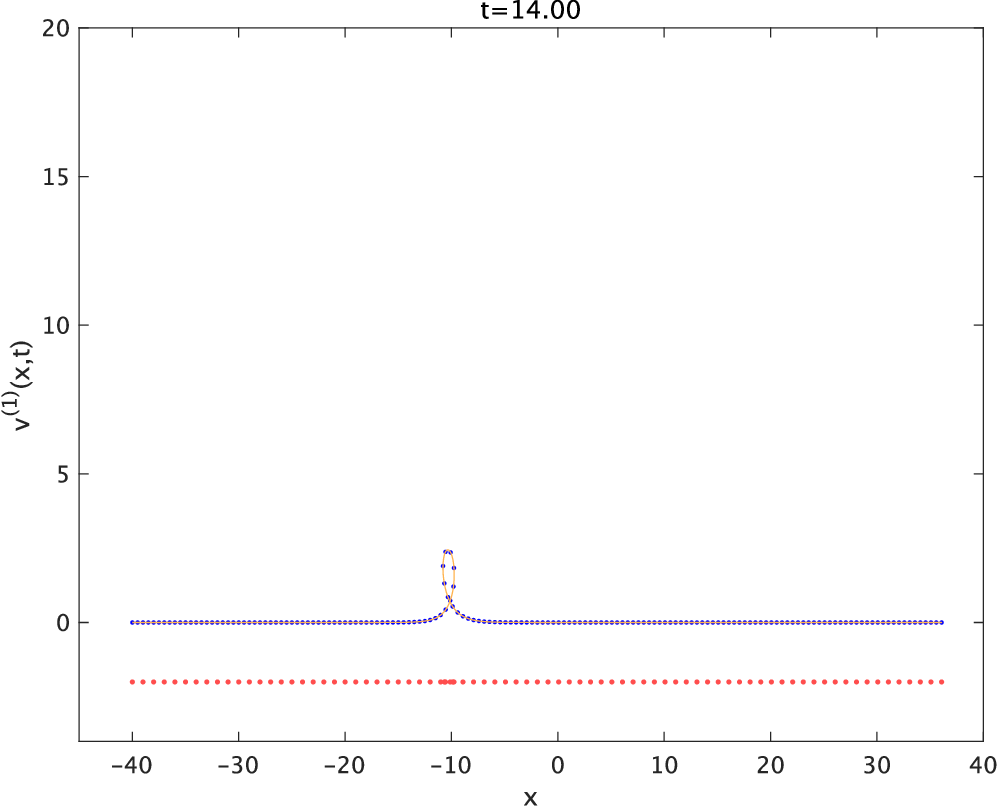}
      \end{minipage} &
      \begin{minipage}[t]{0.47\hsize}
        \centering
        \includegraphics[keepaspectratio, scale=0.33]{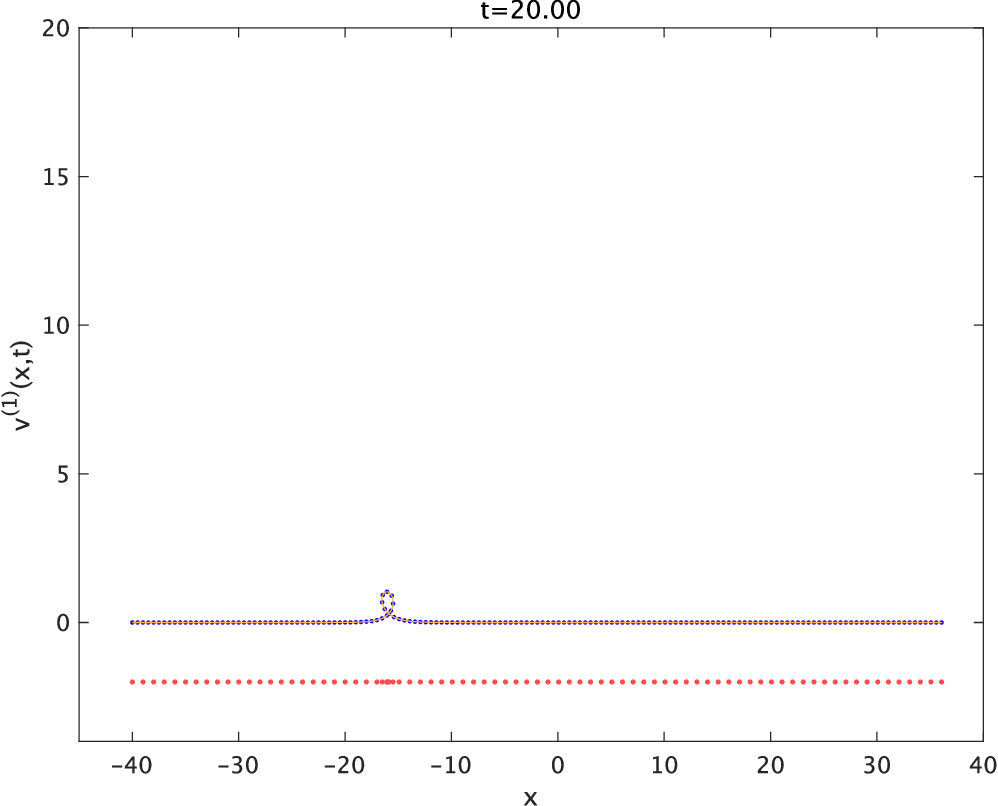}
      \end{minipage} 
       \end{tabular}
     \caption{Numerical simulation of the $v^{(1)}$-profile of 
     the one-soliton solution of the 2-SP equation. ${\rm maxerr}(v^{(1)})=9.80\times 10^{-4}$.}
              \label{2SP_1_v1}
  \end{figure}


\begin{figure}[htbp]
 \begin{tabular}{cc}
      \begin{minipage}[t]{0.47\hsize}
       \centering
        \includegraphics[keepaspectratio, scale=0.33]{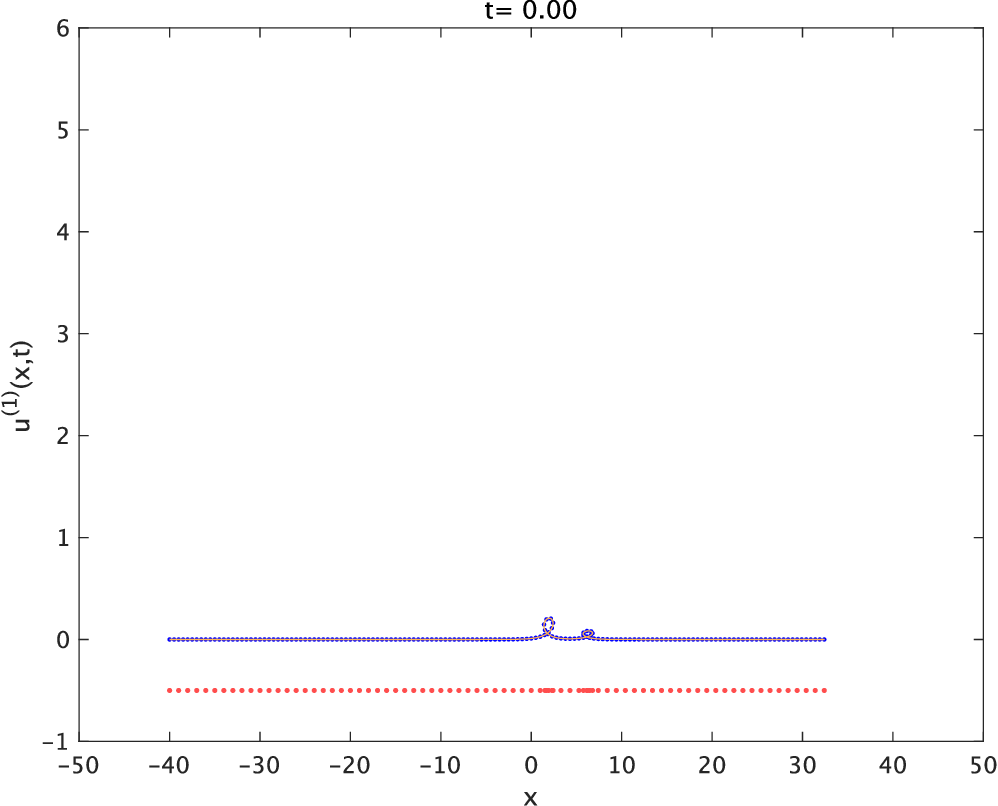}
      \end{minipage} &
      \begin{minipage}[t]{0.47\hsize}
        \centering
        \includegraphics[keepaspectratio, scale=0.33]{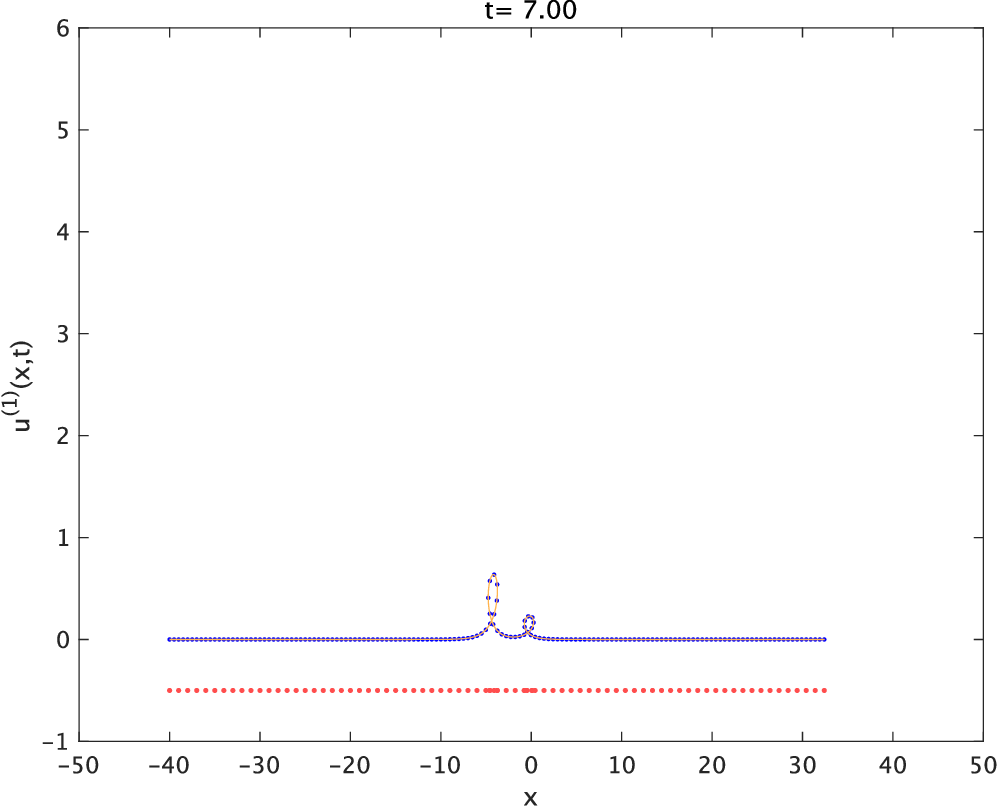}
      \end{minipage}\\ 
  
      \begin{minipage}[t]{0.47\hsize}
        \centering
        \includegraphics[keepaspectratio, scale=0.33]{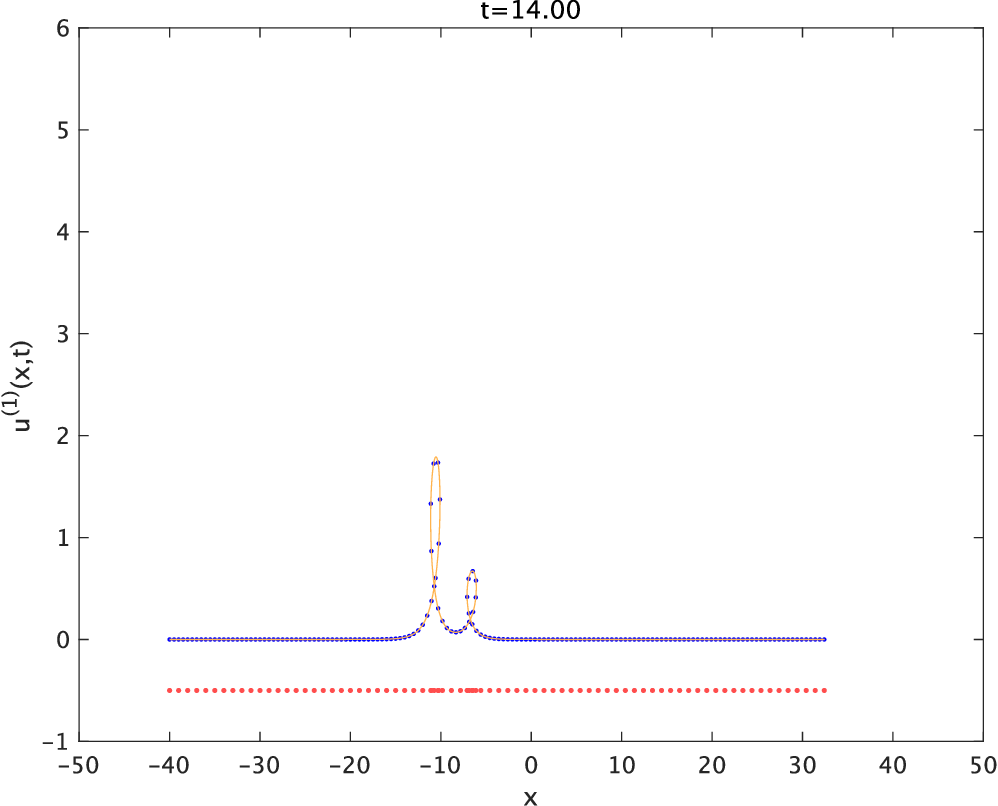}
      \end{minipage} &
      \begin{minipage}[t]{0.47\hsize}
        \centering
        \includegraphics[keepaspectratio, scale=0.33]{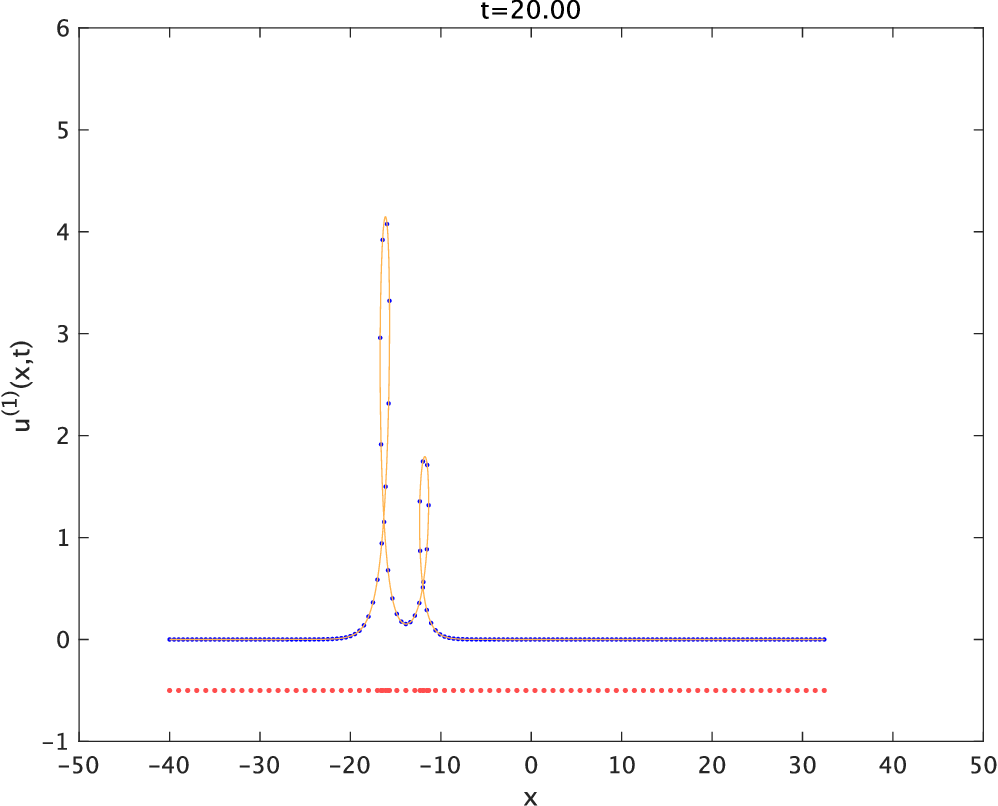}
      \end{minipage} 
       \end{tabular}
     \caption{Numerical simulation of the $u^{(1)}$-profile of 
     the two-soliton solution of the 2-SP equation. ${\rm maxerr}(u^{(1)})=1.00\times 10^{-3}$.}
              \label{2SP_2_u1}
  \end{figure}


\begin{figure}[htbp]
 \begin{tabular}{cc}
      \begin{minipage}[t]{0.47\hsize}
       \centering
        \includegraphics[keepaspectratio, scale=0.33]{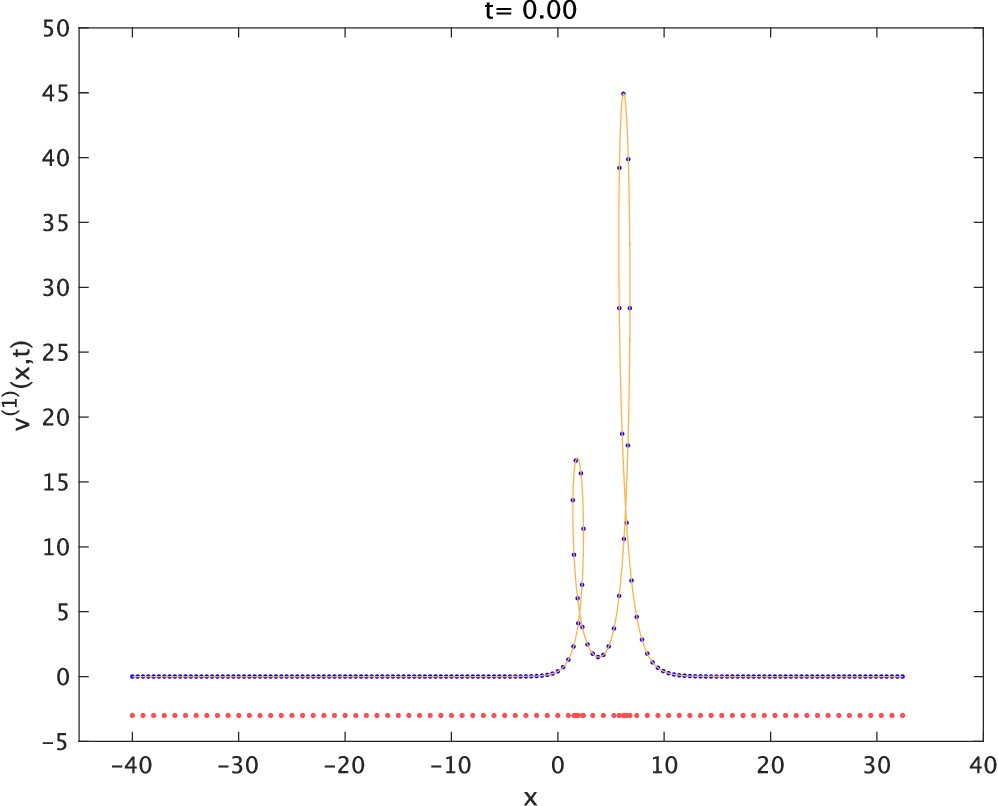}
      \end{minipage} &
      \begin{minipage}[t]{0.47\hsize}
        \centering
        \includegraphics[keepaspectratio, scale=0.33]{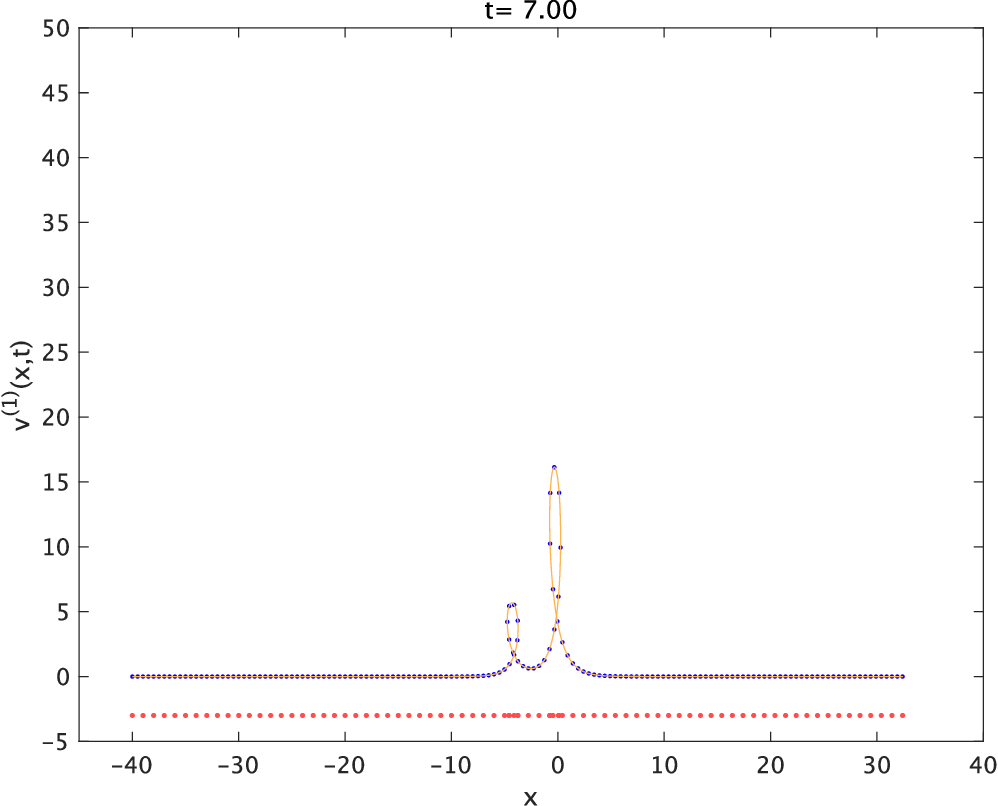}
      \end{minipage}\\ 
  
      \begin{minipage}[t]{0.47\hsize}
        \centering
        \includegraphics[keepaspectratio, scale=0.33]{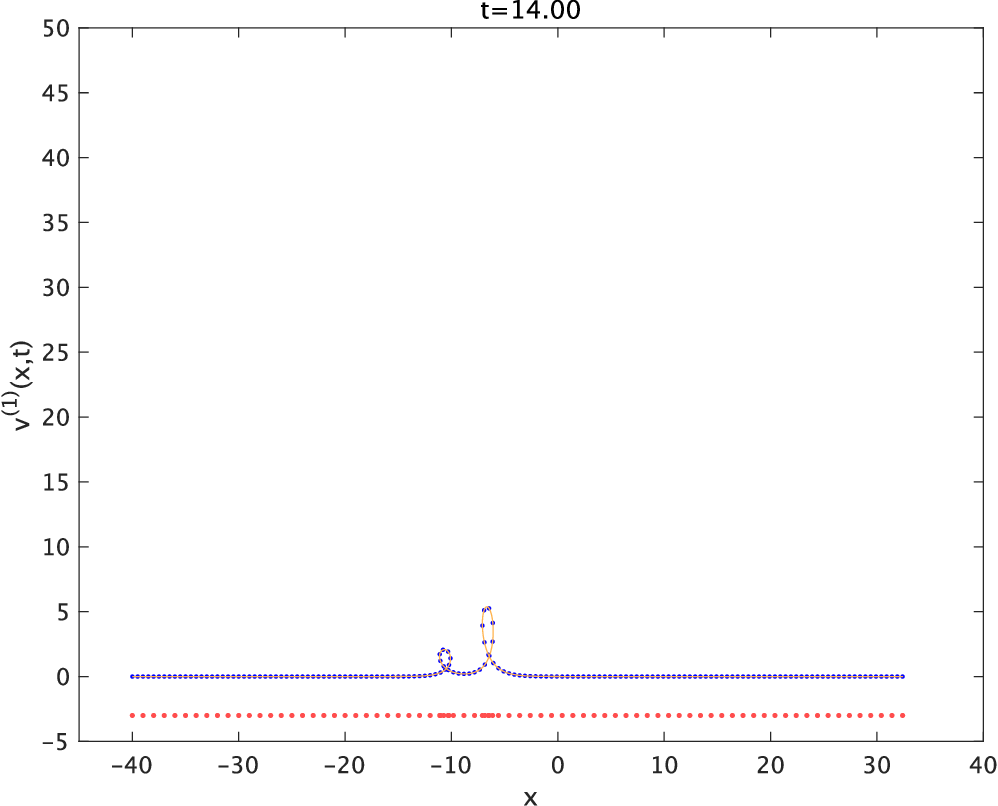}
      \end{minipage} &
      \begin{minipage}[t]{0.47\hsize}
        \centering
        \includegraphics[keepaspectratio, scale=0.33]{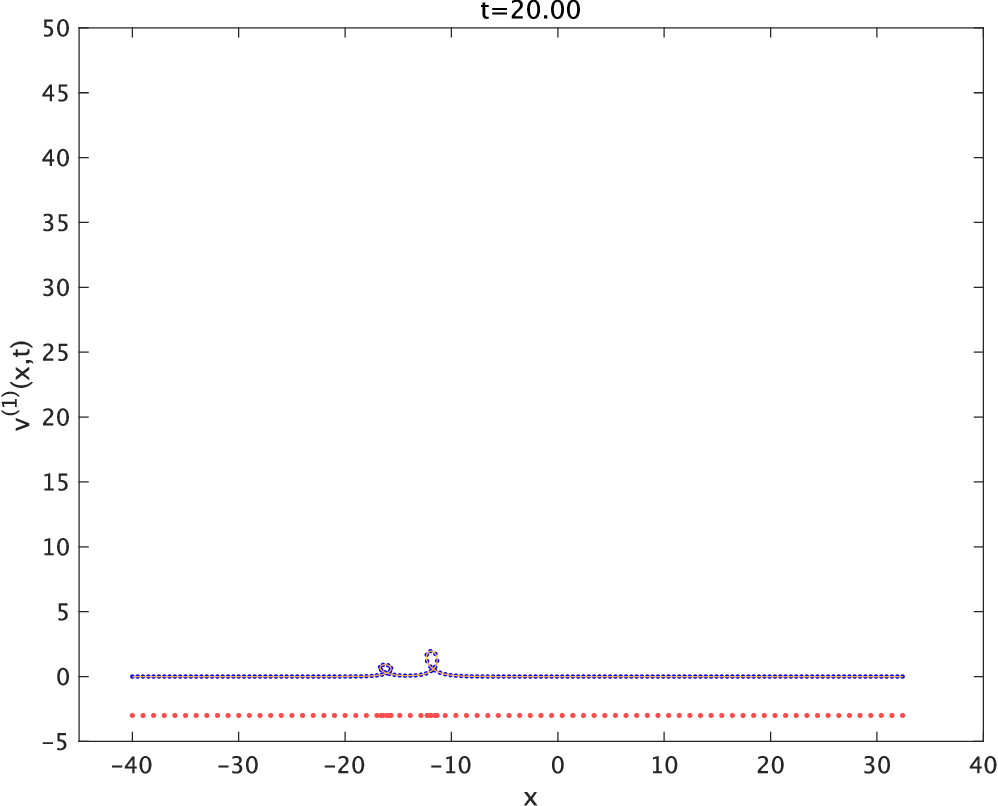}
      \end{minipage} 
       \end{tabular}
     \caption{Numerical simulation of the $v^{(1)}$-profile of 
     the two-soliton solution of the 2-SP equation. ${\rm maxerr}(v^{(1)})=3.70\times 10^{-4}$.}
              \label{2SP_2_v1}
  \end{figure}

\noindent 
One-component CSP (1-CSP) equation:\\
\noindent
Imposing the complex conjugacy condition $v^{(1)}=(u^{(1)})^*$ reduces the 2-SP equation (\ref{2SP}) to the CSP equation
\begin{eqnarray}
u^{(1)}_{xt}=u^{(1)}+\frac{1}{2}(|u^{(1)}|^{2}u^{(1)}_{x})_{x}, \label{CSP}
\end{eqnarray}
where $(u^{(1)})^*$ denotes the complex conjugate of $u^{(1)}$.
Imposing the same constraint on (\ref{2SPfulldis}) gives the fully discrete
CSP equation
\begin{eqnarray}
\left\{
\begin{array}{lll}
\displaystyle\frac{u^{(1)}_{l,m+1}+u^{(1)}_{l-1,m}}{u^{(1)}_{l,m}+u^{(1)}_{l-1,m+1}}=\displaystyle\frac{1-b(x_{l-1,m+1}-x_{l-1,m})}{1-b(x_{l,m+1}-x_{l,m})} \displaystyle\frac{u^{(1)}_{l,m}+u^{(1)}_{l-1,m-1}}{u^{(1)}_{l,m-1}+u^{(1)}_{l-1,m}},\\
\displaystyle\frac{x_{l,m+1}-x_{l,m}}{2b}=-\displaystyle\frac{1}{2}|u^{(1)}_{l,m}|^{2}.
\end{array}
\right.
\label{CSPfulldis}
\end{eqnarray}
We next test the fully discrete CSP equation (\ref{CSPfulldis}) 
using initial data generated from the following 
continuous one- and two-soliton solutions.


\begin{figure}[htbp]
 \begin{tabular}{cc}
      \begin{minipage}[t]{0.47\hsize}
       \centering
        \includegraphics[keepaspectratio, scale=0.33]{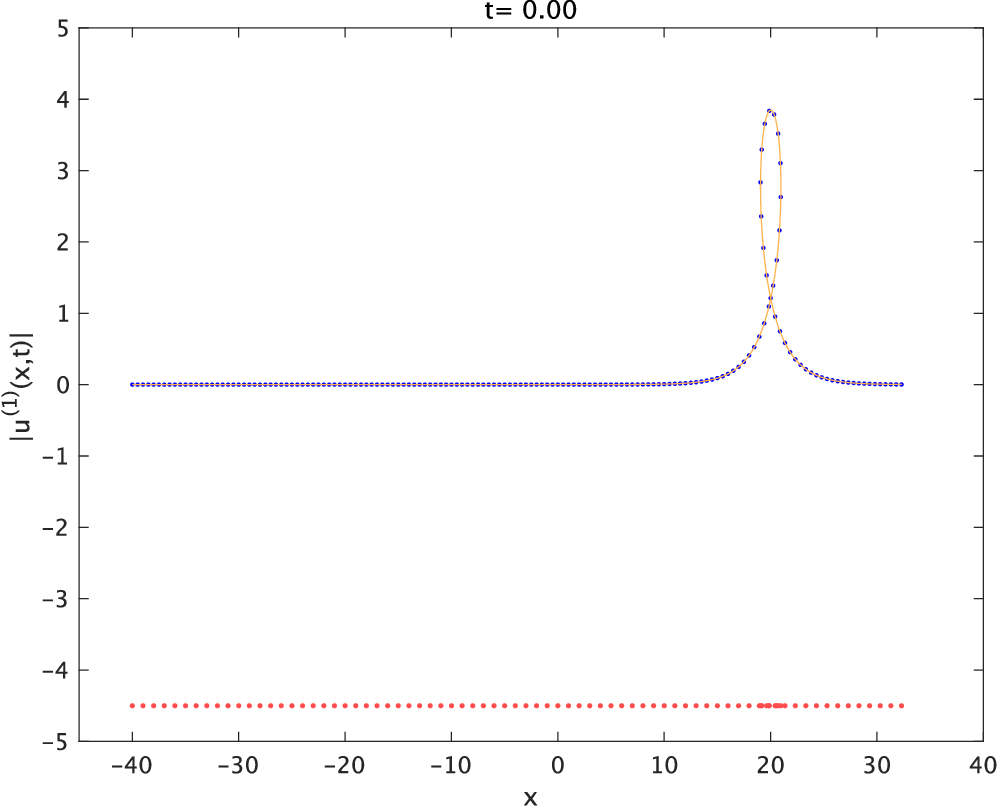}
      \end{minipage} &
      \begin{minipage}[t]{0.47\hsize}
        \centering
        \includegraphics[keepaspectratio, scale=0.33]{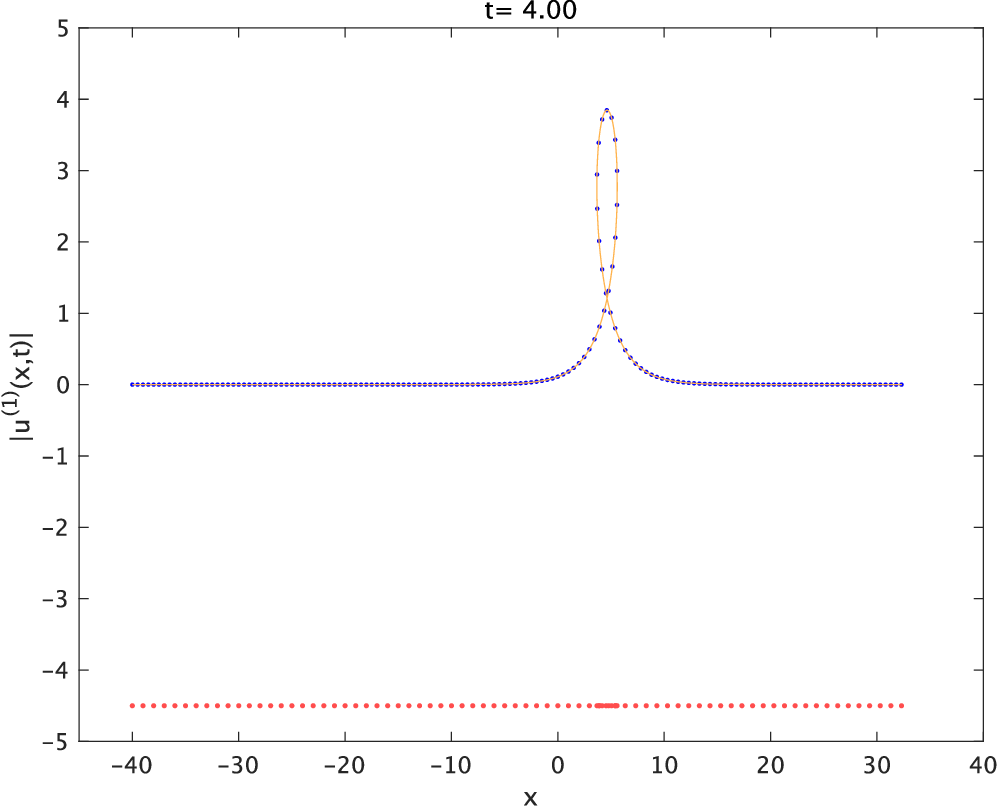}
      \end{minipage}\\ 
  
      \begin{minipage}[t]{0.47\hsize}
        \centering
        \includegraphics[keepaspectratio, scale=0.33]{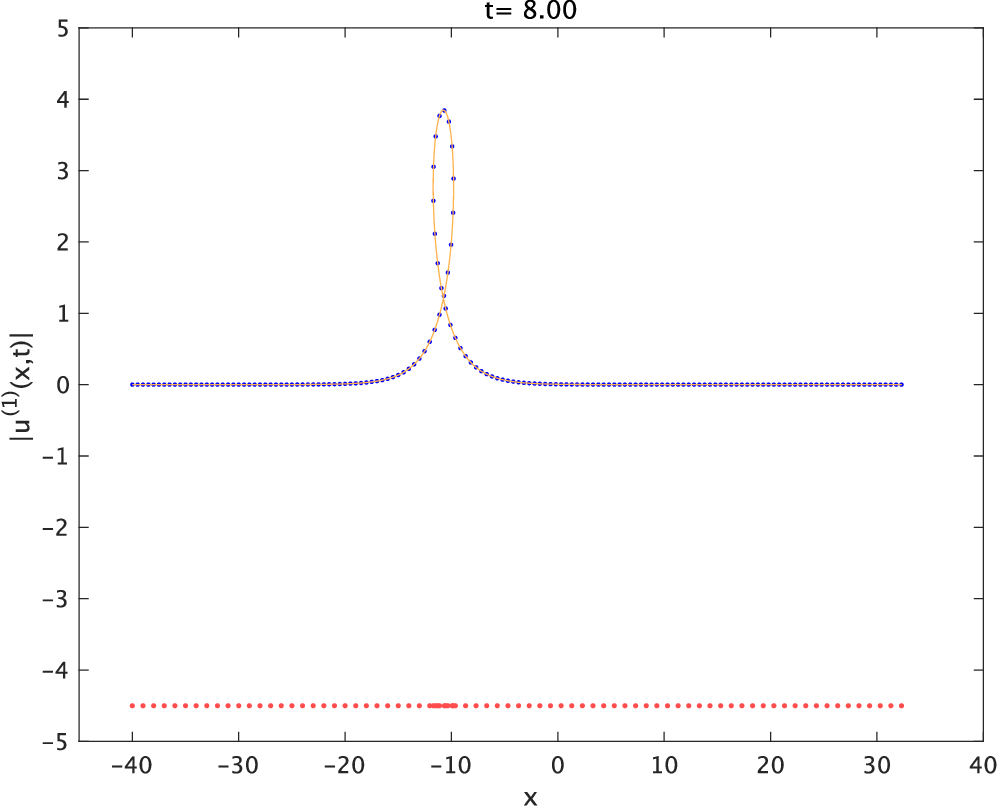}
      \end{minipage} &
      \begin{minipage}[t]{0.47\hsize}
        \centering
        \includegraphics[keepaspectratio, scale=0.33]{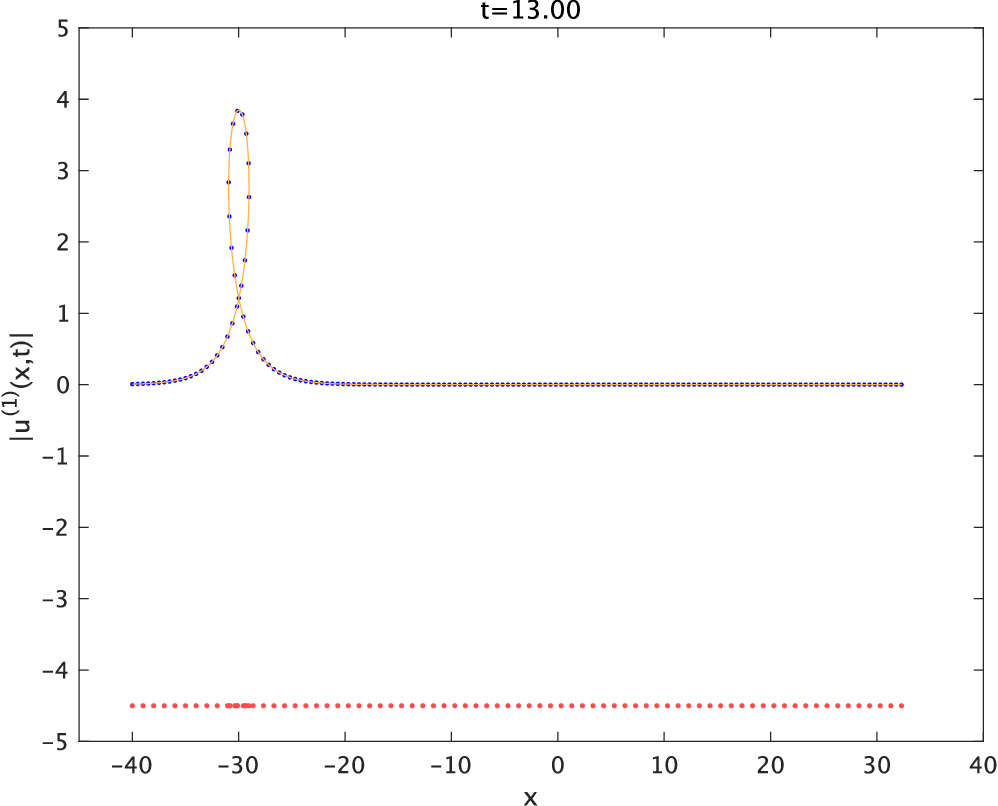}
      \end{minipage} 
       \end{tabular}
     \caption{Numerical simulation of the $|u^{(1)}|$-profile of the one-soliton solution of the 1-CSP equation. 
     ${\rm maxerr}(|u^{(1)}|)=1.93\times 10^{-4}$.}
              \label{1CSP_1_abs}
  \end{figure}
\begin{figure}[htbp]
 \begin{tabular}{cc}
      \begin{minipage}[t]{0.47\hsize}
       \centering
        \includegraphics[keepaspectratio, scale=0.33]{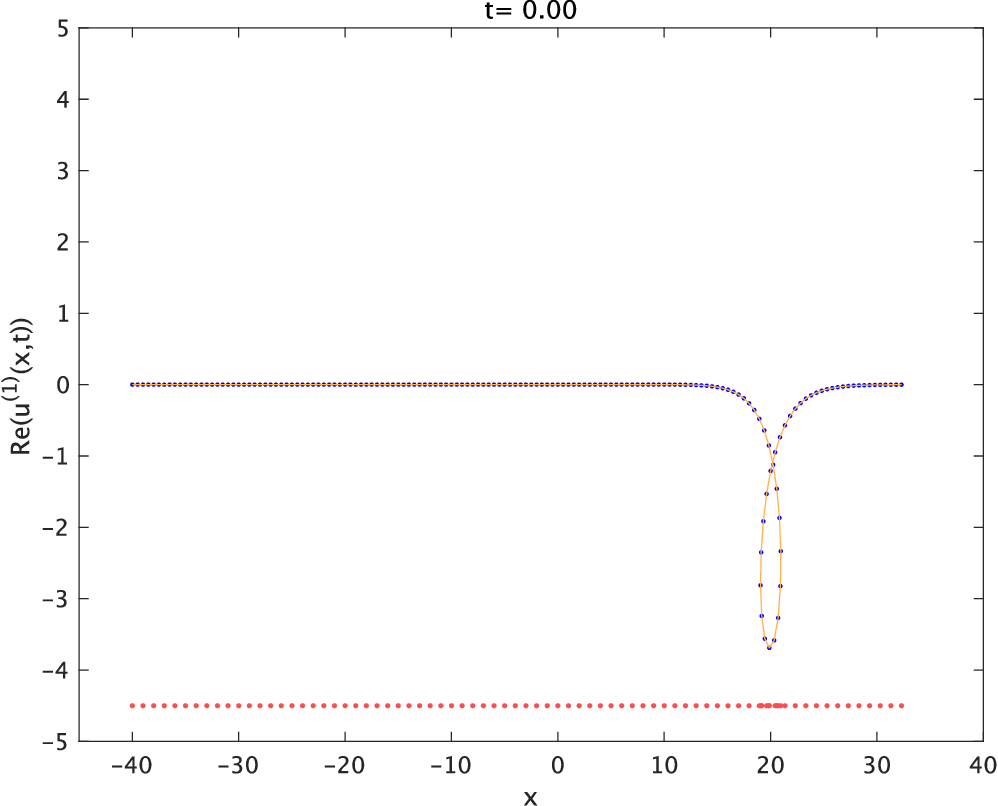}
      \end{minipage} &
      \begin{minipage}[t]{0.47\hsize}
        \centering
        \includegraphics[keepaspectratio, scale=0.33]{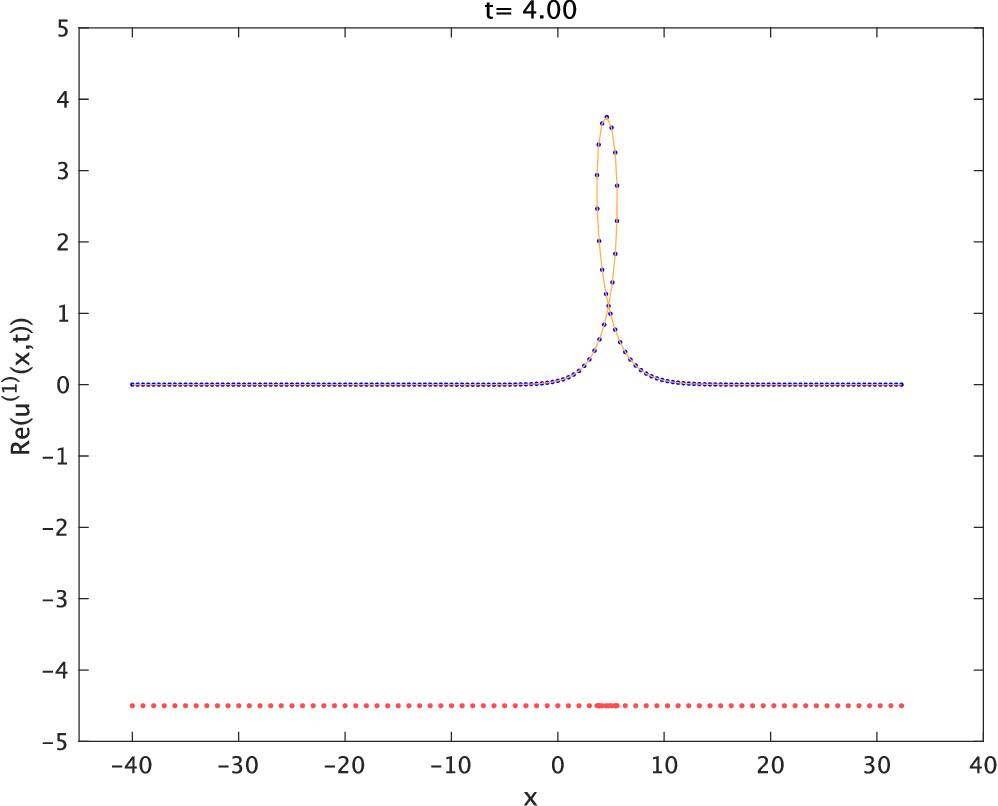}
      \end{minipage}\\ 
  
      \begin{minipage}[t]{0.47\hsize}
        \centering
        \includegraphics[keepaspectratio, scale=0.33]{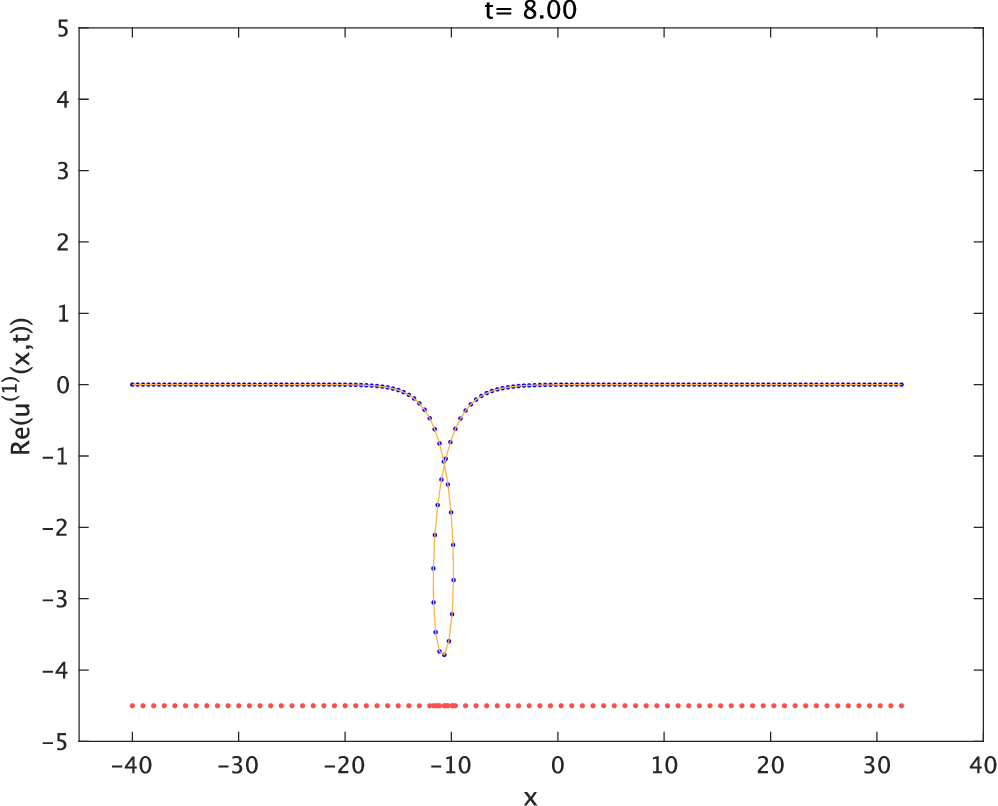}
      \end{minipage} &
      \begin{minipage}[t]{0.47\hsize}
        \centering
        \includegraphics[keepaspectratio, scale=0.33]{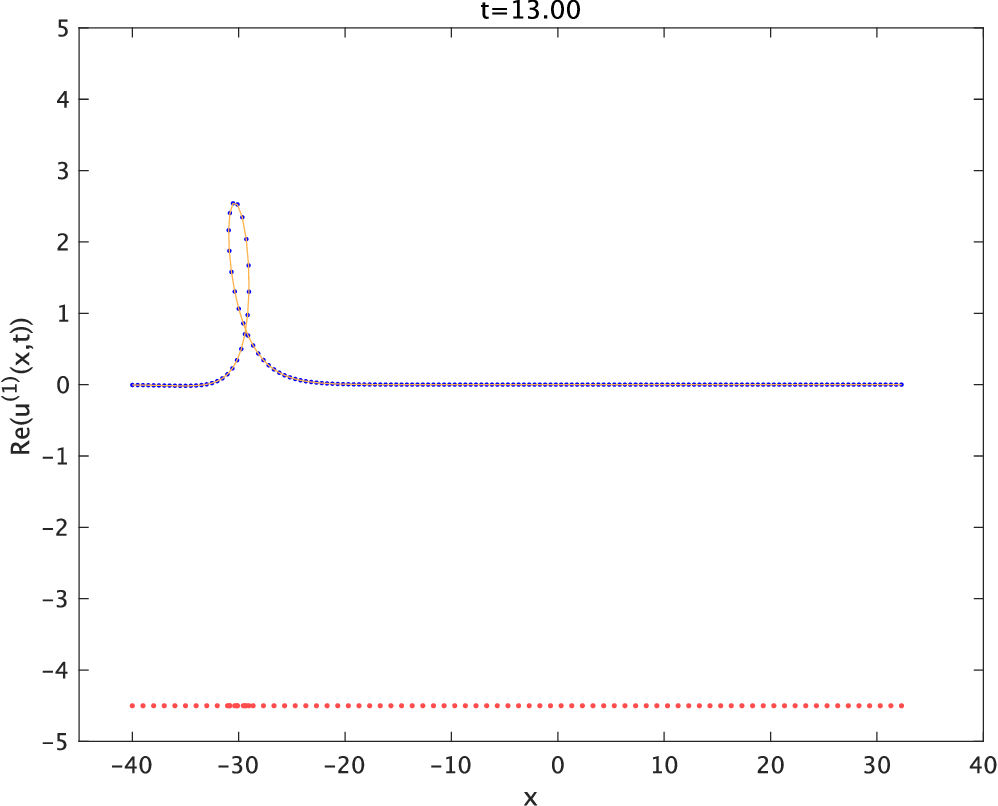}
      \end{minipage} 
       \end{tabular}
     \caption{Numerical simulation of the ${\rm Re}(u^{(1)})$-profile of the one-soliton solution of the 1-CSP equation. ${\rm maxerr}({\rm Re}(u^{(1)}))=2.99\times 10^{-3}$.}
              \label{1CSP_1_re}
  \end{figure}
\begin{figure}[htbp]
 \begin{tabular}{cc}
      \begin{minipage}[t]{0.47\hsize}
       \centering
        \includegraphics[keepaspectratio, scale=0.33]{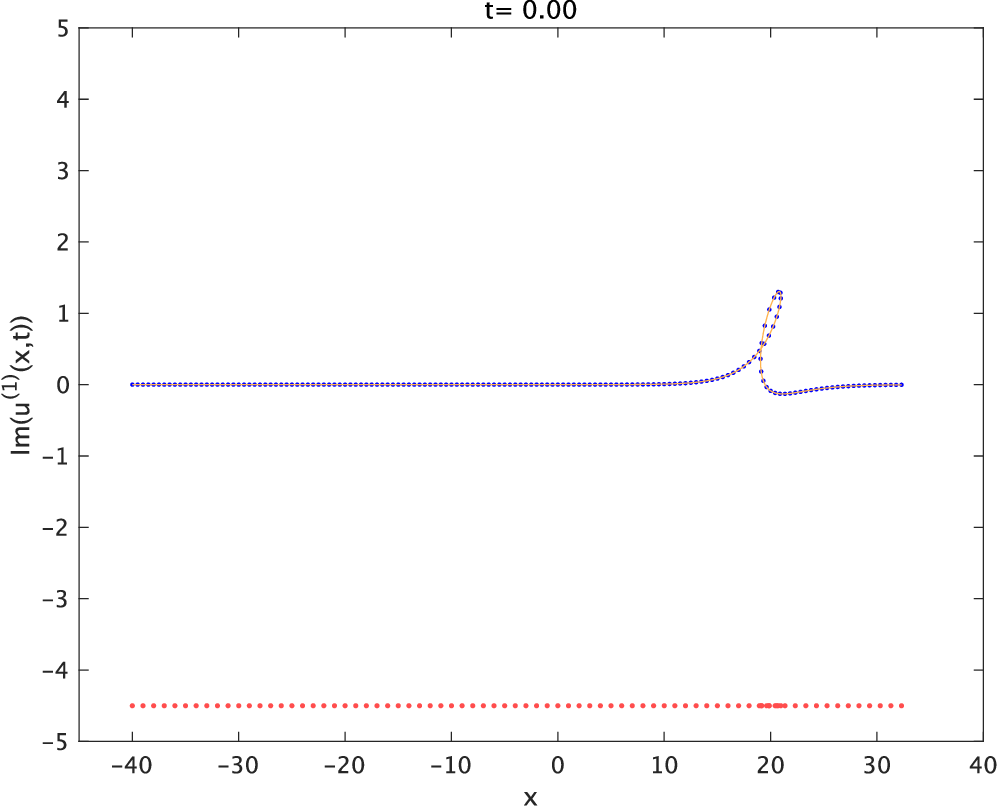}
      \end{minipage} &
      \begin{minipage}[t]{0.47\hsize}
        \centering
        \includegraphics[keepaspectratio, scale=0.33]{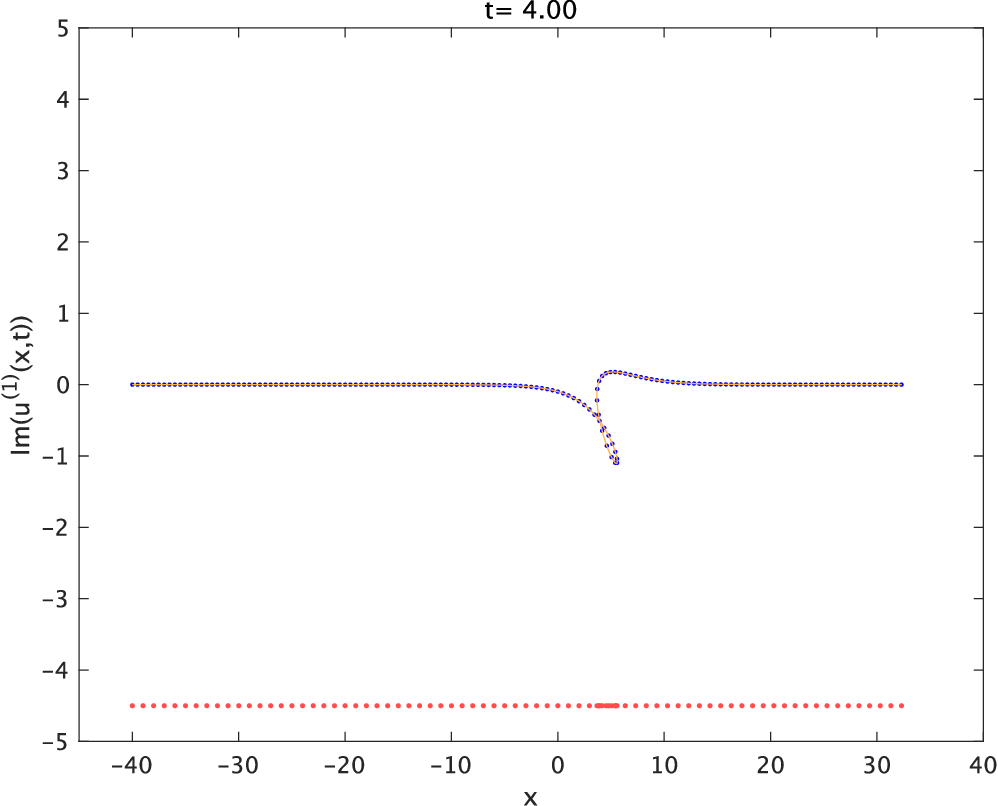}
      \end{minipage}\\ 
  
      \begin{minipage}[t]{0.47\hsize}
        \centering
        \includegraphics[keepaspectratio, scale=0.33]{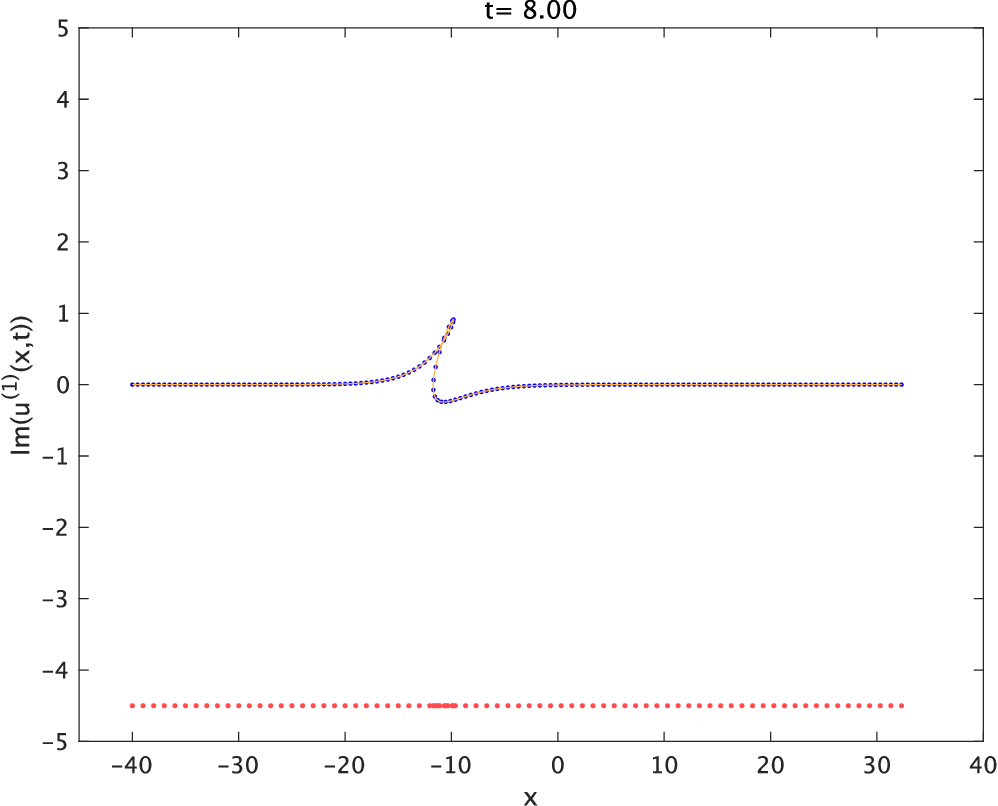}
      \end{minipage} &
      \begin{minipage}[t]{0.47\hsize}
        \centering
        \includegraphics[keepaspectratio, scale=0.33]{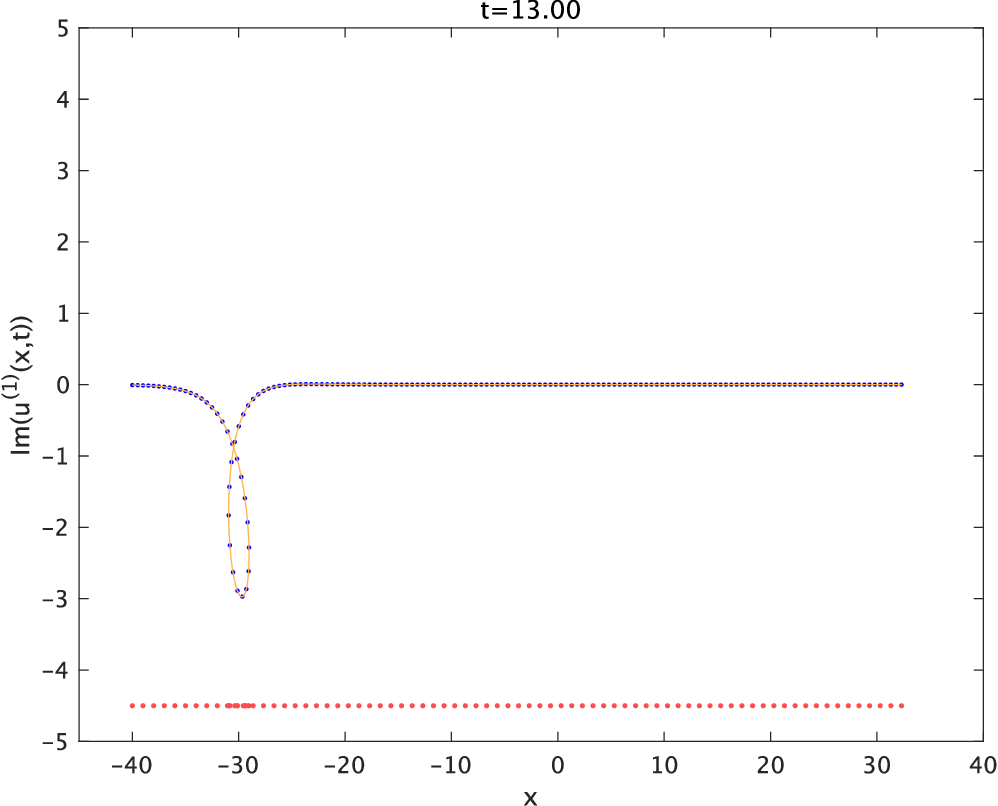}
      \end{minipage} 
       \end{tabular}
     \caption{Numerical simulation of the ${\rm Im}(u^{(1)})$-profile of the one-soliton solution of the 1-CSP equation. ${\rm maxerr}({\rm Im}(u^{(1)}))=4.32\times 10^{-3}$.}
              \label{1CSP_1_imag}
  \end{figure}

\noindent
One-soliton:\\
Applying the constraints $p_{1}=p_{2}^{*}$, $B_{1}=B_{2}^{*}$,  
and $a_{1}^{(1)}=(b_{2}^{(1)})^*$ to (\ref{eq_5.4}) gives
 \begin{eqnarray}
\fl{u^{(1)}=\frac{g^{(1)}}{f},\quad x=X-2(\log{f})_{T},\quad t=T,}\nonumber\\
\fl{f=1+\displaystyle\frac{1}{4}\left(\frac{p_{1}p_{1}^{*}}{p_{1}+p_{1}^{*}}\right)^{2}
\varphi_{1}\varphi_{1}^{*}a_{1}^{(1)}(a_{1}^{(1)})^{*},\quad
g^{(1)}=a_{1}^{(1)}\varphi_{1},}
\end{eqnarray}
where $\varphi_{1}=B_{1}e^{p_{1}X+p_{1}^{-1}T}$.\\

\noindent
\begin{samepage}
Two-soliton:\\
Applying the constraints $p_{1}=p_{4}^{*}$, $p_{2}=p_{3}^{*}$, 
$B_{1}=B_{4}^{*}$, $B_{2}=B_{3}^{*}$, $a_{1}^{(1)}=(b_{4}^{(1)})^*$, and
$a_{2}^{(1)}=(b_{3}^{(1)})^*$ to (\ref{eq_5.5}) 
and using the notation in (\ref{con_compact_coefficients}), we obtain 
\begin{eqnarray*}
\fl C_{14}=a_1^{(1)}(a_1^{(1)})^*,\quad
C_{13}=a_1^{(1)}(a_2^{(1)})^*,\quad
C_{24}=a_2^{(1)}(a_1^{(1)})^*,\quad
C_{23}=a_2^{(1)}(a_2^{(1)})^*,\\
\fl \varphi_3=\varphi_2^*,\qquad \varphi_4=\varphi_1^*.
\end{eqnarray*}
\begin{eqnarray}
\fl u^{(1)}&=&\frac{g^{(1)}}{f},\qquad x=X-2(\log f)_T,\qquad t=T\nonumber\\
\fl f&=&1+\tilde{C}_{14}\varphi_1\varphi_4+\tilde{C}_{13}\varphi_1\varphi_3
+\tilde{C}_{24}\varphi_2\varphi_4+\tilde{C}_{23}\varphi_2\varphi_3\nonumber\\
\fl &&+\Delta\left(\prod_{1\leq i<j\leq4}R_{ij}\right)
\varphi_1\varphi_2\varphi_3\varphi_4,\nonumber\\
\fl g^{(1)}&=&a_{1}^{(1)}\left(\varphi_1+R_{12}R_{14}\tilde{C}_{24}\varphi_1\varphi_2\varphi_4
+R_{12}R_{13}\tilde{C}_{23}\varphi_1\varphi_2\varphi_3\right)\nonumber\\
\fl &&+a_{2}^{(1)}\left(\varphi_2-R_{12}R_{24}\tilde{C}_{14}\varphi_1\varphi_2\varphi_4
-R_{12}R_{23}\tilde{C}_{13}\varphi_1\varphi_2\varphi_3\right).\label{eq_5.9}
\end{eqnarray}
where $\varphi_{i}=B_{i}e^{p_{i}X+p_{i}^{-1}T}$ and $i=1, 2, 3, 4$.
\end{samepage}

\begin{figure}[htbp]
 \begin{tabular}{cc}
      \begin{minipage}[t]{0.47\hsize}
       \centering
        \includegraphics[keepaspectratio, scale=0.33]{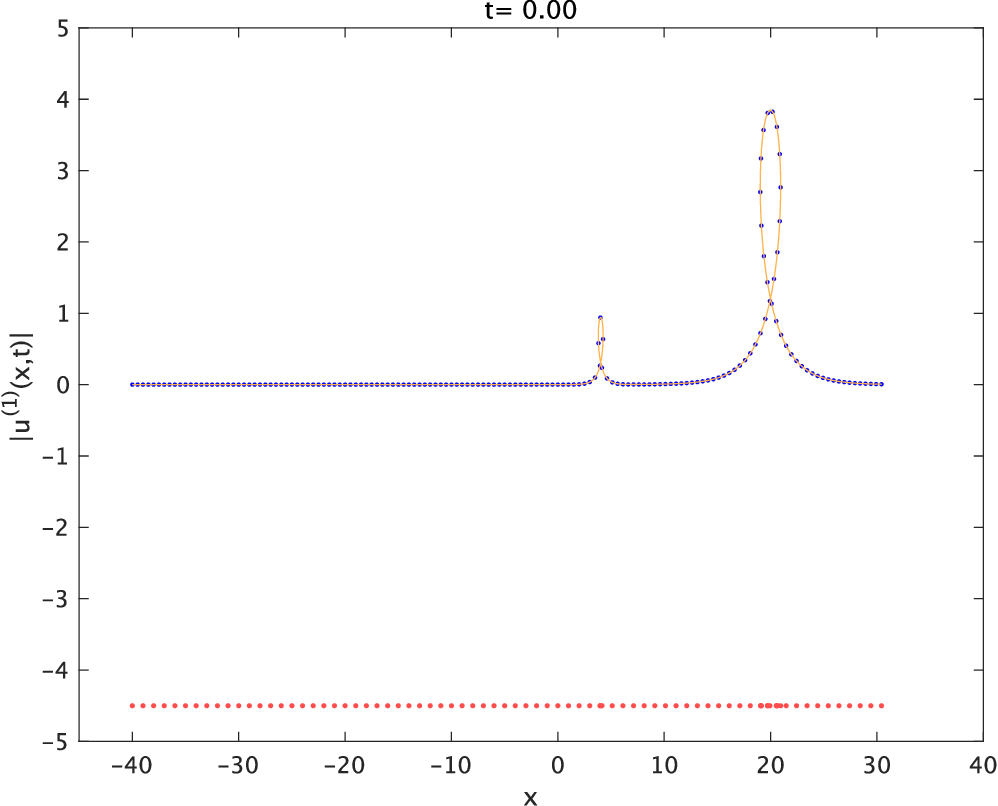}
      \end{minipage} &
      \begin{minipage}[t]{0.47\hsize}
        \centering
        \includegraphics[keepaspectratio, scale=0.33]{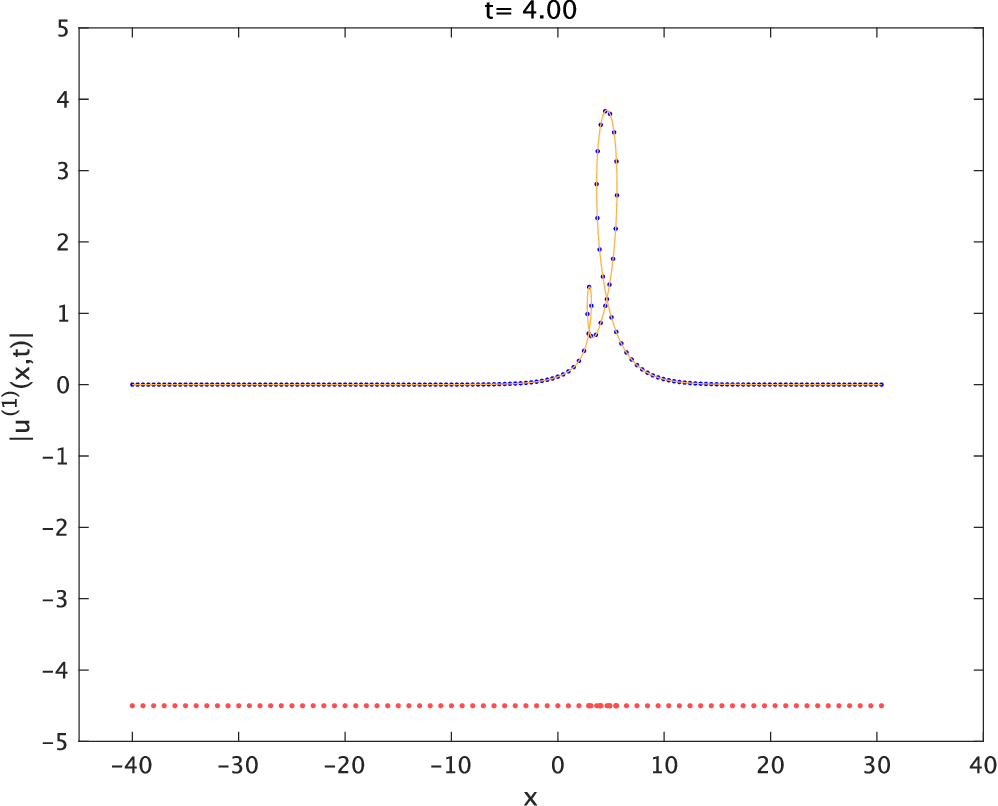}
      \end{minipage}\\ 
  
      \begin{minipage}[t]{0.47\hsize}
        \centering
        \includegraphics[keepaspectratio, scale=0.33]{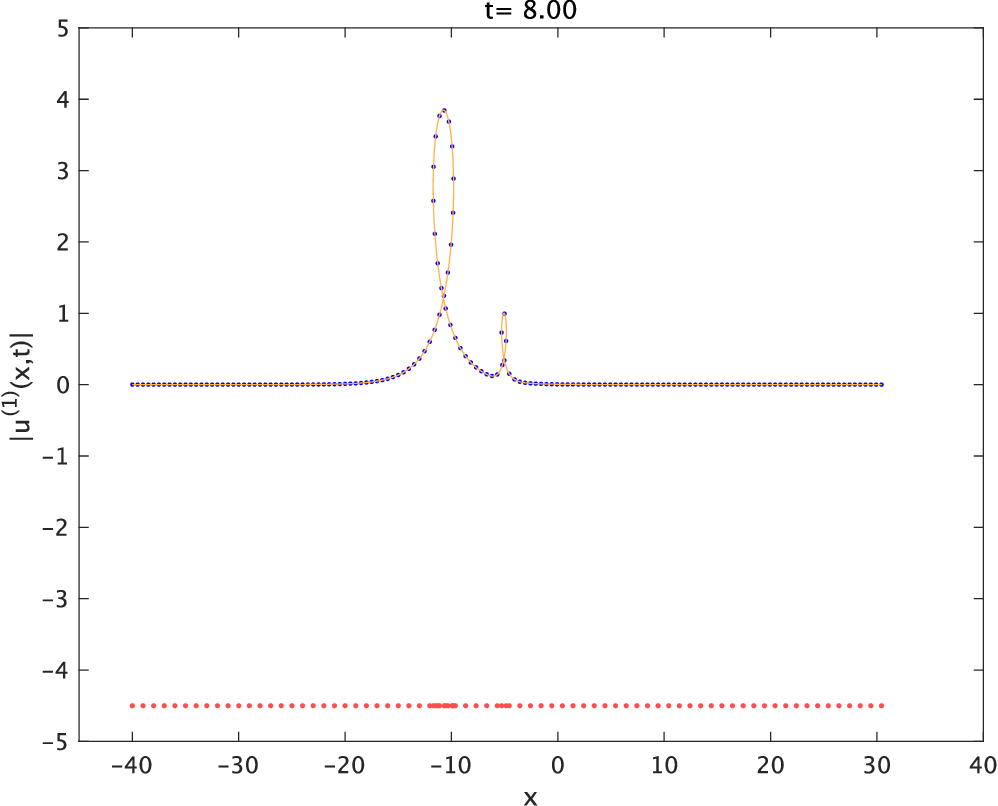}
      \end{minipage} &
      \begin{minipage}[t]{0.47\hsize}
        \centering
        \includegraphics[keepaspectratio, scale=0.33]{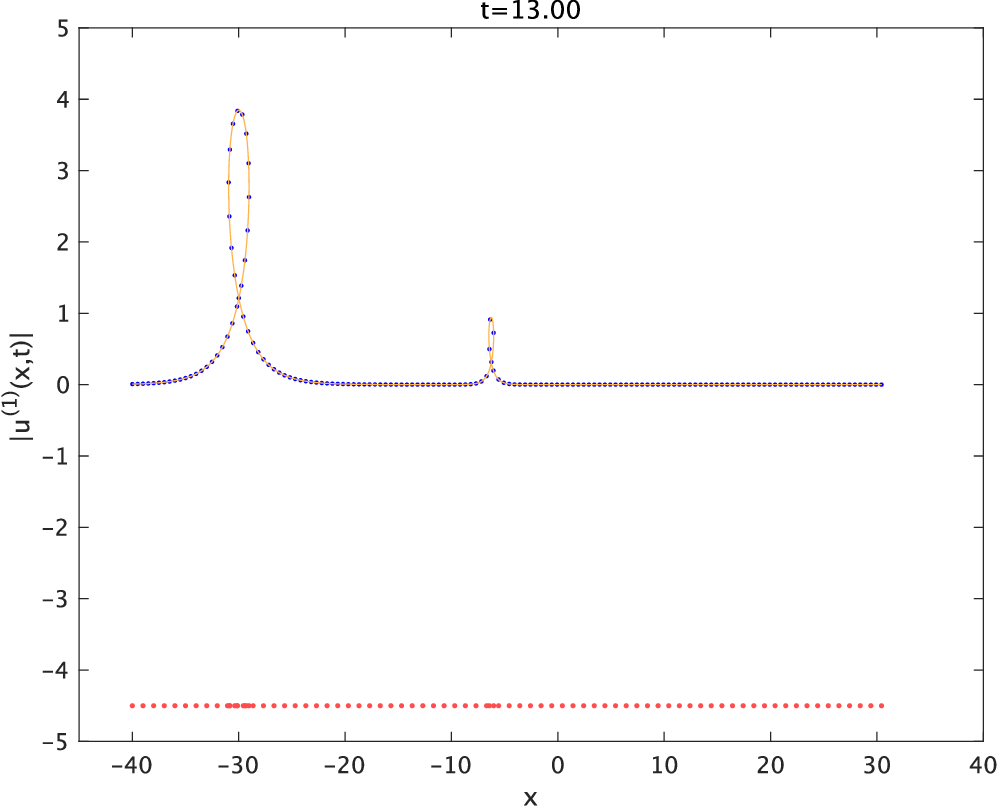}
      \end{minipage} 
       \end{tabular}
     \caption{Numerical simulation of the $|u^{(1)}|$-profile of the two-soliton solution of the 1-CSP equation. ${\rm maxerr}(|u^{(1)}|)=4.03\times 10^{-4}$.}
              \label{1CSP_2_abs}
  \end{figure}
\begin{figure}[htbp]
 \begin{tabular}{cc}
      \begin{minipage}[t]{0.47\hsize}
       \centering
        \includegraphics[keepaspectratio, scale=0.33]{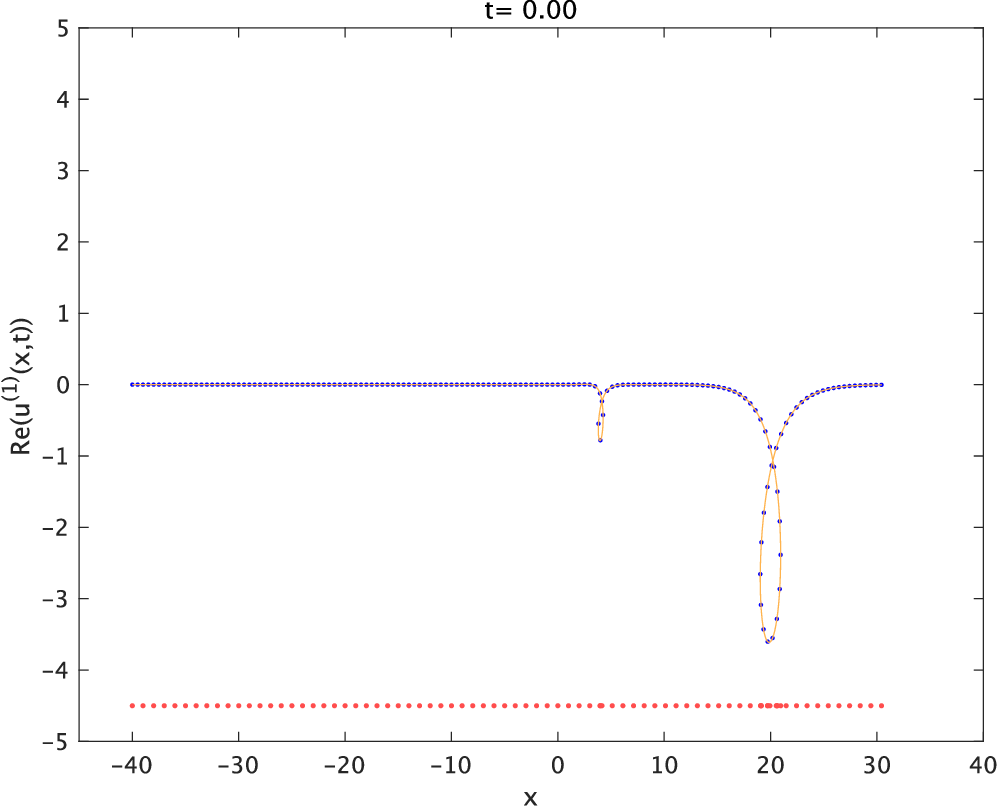}
      \end{minipage} &
      \begin{minipage}[t]{0.47\hsize}
        \centering
        \includegraphics[keepaspectratio, scale=0.33]{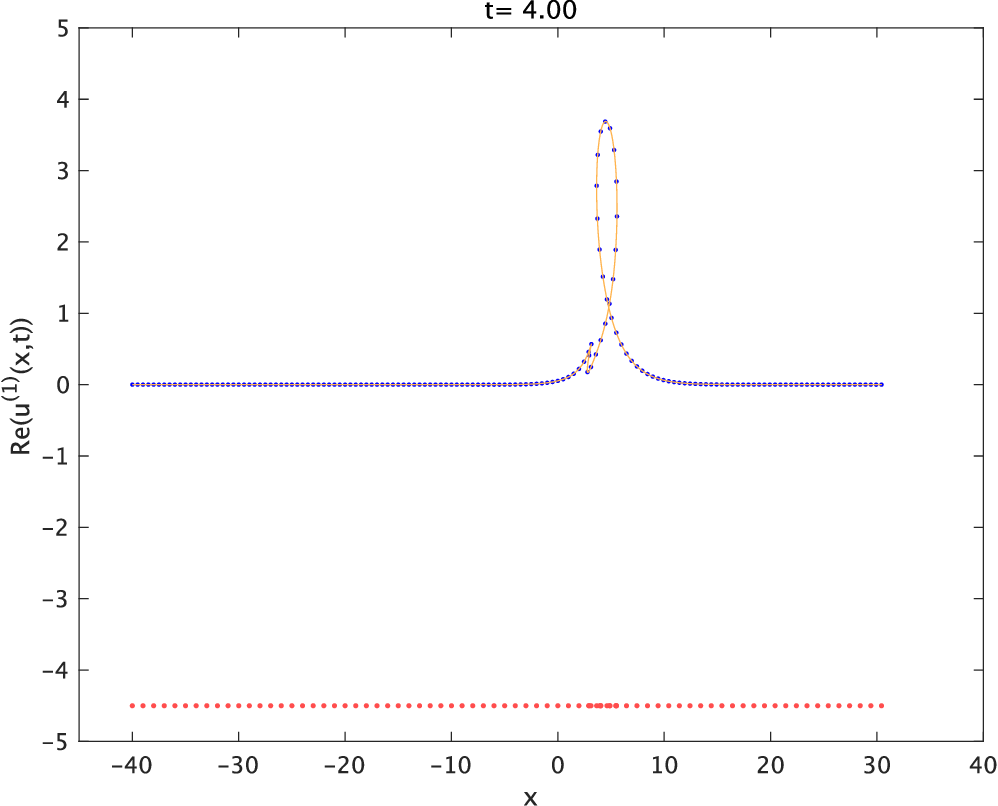}
      \end{minipage}\\ 
  
      \begin{minipage}[t]{0.47\hsize}
        \centering
        \includegraphics[keepaspectratio, scale=0.33]{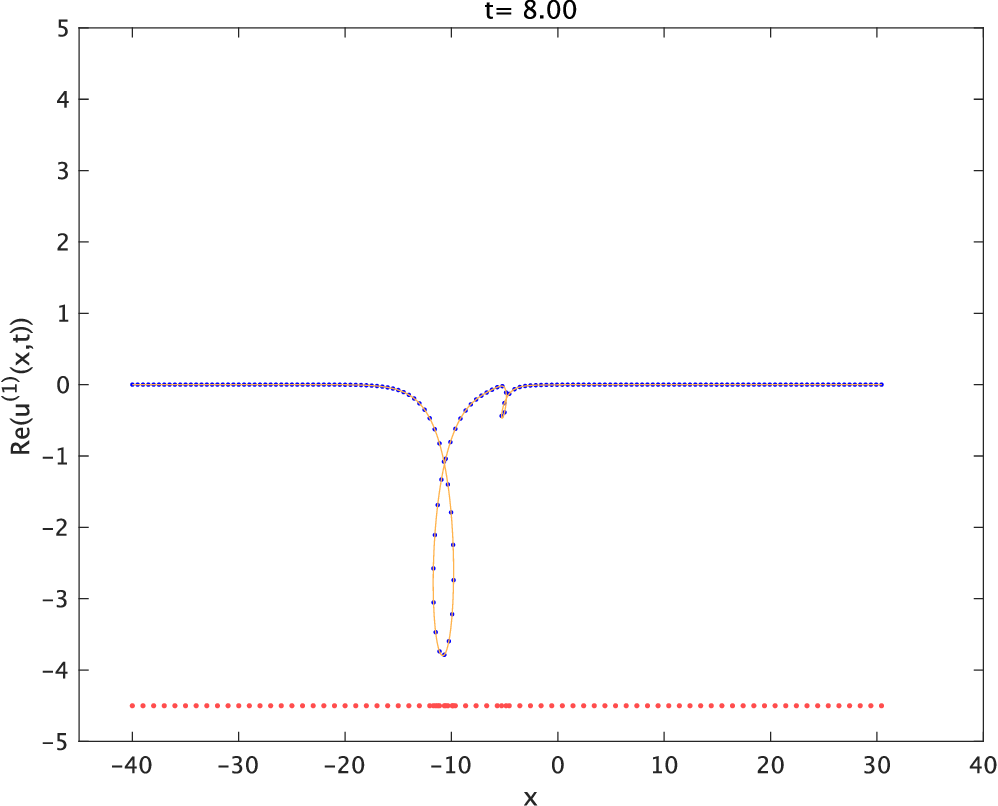}
      \end{minipage} &
      \begin{minipage}[t]{0.47\hsize}
        \centering
        \includegraphics[keepaspectratio, scale=0.33]{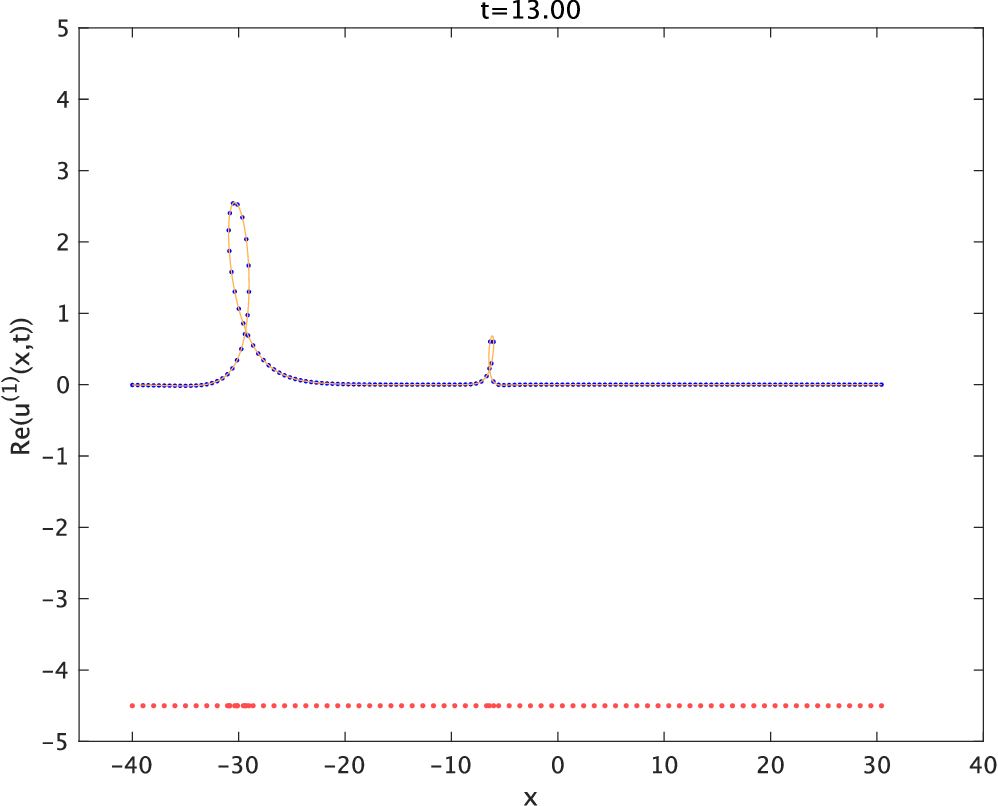}
      \end{minipage} 
       \end{tabular}
     \caption{Numerical simulation of the ${\rm Re}(u^{(1)})$-profile of the two-soliton solution of the 1-CSP equation. ${\rm maxerr}({\rm Re}(u^{(1)}))=2.07\times 10^{-3}$.}
              \label{1CSP_2_re}
  \end{figure}
\begin{figure}[htbp]
 \begin{tabular}{cc}
      \begin{minipage}[t]{0.47\hsize}
       \centering
        \includegraphics[keepaspectratio, scale=0.33]{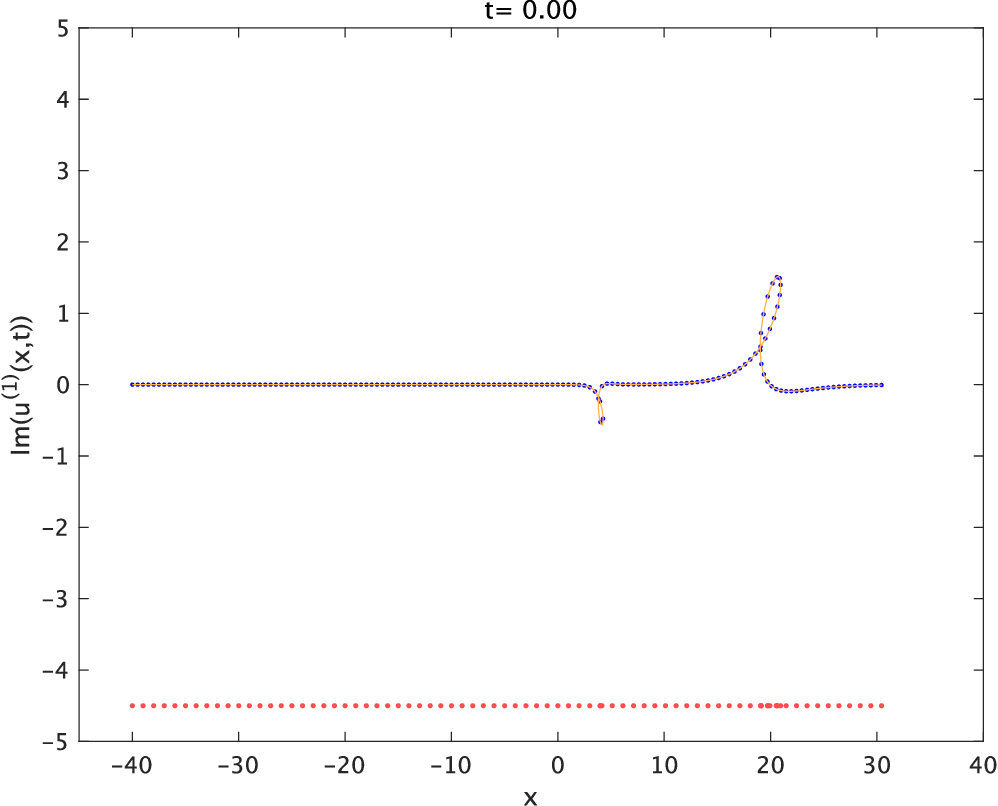}
      \end{minipage} &
      \begin{minipage}[t]{0.47\hsize}
        \centering
        \includegraphics[keepaspectratio, scale=0.33]{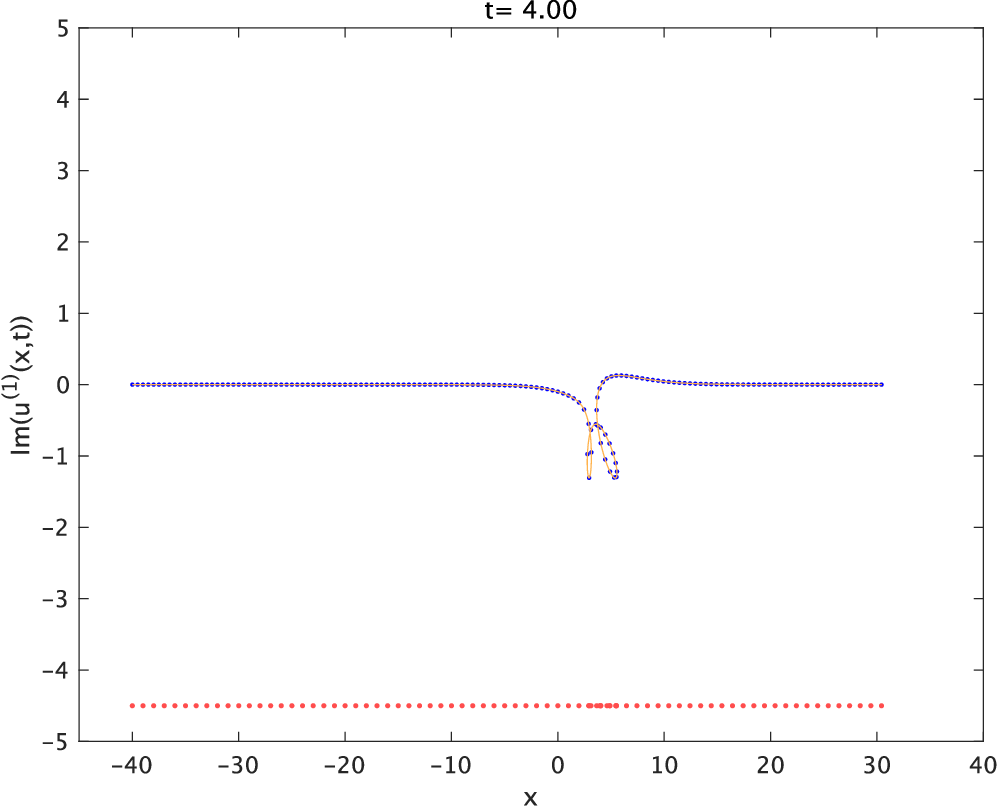}
      \end{minipage}\\ 
  
      \begin{minipage}[t]{0.47\hsize}
        \centering
        \includegraphics[keepaspectratio, scale=0.33]{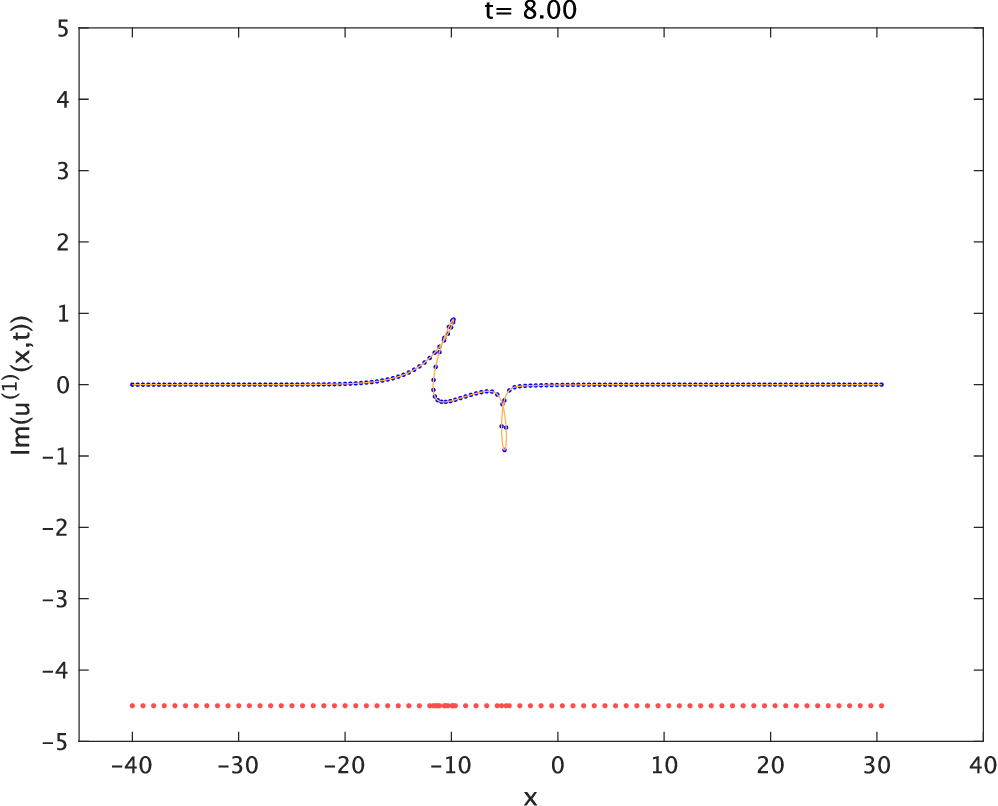}
      \end{minipage} &
      \begin{minipage}[t]{0.47\hsize}
        \centering
        \includegraphics[keepaspectratio, scale=0.33]{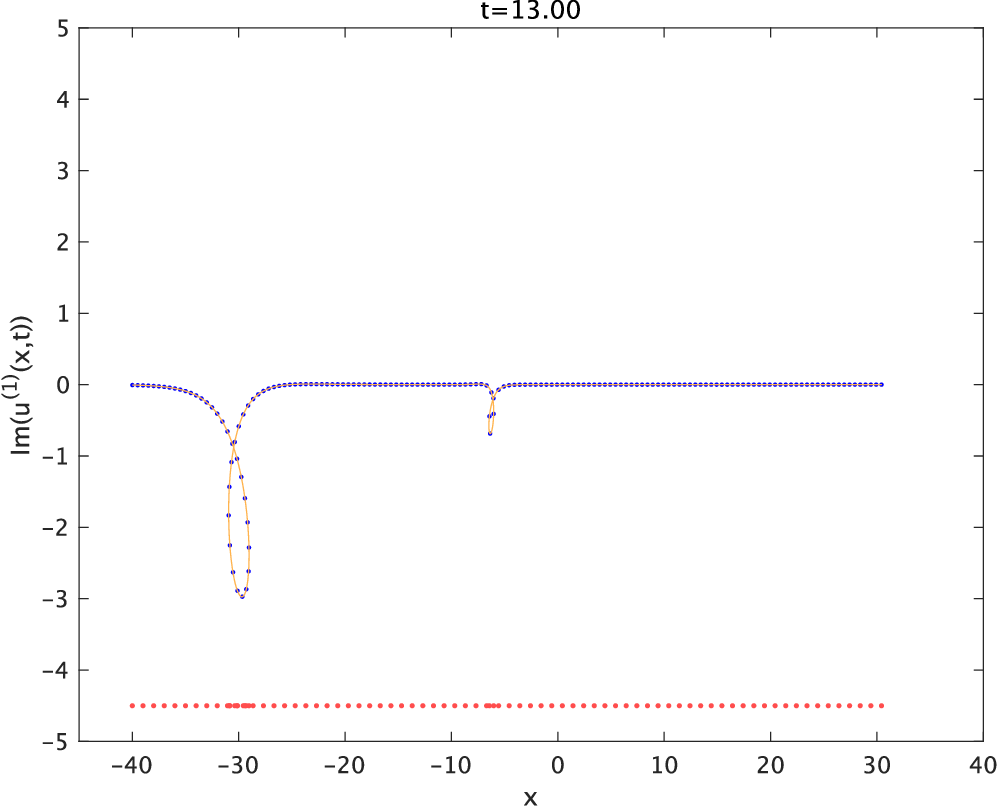}
      \end{minipage} 
       \end{tabular}
     \caption{Numerical simulation of the ${\rm Im}(u^{(1)})$-profile of the two-soliton solution of the 1-CSP equation. ${\rm maxerr}({\rm Im}(u^{(1)}))=3.87\times 10^{-3}$.}
              \label{1CSP_2_imag}
  \end{figure}

Figures \ref{1CSP_1_abs}--\ref{1CSP_1_imag} show the numerical
$|u^{(1)}|$, ${\rm Re}(u^{(1)})$, and ${\rm Im}(u^{(1)})$ profiles for the
one-soliton solution with $p_{1}=0.5+0.1\mathrm{i}$,
$a_{1}^{(1)}=1+0.5\mathrm{i}$, and $B_{1}={\rm exp}(-10)$.
Figures \ref{1CSP_2_abs}--\ref{1CSP_2_imag} show the corresponding profiles for
the two-soliton solution with $p_{1}=0.5+0.1\mathrm{i}$,
$p_{2}=2-0.5\mathrm{i}$, $a_{1}^{(1)}=1+0.5\mathrm{i}$,
$a_{2}^{(1)}=2-0.1\mathrm{i}$, and $B_{1}=B_{2}={\rm exp}(-10)$. In both
cases, the solitons propagate to the left along the $x$-axis.\\

  \noindent 
Two-component CSP (2-CSP) equation:\\
Setting $I=J=\{1,2\}$ in (\ref{MCSP}) gives the 4-component SP (4-SP)
equation
 \begin{eqnarray}
 \fl{
\left\{
\begin{array}{ll}
u^{(1)}_{xt}=u^{(1)}+\displaystyle\frac{1}{2}\left(\left(c_{11}u^{(1)}v^{(1)}+c_{12}u^{(1)}v^{(2)}+c_{21}u^{(2)}v^{(1)}+c_{22}u^{(2)}v^{(2)}\right)u^{(1)}_{x}\right)_{x},\\
\\
u^{(2)}_{xt}=u^{(2)}+\displaystyle\frac{1}{2}\left(\left(c_{11}u^{(1)}v^{(1)}+c_{12}u^{(1)}v^{(2)}+c_{21}u^{(2)}v^{(1)}+c_{22}u^{(2)}v^{(2)}\right)u^{(2)}_{x}\right)_{x},\\
\\
v^{(1)}_{xt}=v^{(1)}+\displaystyle\frac{1}{2}\left(\left(c_{11}u^{(1)}v^{(1)}+c_{12}u^{(1)}v^{(2)}+c_{21}u^{(2)}v^{(1)}+c_{22}u^{(2)}v^{(2)}\right)v^{(1)}_{x}\right)_{x},\\
\\
v^{(2)}_{xt}=v^{(2)}+\displaystyle\frac{1}{2}\left(\left(c_{11}u^{(1)}v^{(1)}+c_{12}u^{(1)}v^{(2)}+c_{21}u^{(2)}v^{(1)}+c_{22}u^{(2)}v^{(2)}\right)v^{(2)}_{x}\right)_{x}.\\
\end{array}
\right.}
\label{4SP}
\end{eqnarray}
If $u^{(1)}$, $u^{(2)}$, $v^{(1)}$, and $v^{(2)}$ are complex-valued,
then imposing $v^{(1)}=(u^{(1)})^*$, $v^{(2)}=(u^{(2)})^*$,
$c_{11}=c_{22}=1$, and $c_{12}=c_{21}=0$ reduces (\ref{4SP}) to the 2-CSP
equation
 \begin{eqnarray}
\left\{
\begin{array}{ll}
u^{(1)}_{xt}=u^{(1)}+\displaystyle\frac{1}{2}\Big(\left(|u^{(1)}|^{2}+|u^{(2)}|^{2}\right)u^{(1)}_{x}\Big)_{x}, \\
\\
u^{(2)}_{xt}=u^{(2)}+\displaystyle\frac{1}{2}\Big(\left(|u^{(1)}|^{2}+|u^{(2)}|^{2}\right)u^{(2)}_{x}\Big)_{x}. \\
\end{array}
\right.
\label{2CSP}
\end{eqnarray}
The corresponding fully discrete analogue is
\begin{eqnarray}
\left\{
\begin{array}{lll}
\displaystyle\frac{u^{(1)}_{l,m+1}+u^{(1)}_{l-1,m}}{u^{(1)}_{l,m}+u^{(1)}_{l-1,m+1}}=\displaystyle\frac{1-b(x_{l-1,m+1}-x_{l-1,m})}{1-b(x_{l,m+1}-x_{l,m})} \displaystyle\frac{u^{(1)}_{l,m}+u^{(1)}_{l-1,m-1}}{u^{(1)}_{l,m-1}+u^{(1)}_{l-1,m}},\\
\displaystyle\frac{u^{(2)}_{l,m+1}+u^{(2)}_{l-1,m}}{u^{(2)}_{l,m}+u^{(2)}_{l-1,m+1}}=\displaystyle\frac{1-b(x_{l-1,m+1}-x_{l-1,m})}{1-b(x_{l,m+1}-x_{l,m})} \displaystyle\frac{u^{(2)}_{l,m}+u^{(2)}_{l-1,m-1}}{u^{(2)}_{l,m-1}+u^{(2)}_{l-1,m}},\\
\displaystyle\frac{x_{l,m+1}-x_{l,m}}{2b}=-\displaystyle\frac{1}{2}\left(|u^{(1)}_{l,m}|^{2}+|u^{(2)}_{l,m}|^{2}\right).
\end{array}
\right.
\label{2CSPfulldis}
\end{eqnarray}
We test (\ref{2CSPfulldis}) using initial data generated from complex-conjugacy
reductions of the 4-SP soliton solutions listed in the Appendix B.
For the one-soliton solution, we impose
$p_{1}=p_{2}^{*}$, $a_{1}^{(\mu)}=(b_{2}^{(\mu)})^*$ $(\mu=1,2)$, and
$B_{1}=B_{2}^{*}$. 
For the two-soliton solution, we impose
$p_{1}=p_{4}^{*}$, $p_{2}=p_{3}^{*}$,
$a_{1}^{(\mu)}=(b_{4}^{(\mu)})^*$, $a_{2}^{(\mu)}=(b_{3}^{(\mu)})^*$
$(\mu=1,2)$, $B_{1}=B_{4}^{*}$, and $B_{2}=B_{3}^{*}$. 
The resulting reduced solutions used to generate the initial data are given below.

\noindent
One-soliton:
 \begin{eqnarray}
\fl{u^{(1)}=\frac{g^{(1)}}{f},\quad u^{(2)}=\frac{g^{(2)}}{f},\quad x=X-2(\log{f})_{T},\quad t=T,}\nonumber\\
\fl{f=1+\displaystyle\frac{1}{4}\left(\displaystyle\frac{p_{1}p_{1}^{*}}{p_{1}+p_{1}^{*}}\right)^{2}\varphi_{1}\varphi_{1}^{*}(a^{(1)}_{1}(a^{(1)}_{1})^*+a^{(2)}_{1}(a^{(2)}_{1})^*),}\nonumber\\
\fl{g^{(1)}=a^{(1)}_{1}\varphi_{1},\quad g^{(2)}=a^{(2)}_{1}\varphi_{1},}\label{eq_5.12}
\end{eqnarray}
where $\varphi_{1}=B_{1}e^{p_{1}X+p_{1}^{-1}T}$.
\begin{figure}[htbp]
 \begin{tabular}{cc}
      \begin{minipage}[t]{0.47\hsize}
       \centering
        \includegraphics[keepaspectratio, scale=0.33]{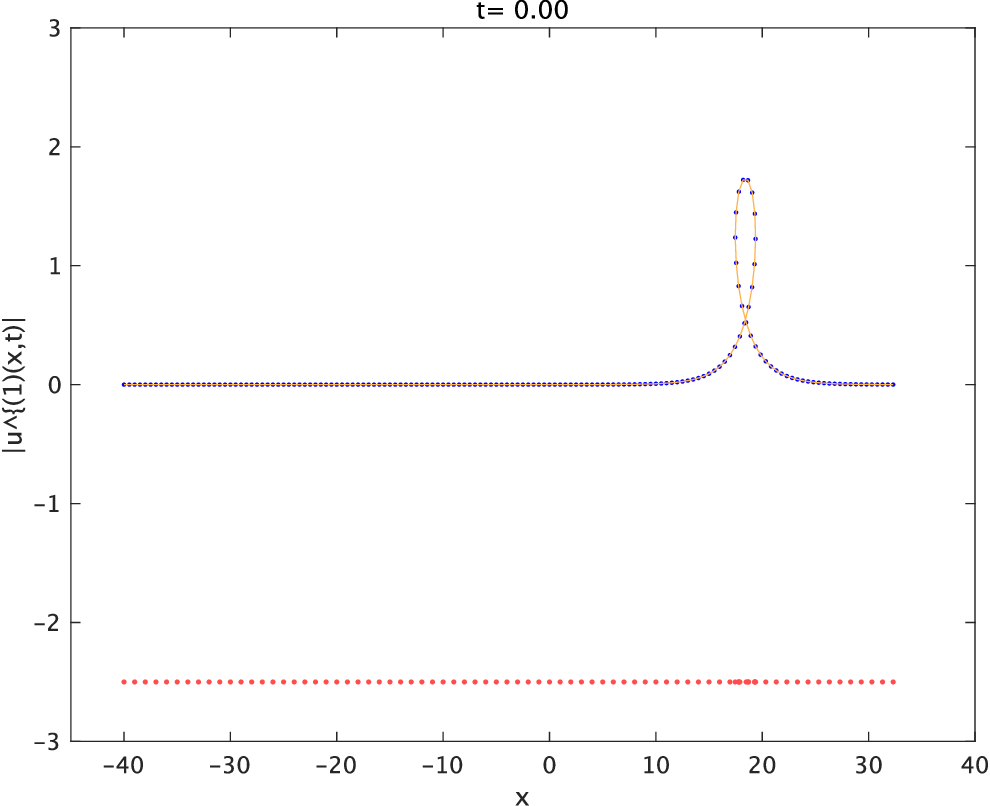}
      \end{minipage} &
      \begin{minipage}[t]{0.47\hsize}
        \centering
        \includegraphics[keepaspectratio, scale=0.33]{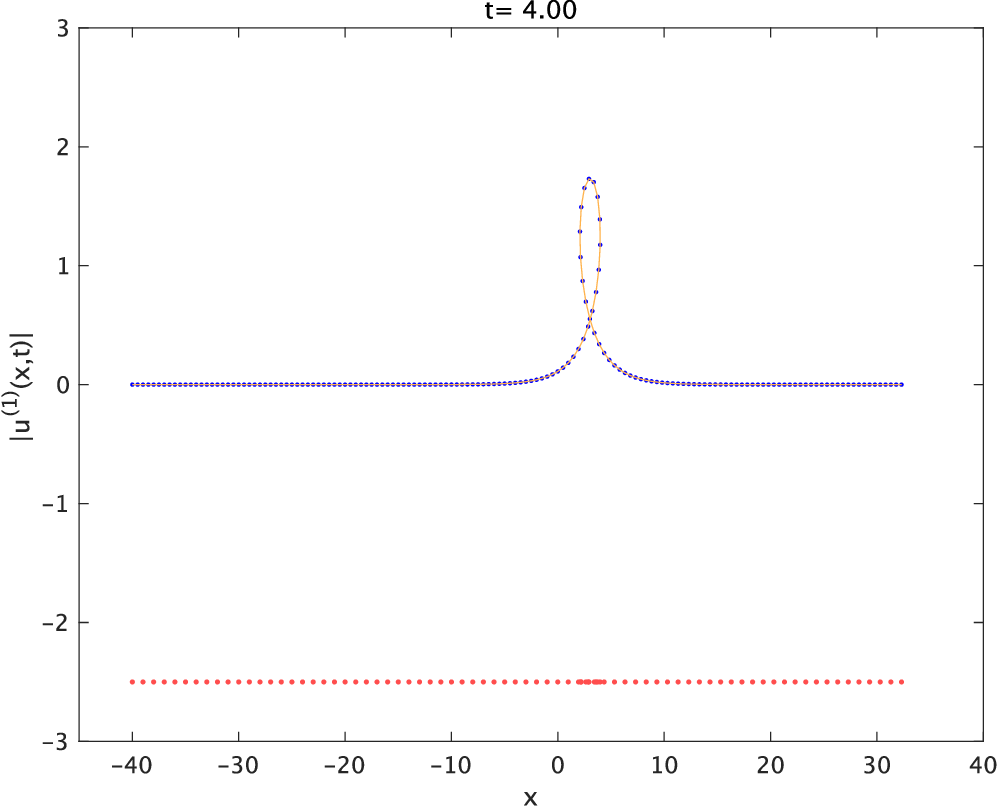}
      \end{minipage}\\ 
  
      \begin{minipage}[t]{0.47\hsize}
        \centering
        \includegraphics[keepaspectratio, scale=0.33]{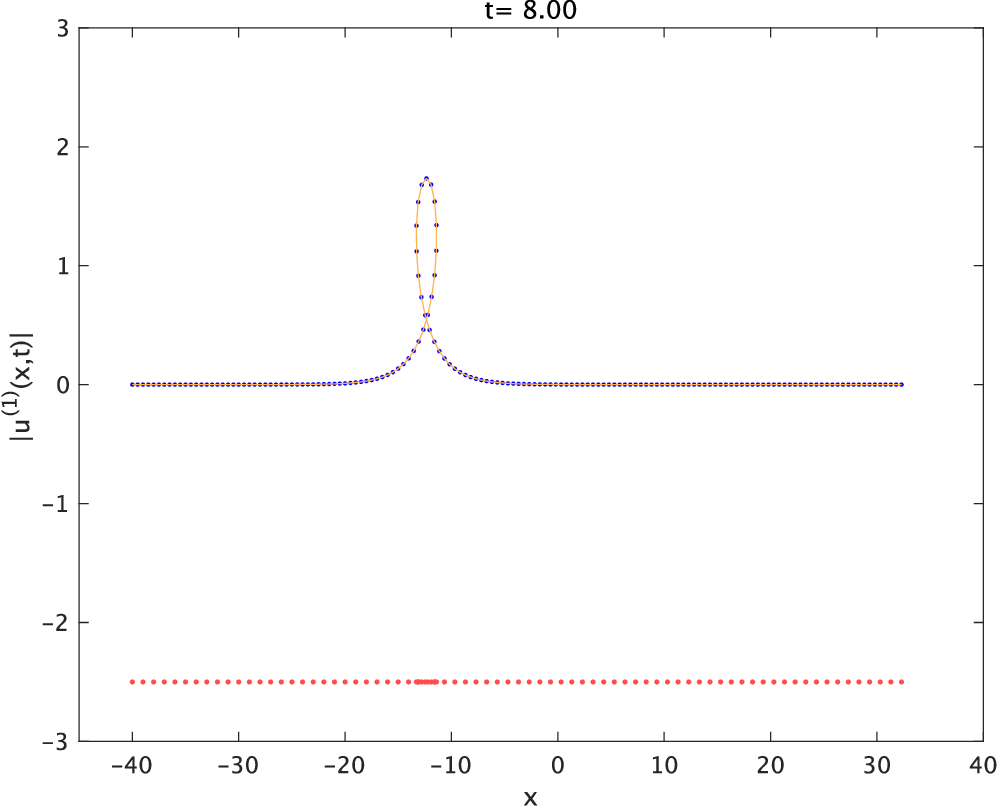}
      \end{minipage} &
      \begin{minipage}[t]{0.47\hsize}
        \centering
        \includegraphics[keepaspectratio, scale=0.33]{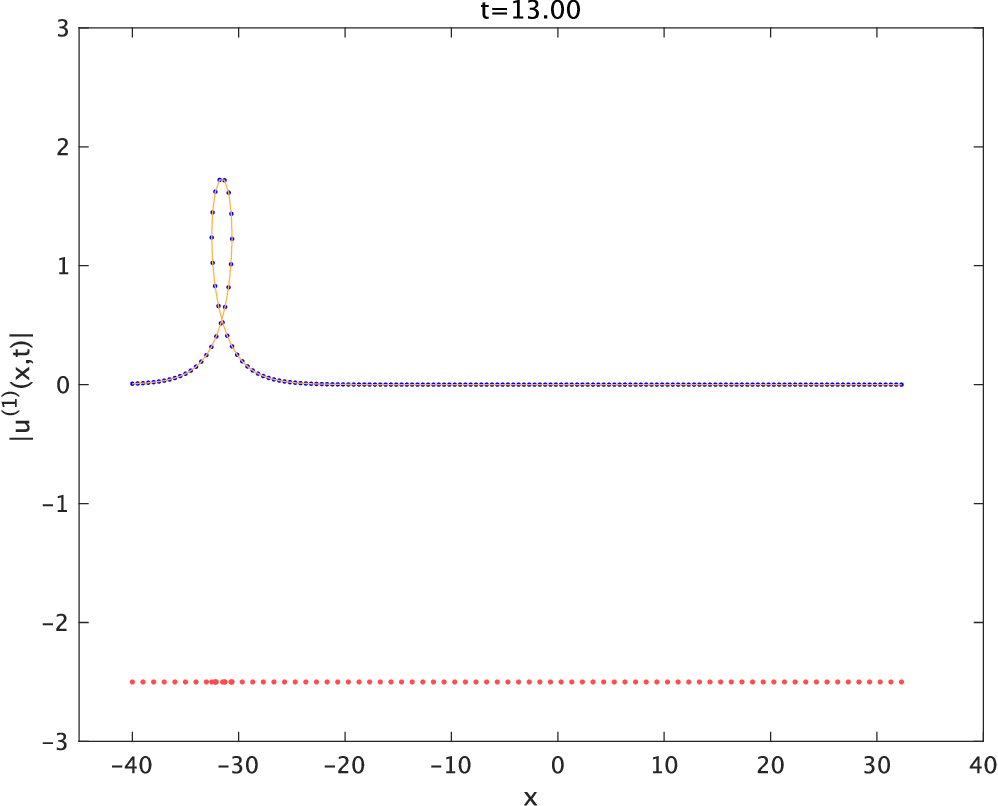}
      \end{minipage} 
       \end{tabular}
     \caption{Numerical simulation of the $|u^{(1)}|$-profile of the one-soliton solution of the 2-CSP equation. ${\rm maxerr}(|u^{(1)}|)=1.67\times 10^{-4}$.}
              \label{2CSP_1_abs_u1}
  \end{figure}

\begin{figure}[htbp]
 \begin{tabular}{cc}
      \begin{minipage}[t]{0.47\hsize}
       \centering
        \includegraphics[keepaspectratio, scale=0.33]{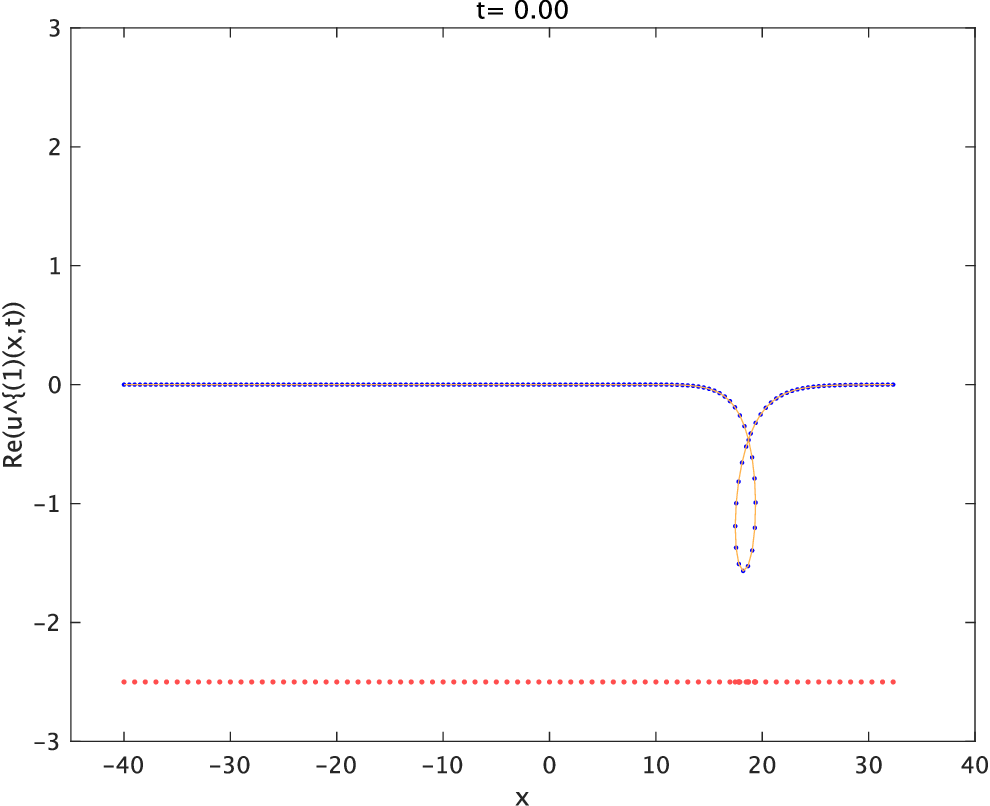}
      \end{minipage} &
      \begin{minipage}[t]{0.47\hsize}
        \centering
        \includegraphics[keepaspectratio, scale=0.33]{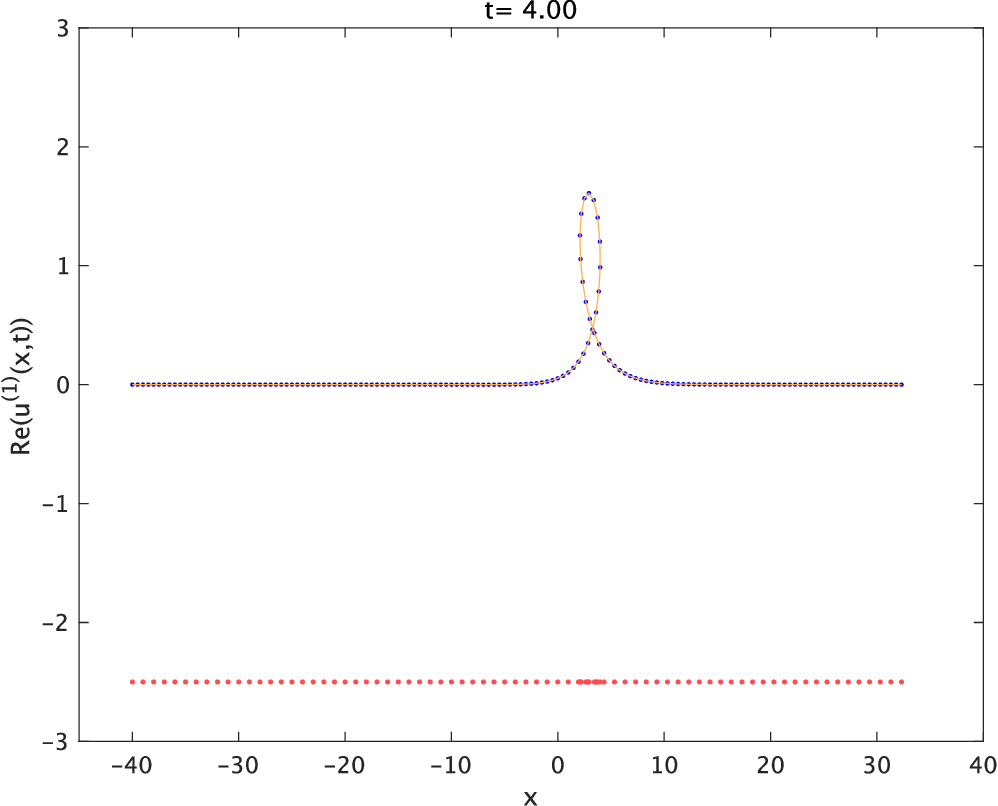}
      \end{minipage}\\ 
  
      \begin{minipage}[t]{0.47\hsize}
        \centering
        \includegraphics[keepaspectratio, scale=0.33]{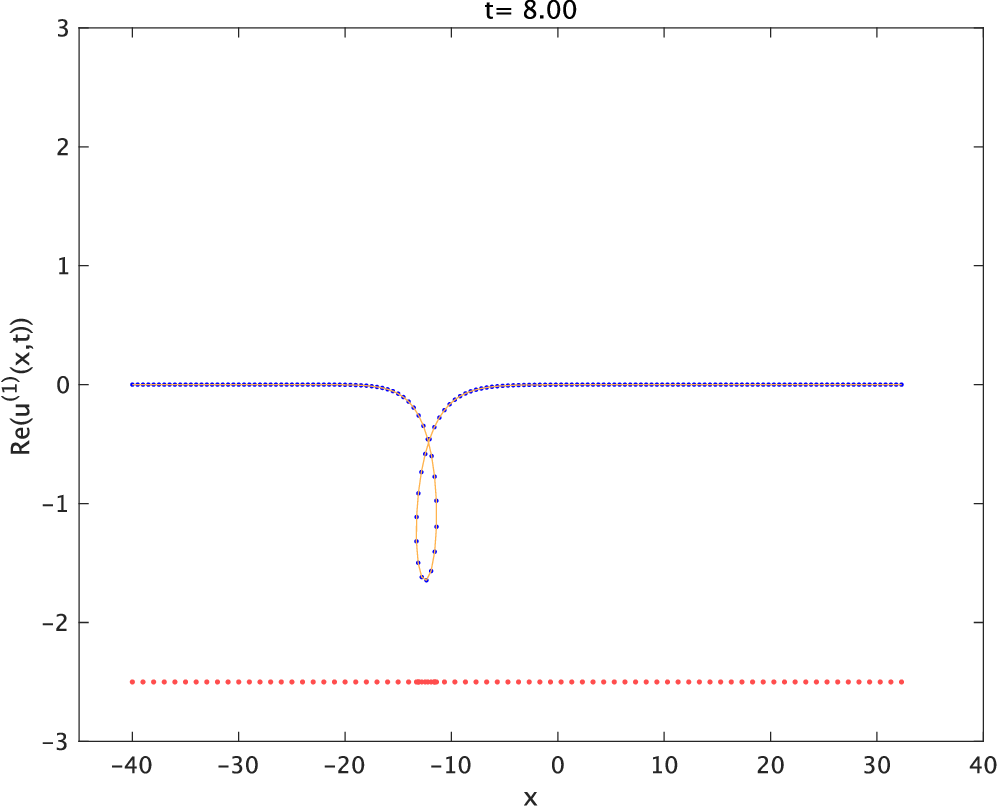}
      \end{minipage} &
      \begin{minipage}[t]{0.47\hsize}
        \centering
        \includegraphics[keepaspectratio, scale=0.33]{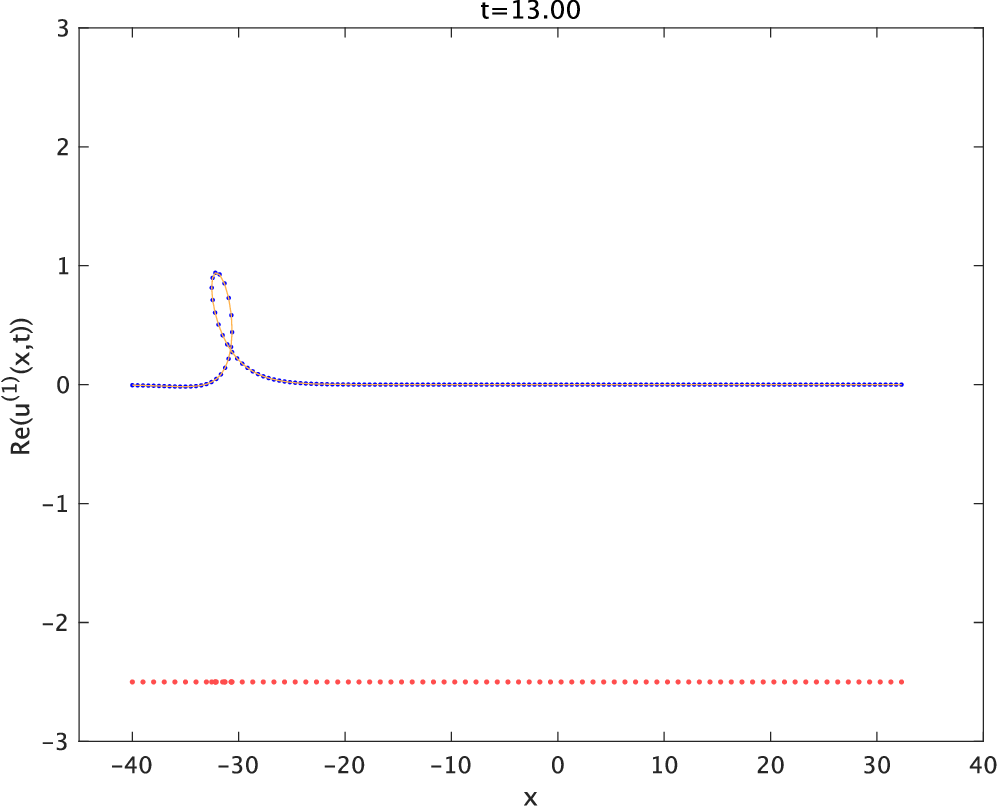}
      \end{minipage} 
       \end{tabular}
     \caption{Numerical simulation of the ${\rm Re}(u^{(1)})$-profile of the one-soliton solution of the 2-CSP equation. ${\rm maxerr}({\rm Re}(u^{(1)}))=2.87\times 10^{-3}$.}
              \label{2CSP_1_re_u1}
  \end{figure}

\begin{figure}[htbp]
 \begin{tabular}{cc}
      \begin{minipage}[t]{0.47\hsize}
       \centering
        \includegraphics[keepaspectratio, scale=0.33]{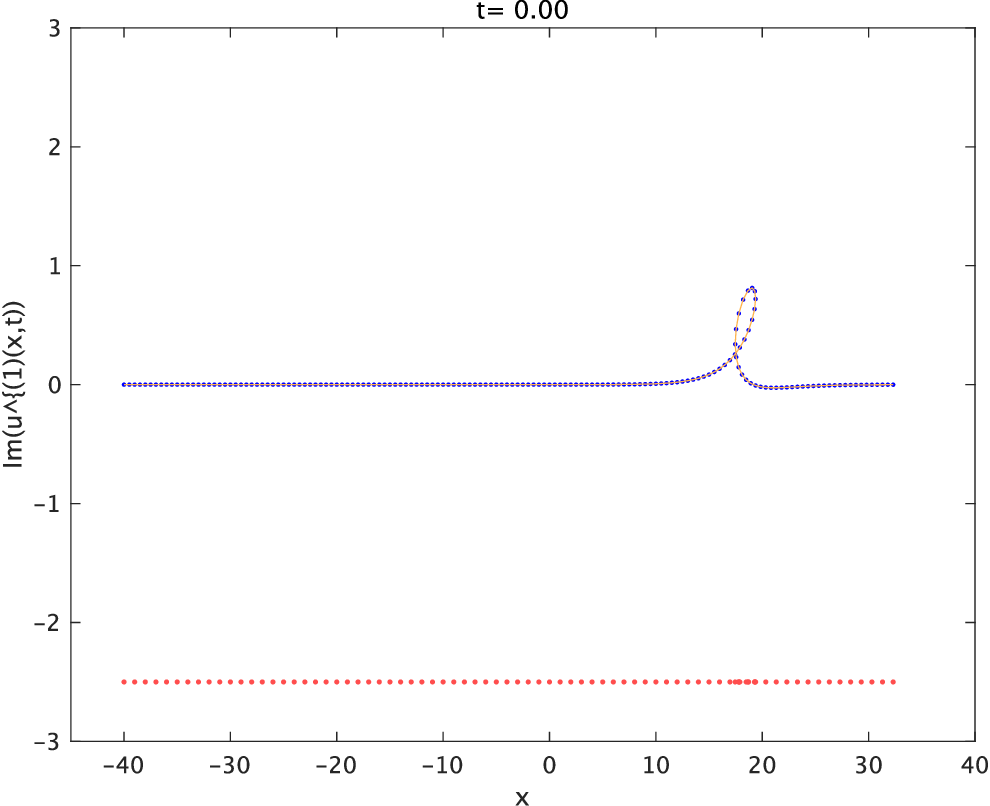}
      \end{minipage} &
      \begin{minipage}[t]{0.47\hsize}
        \centering
        \includegraphics[keepaspectratio, scale=0.33]{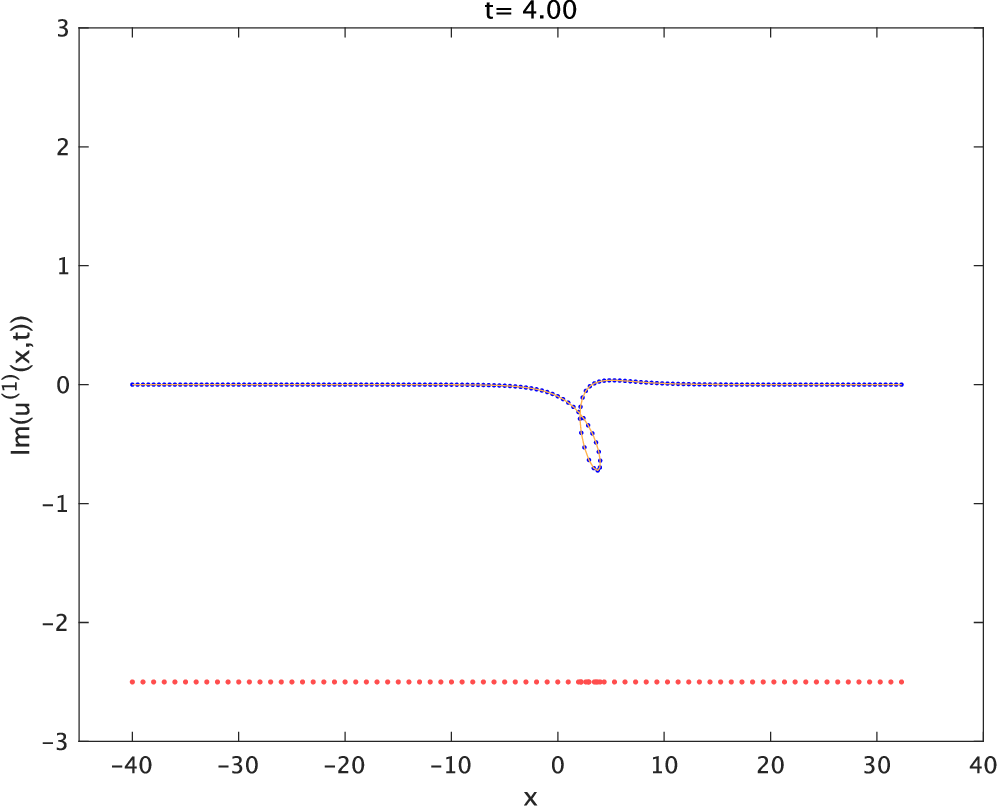}
      \end{minipage}\\ 
  
      \begin{minipage}[t]{0.47\hsize}
        \centering
        \includegraphics[keepaspectratio, scale=0.33]{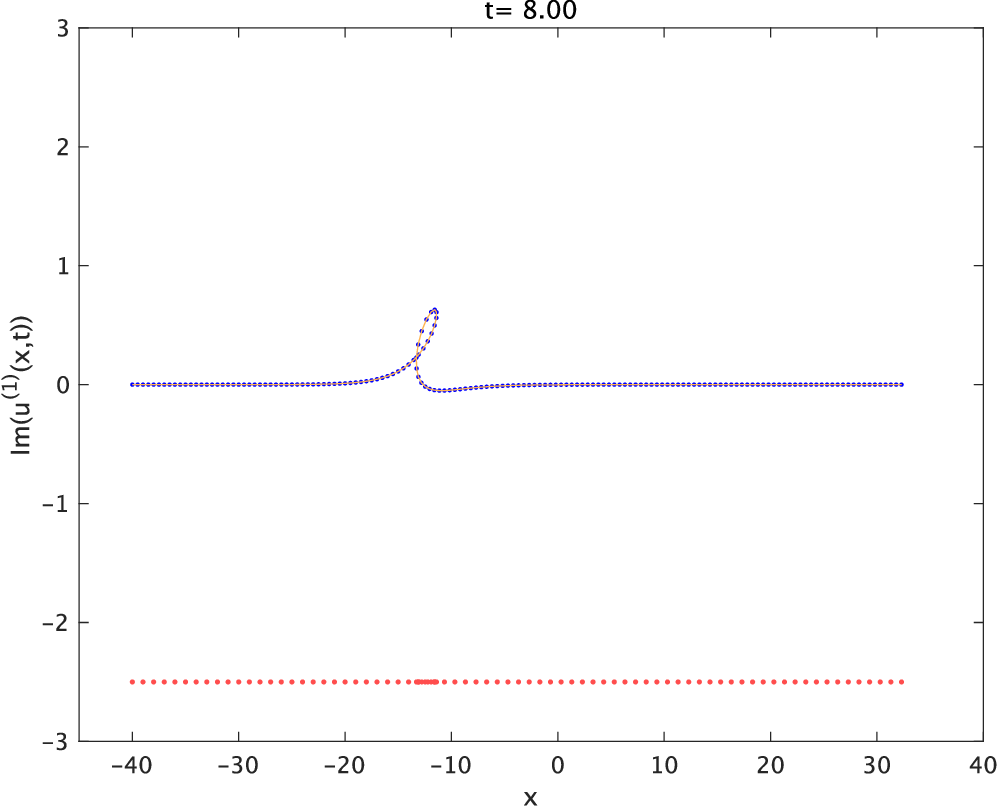}
      \end{minipage} &
      \begin{minipage}[t]{0.47\hsize}
        \centering
        \includegraphics[keepaspectratio, scale=0.33]{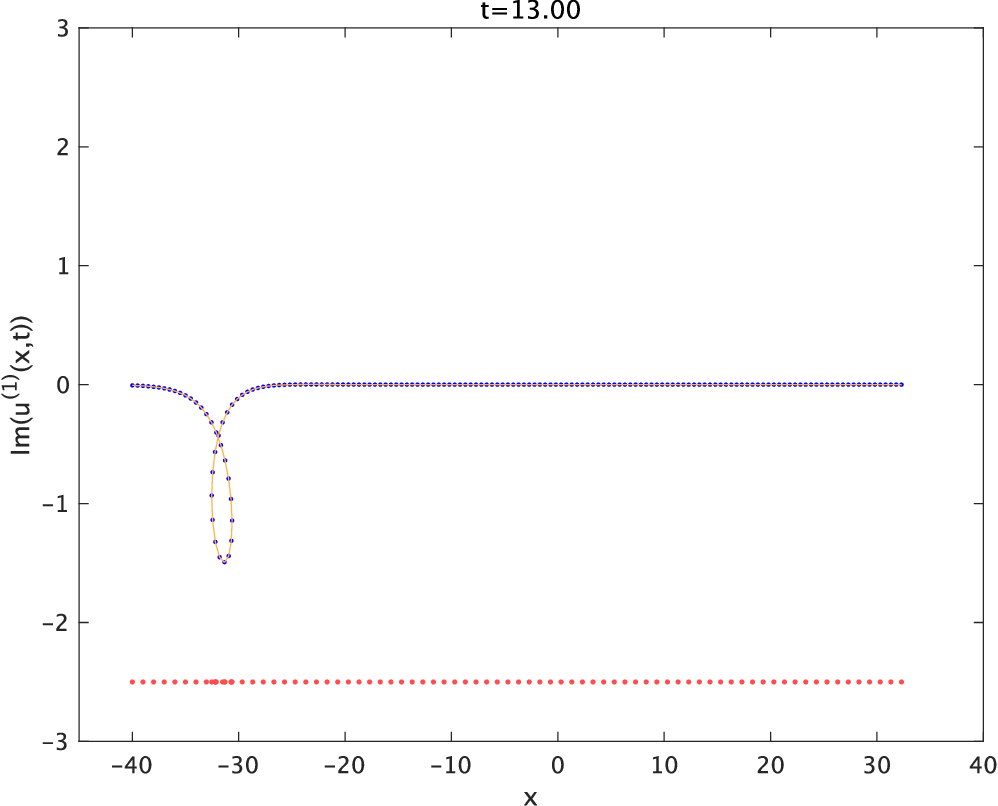}
      \end{minipage} 
       \end{tabular}
     \caption{Numerical simulation of the ${\rm Im}(u^{(1)})$-profile of the one-soliton solution of the 2-CSP equation. ${\rm maxerr}({\rm Im}(u^{(1)}))=4.22\times 10^{-3}$.}
              \label{2CSP_1_im_u1}
  \end{figure}

\begin{figure}[htbp]
 \begin{tabular}{cc}
      \begin{minipage}[t]{0.47\hsize}
       \centering
        \includegraphics[keepaspectratio, scale=0.33]{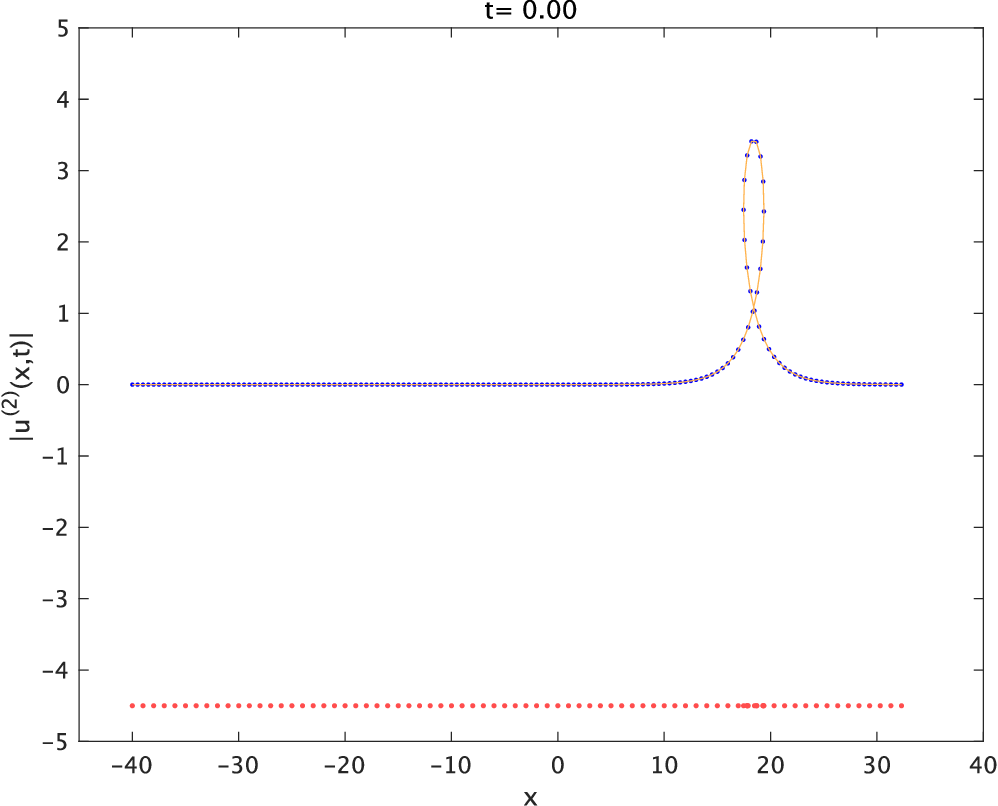}
      \end{minipage} &
      \begin{minipage}[t]{0.47\hsize}
        \centering
        \includegraphics[keepaspectratio, scale=0.33]{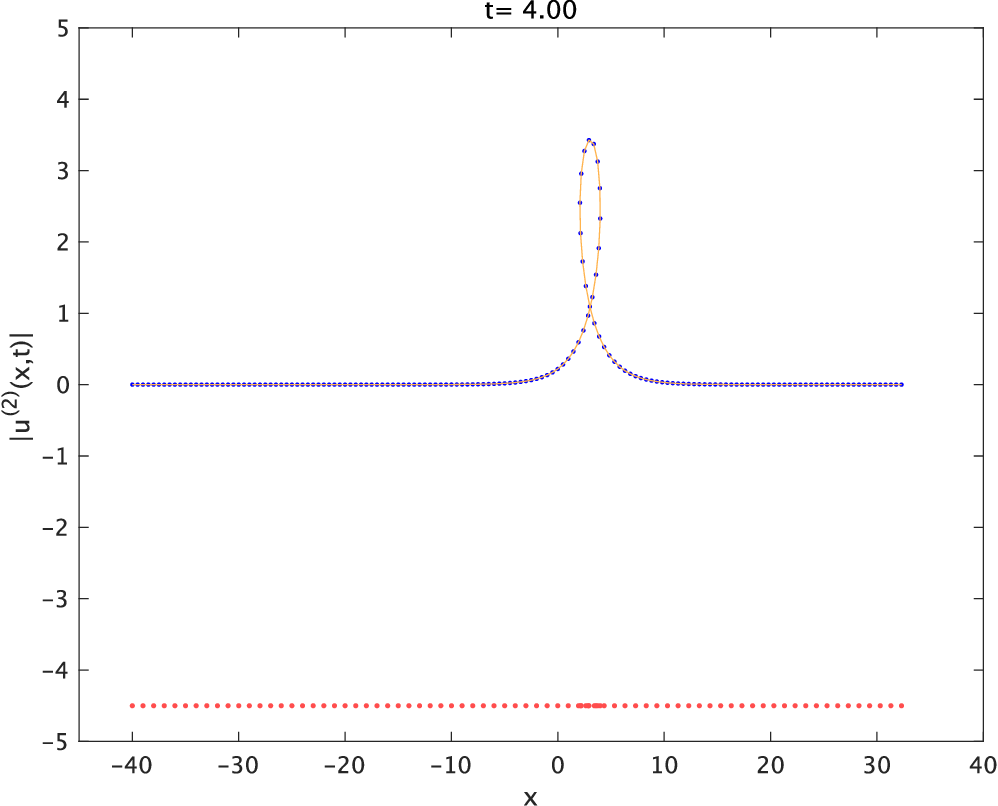}
      \end{minipage}\\ 
  
      \begin{minipage}[t]{0.47\hsize}
        \centering
        \includegraphics[keepaspectratio, scale=0.33]{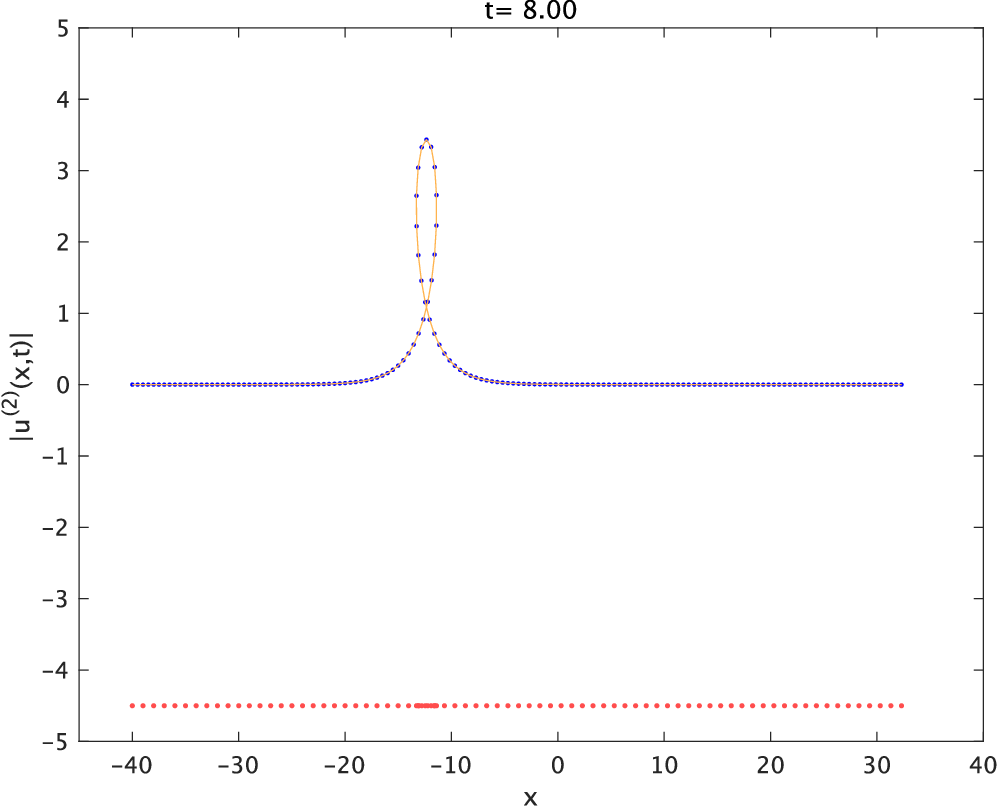}
      \end{minipage} &
      \begin{minipage}[t]{0.47\hsize}
        \centering
        \includegraphics[keepaspectratio, scale=0.33]{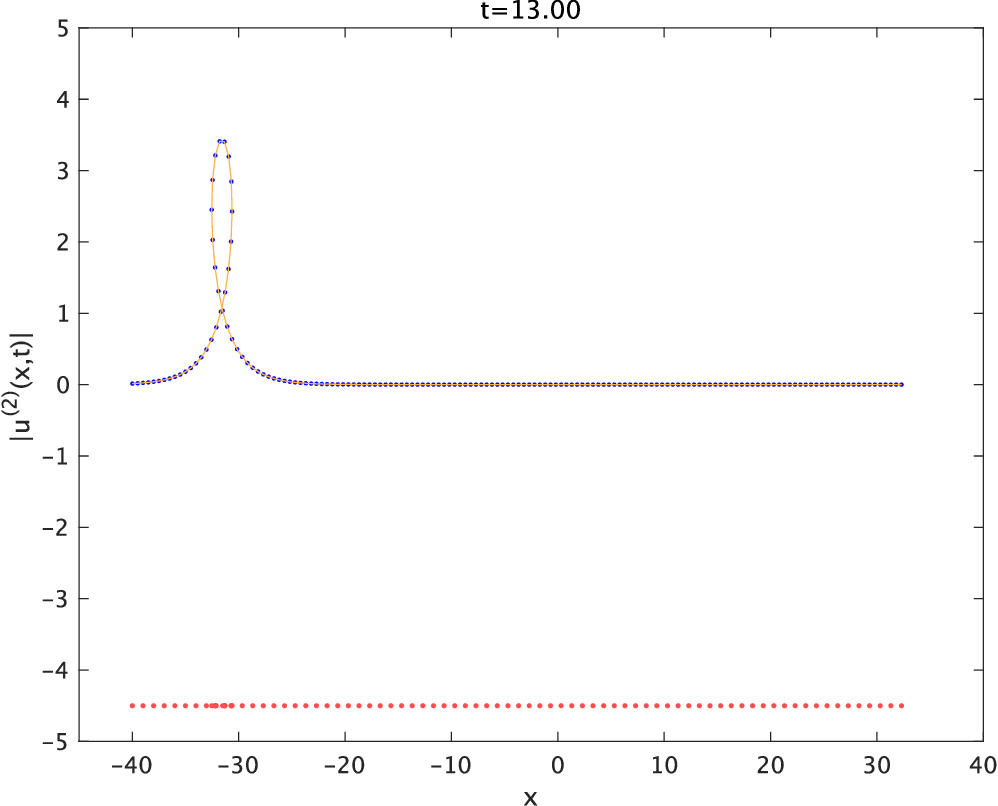}
      \end{minipage} 
       \end{tabular}
     \caption{Numerical simulation of the $|u^{(2)}|$-profile of the one-soliton solution of the 2-CSP equation. ${\rm maxerr}(|u^{(2)}|)=1.67\times 10^{-4}$.}
              \label{2CSP_1_abs_u2}
  \end{figure}

\begin{figure}[htbp]
 \begin{tabular}{cc}
      \begin{minipage}[t]{0.47\hsize}
       \centering
        \includegraphics[keepaspectratio, scale=0.33]{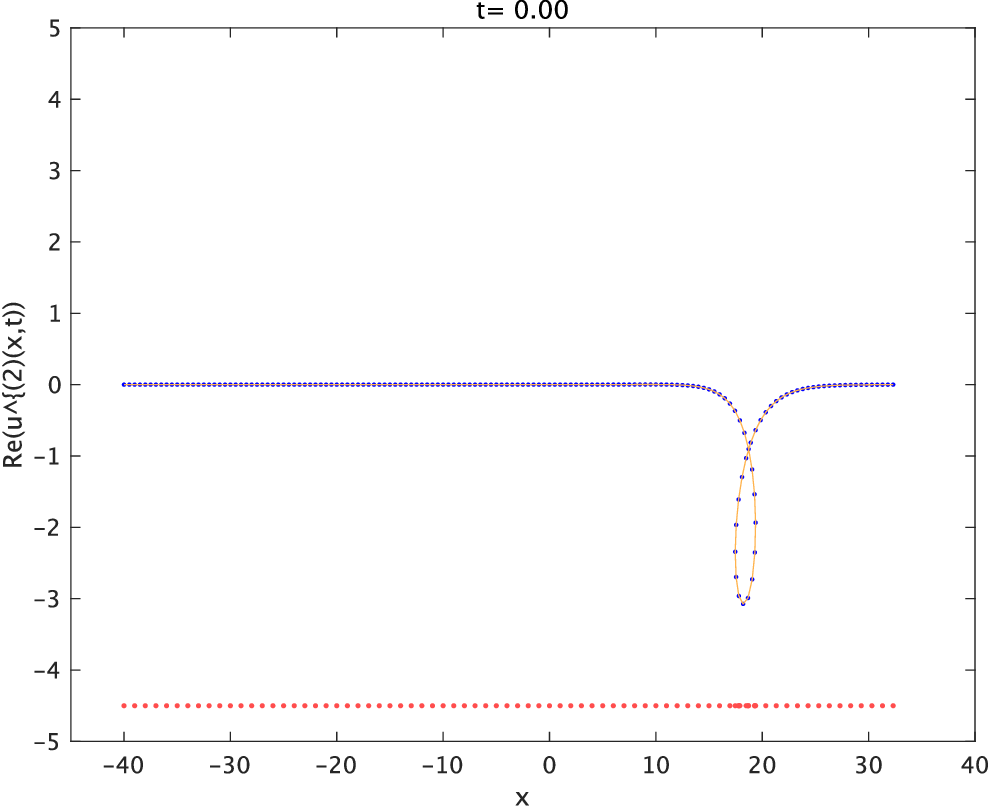}
      \end{minipage} &
      \begin{minipage}[t]{0.47\hsize}
        \centering
        \includegraphics[keepaspectratio, scale=0.33]{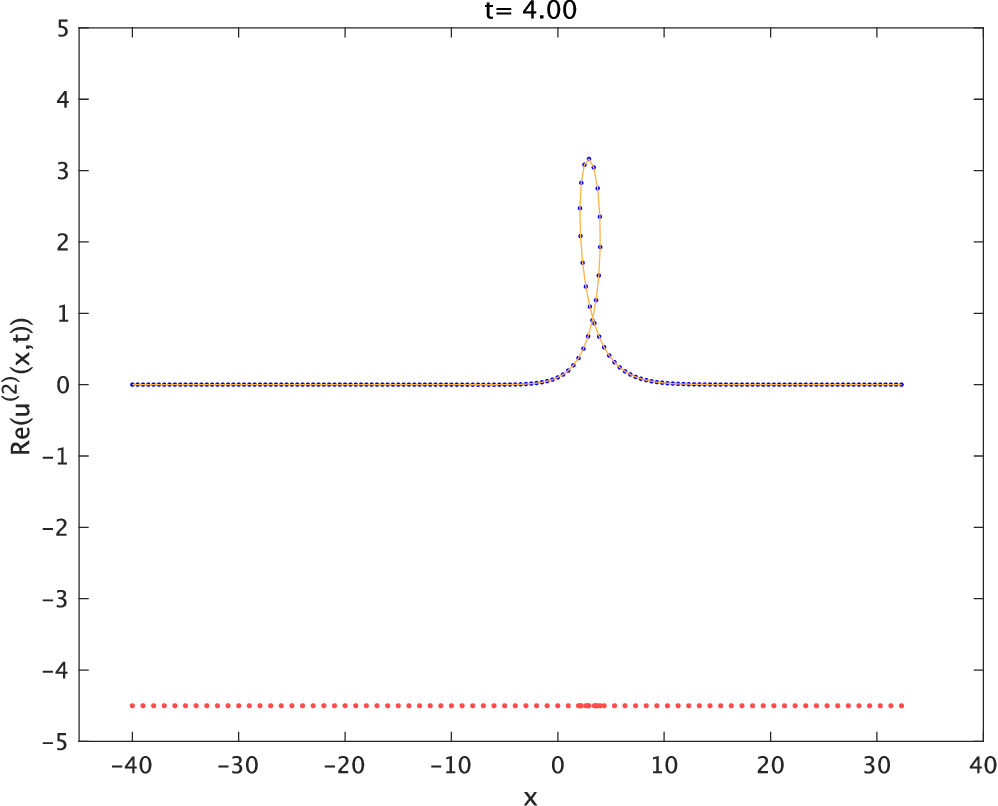}
      \end{minipage}\\ 
  
      \begin{minipage}[t]{0.47\hsize}
        \centering
        \includegraphics[keepaspectratio, scale=0.33]{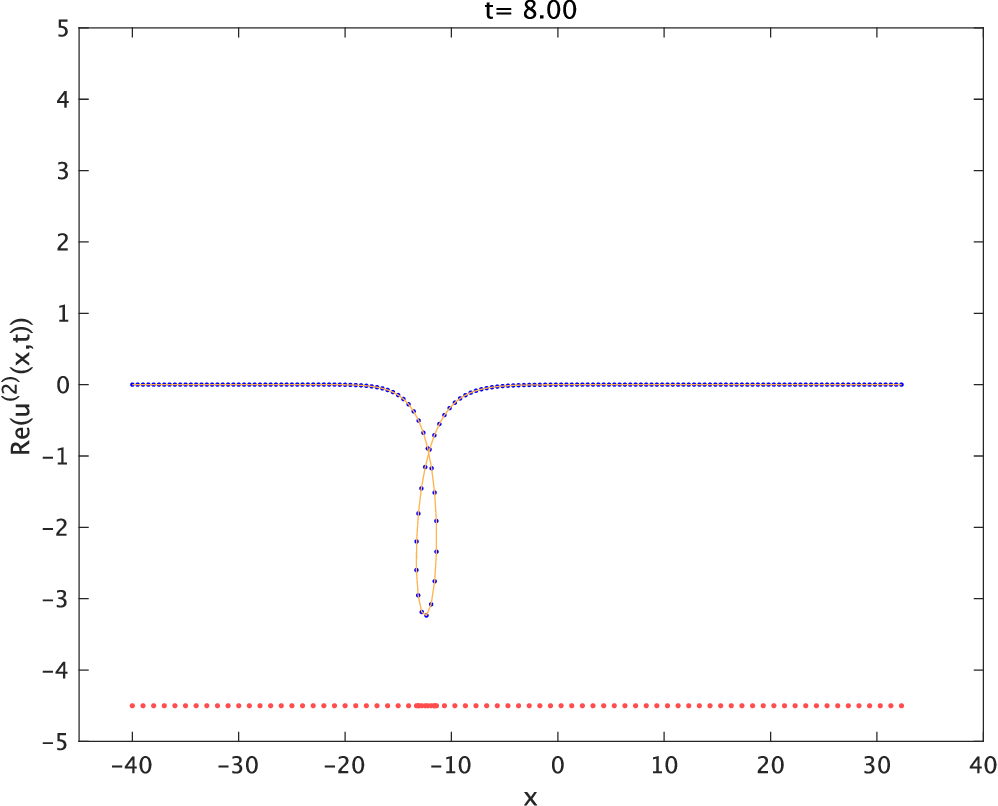}
      \end{minipage} &
      \begin{minipage}[t]{0.47\hsize}
        \centering
        \includegraphics[keepaspectratio, scale=0.33]{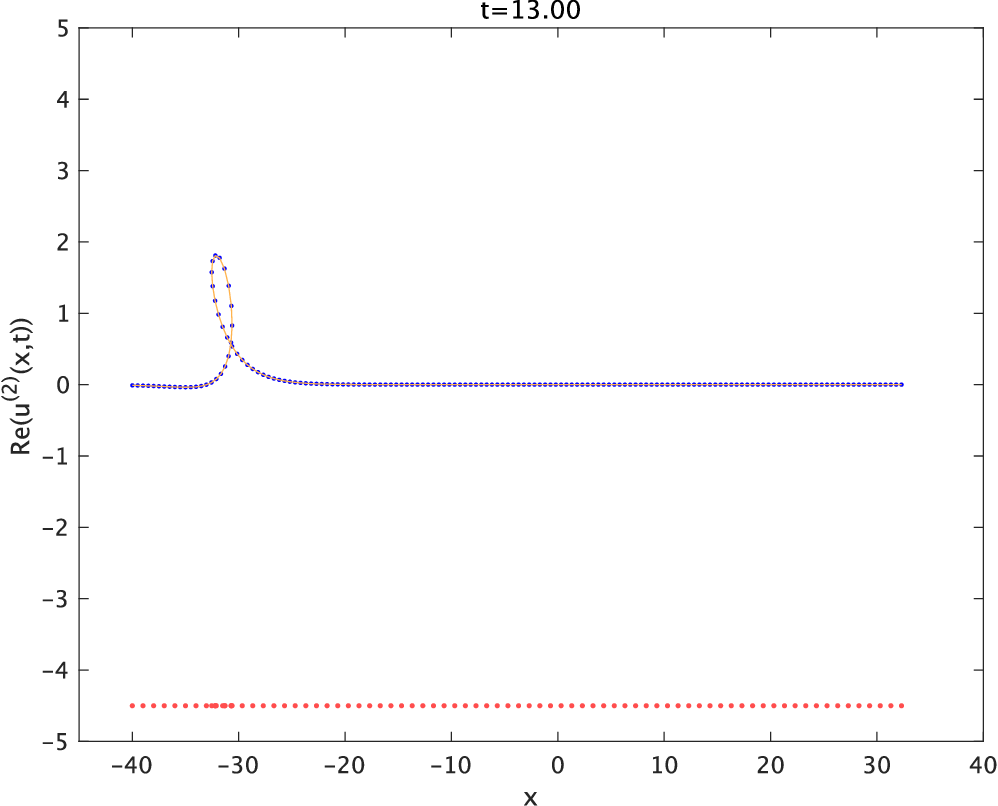}
      \end{minipage} 
       \end{tabular}
     \caption{Numerical simulation of the ${\rm Re}(u^{(2)})$-profile of the one-soliton solution of the 2-CSP equation. ${\rm maxerr}({\rm Re}(u^{(2)}))=2.85\times 10^{-3}$.}
              \label{2CSP_1_re_u2}
  \end{figure}

\begin{figure}[htbp]
 \begin{tabular}{cc}
      \begin{minipage}[t]{0.47\hsize}
       \centering
        \includegraphics[keepaspectratio, scale=0.33]{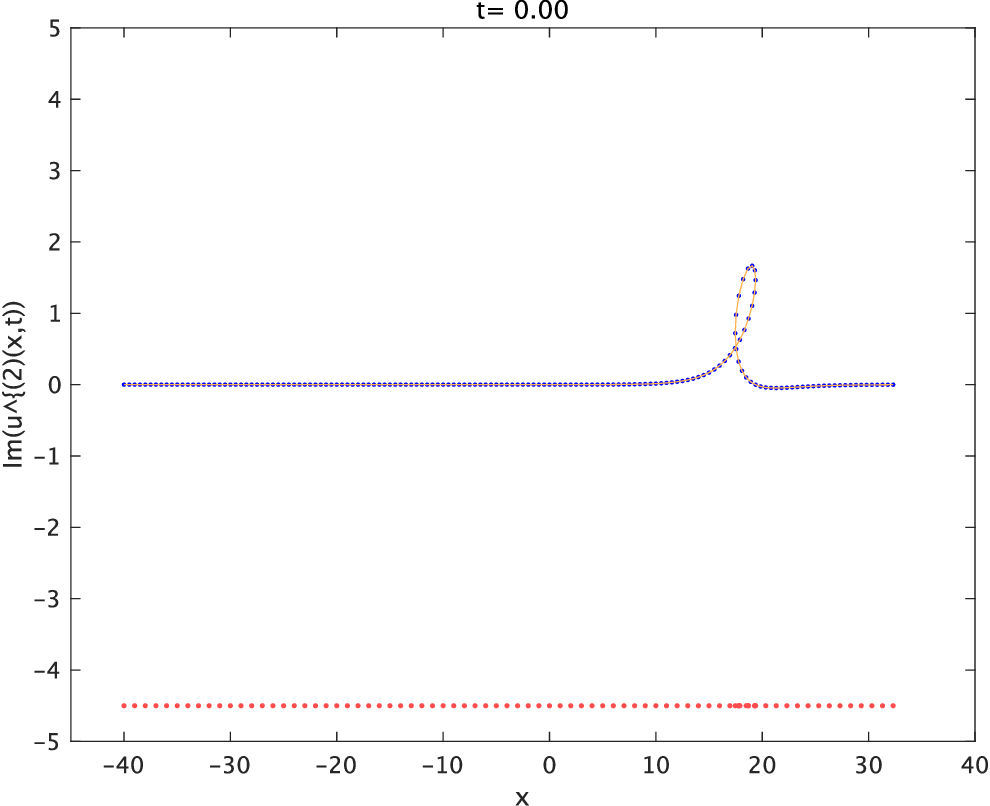}
      \end{minipage} &
      \begin{minipage}[t]{0.47\hsize}
        \centering
        \includegraphics[keepaspectratio, scale=0.33]{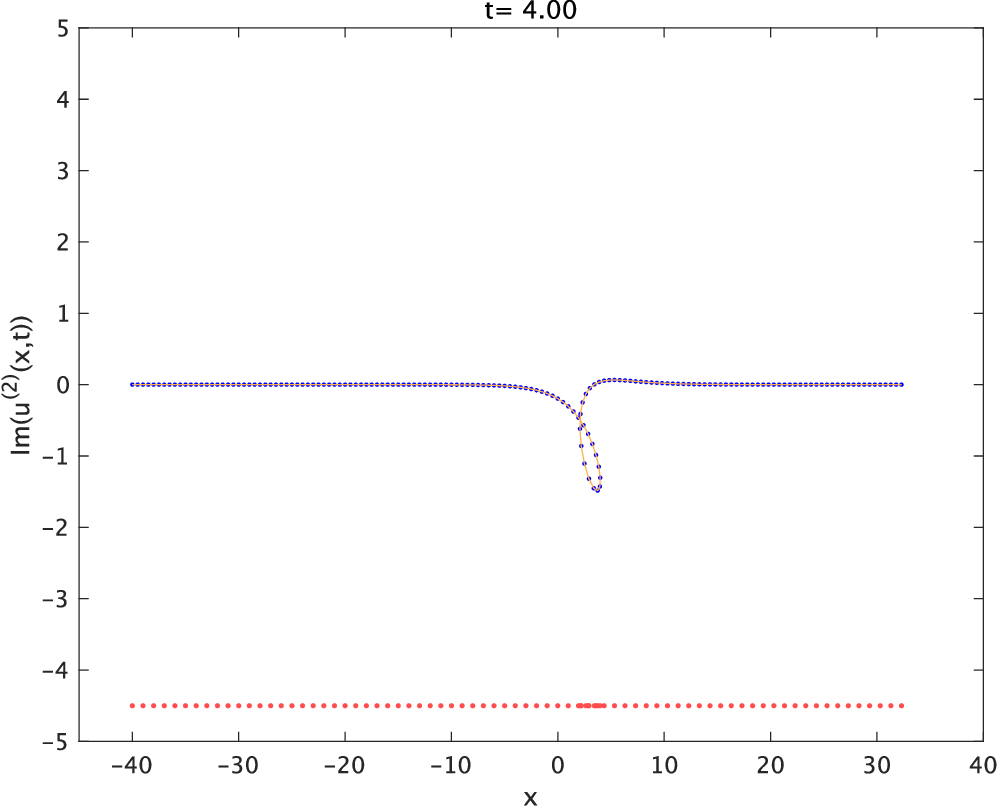}
      \end{minipage}\\ 
  
      \begin{minipage}[t]{0.47\hsize}
        \centering
        \includegraphics[keepaspectratio, scale=0.33]{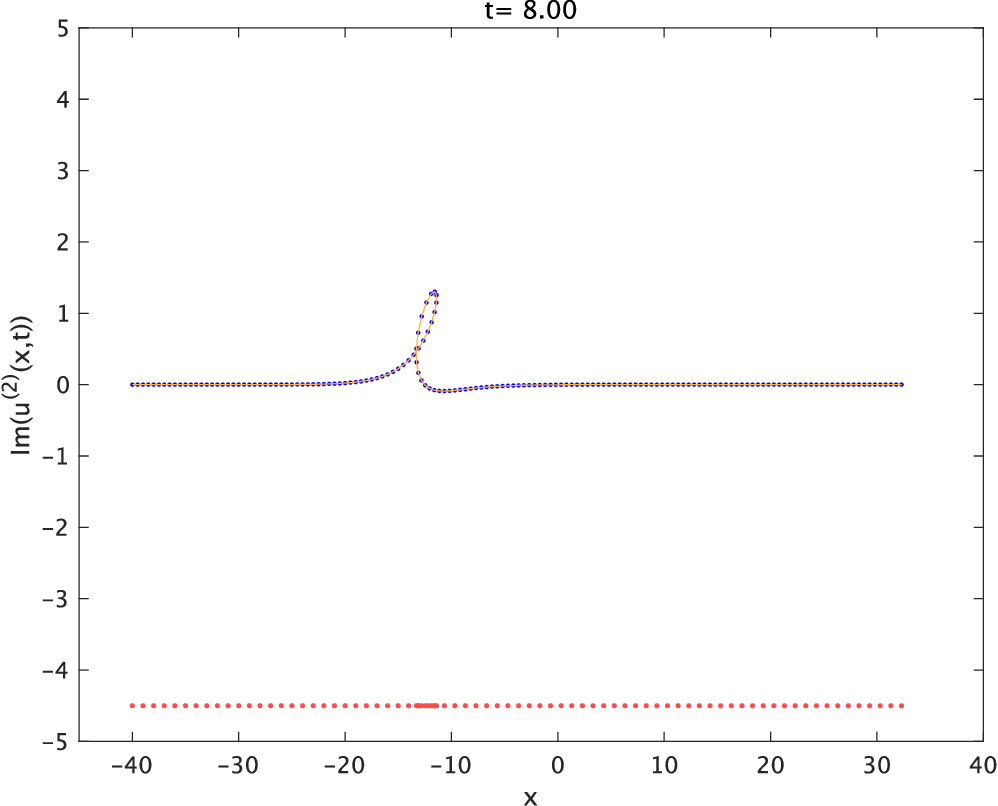}
      \end{minipage} &
      \begin{minipage}[t]{0.47\hsize}
        \centering
        \includegraphics[keepaspectratio, scale=0.33]{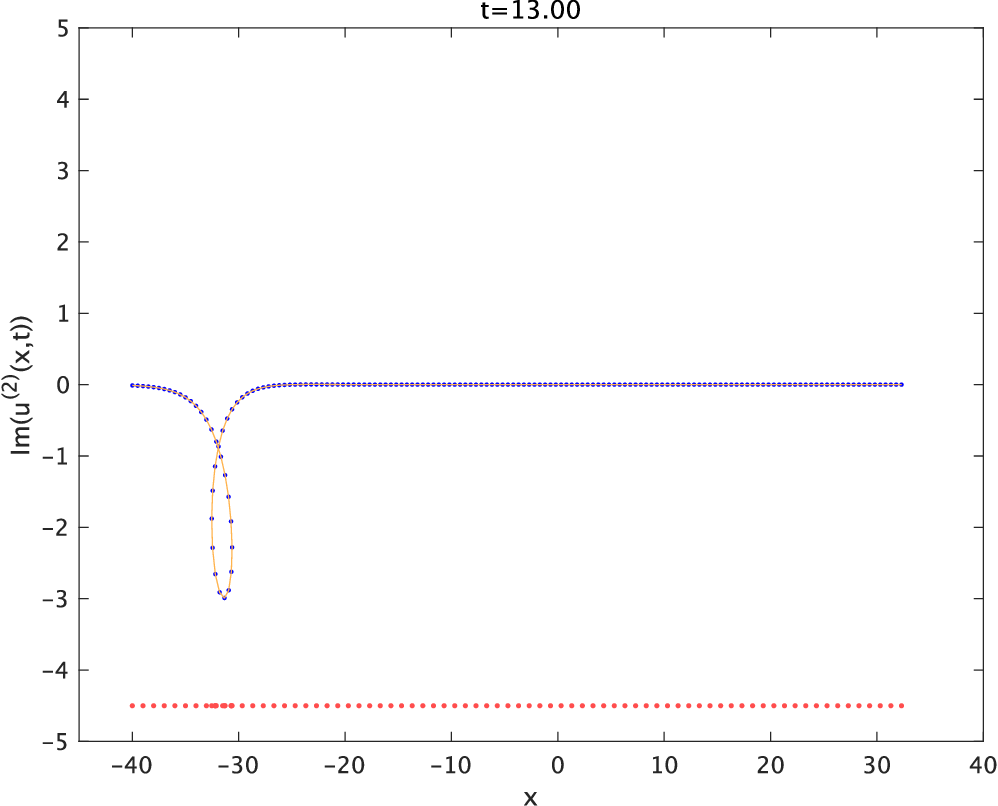}
      \end{minipage} 
       \end{tabular}
     \caption{Numerical simulation of the ${\rm Im}(u^{(2)})$-profile of the one-soliton solution of the 2-CSP equation. ${\rm maxerr}({\rm Im}(u^{(2)}))=4.21\times 10^{-3}$.}
              \label{2CSP_1_im_u2}
  \end{figure}

\noindent
Two-soliton:
Using the notation in (\ref{con_compact_coefficients}), the constraints for the
2-CSP equation imply
\begin{eqnarray*}
\fl C_{i3}=\sum_{\alpha=1}^{2}a_i^{(\alpha)}(a_2^{(\alpha)})^*,\qquad
C_{i4}=\sum_{\alpha=1}^{2}a_i^{(\alpha)}(a_1^{(\alpha)})^*,\qquad i=1,2,\\
\fl \varphi_3=\varphi_2^*,\qquad \varphi_4=\varphi_1^*.
\end{eqnarray*}
 \begin{eqnarray}
 \fl u^{(\mu)}&=&\frac{g^{(\mu)}}{f},\qquad x=X-2(\log f)_T,\qquad t=T\nonumber\\
\fl f&=&1+\tilde{C}_{14}\varphi_1\varphi_4+\tilde{C}_{13}\varphi_1\varphi_3
+\tilde{C}_{24}\varphi_2\varphi_4+\tilde{C}_{23}\varphi_2\varphi_3\nonumber\\
\fl &&+\Delta\left(\prod_{1\leq i<j\leq4}R_{ij}\right)
\varphi_1\varphi_2\varphi_3\varphi_4,\nonumber\\
\fl g^{(\mu)}&=&a_{1}^{(\mu)}\left(\varphi_1+R_{12}R_{14}\tilde{C}_{24}\varphi_1\varphi_2\varphi_4
+R_{12}R_{13}\tilde{C}_{23}\varphi_1\varphi_2\varphi_3\right)\nonumber\\
\fl &&+a_{2}^{(\mu)}\left(\varphi_2-R_{12}R_{24}\tilde{C}_{14}\varphi_1\varphi_2\varphi_4
-R_{12}R_{23}\tilde{C}_{13}\varphi_1\varphi_2\varphi_3\right),\qquad \mu=1,2.
\end{eqnarray}
where $\varphi_{i}=B_{i}e^{p_{i}X+p_{i}^{-1}T}$ and $i=1,2$.

Figures \ref{2CSP_1_abs_u1}--\ref{2CSP_1_im_u2} show the numerical
$|u^{(1)}|$, ${\rm Re}(u^{(1)})$, ${\rm Im}(u^{(1)})$, $|u^{(2)}|$,
${\rm Re}(u^{(2)})$, and ${\rm Im}(u^{(2)})$ profiles for the one-soliton
solution with $p_{1}=0.5+0.1\mathrm{i}$, $a^{(1)}_{1}=1+0.5\mathrm{i}$,
$a^{(2)}_{1}=2+0.95\mathrm{i}$, and $B_{1}={\rm exp}(-10)$.
Figures \ref{2CSP_2_abs_u1}--\ref{2CSP_2_im_u2} show the corresponding profiles
for the two-soliton solution with $p_{1}=0.5+0.1\mathrm{i}$,
$p_{2}=2-0.5\mathrm{i}$, $a^{(1)}_{1}=1+0.5\mathrm{i}$,
$a^{(1)}_{2}=2-0.1\mathrm{i}$, $a^{(2)}_{1}=2+0.95\mathrm{i}$,
$a^{(2)}_{2}=4-0.5\mathrm{i}$, and $B_{1}=B_{2}={\rm exp}(-10)$. In both
cases, the solitons propagate to the left along the $x$-axis.

\begin{figure}[htbp]
 \begin{tabular}{cc}
      \begin{minipage}[t]{0.47\hsize}
       \centering
        \includegraphics[keepaspectratio, scale=0.33]{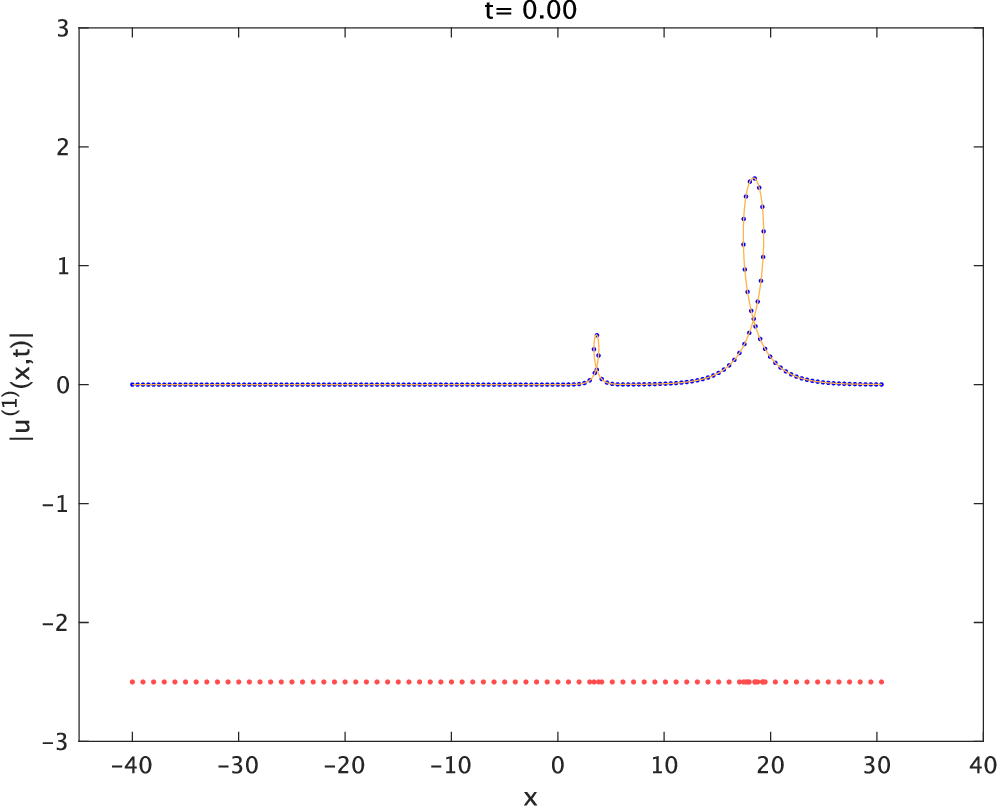}
      \end{minipage} &
      \begin{minipage}[t]{0.47\hsize}
        \centering
        \includegraphics[keepaspectratio, scale=0.33]{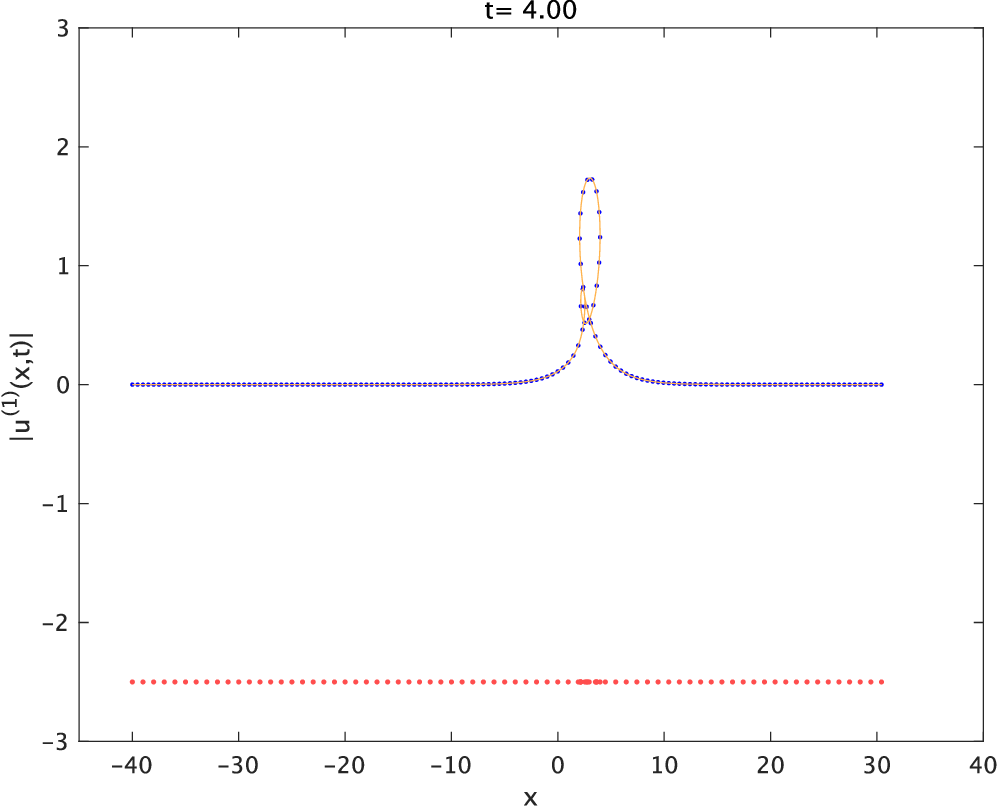}
      \end{minipage}\\ 
  
      \begin{minipage}[t]{0.47\hsize}
        \centering
        \includegraphics[keepaspectratio, scale=0.33]{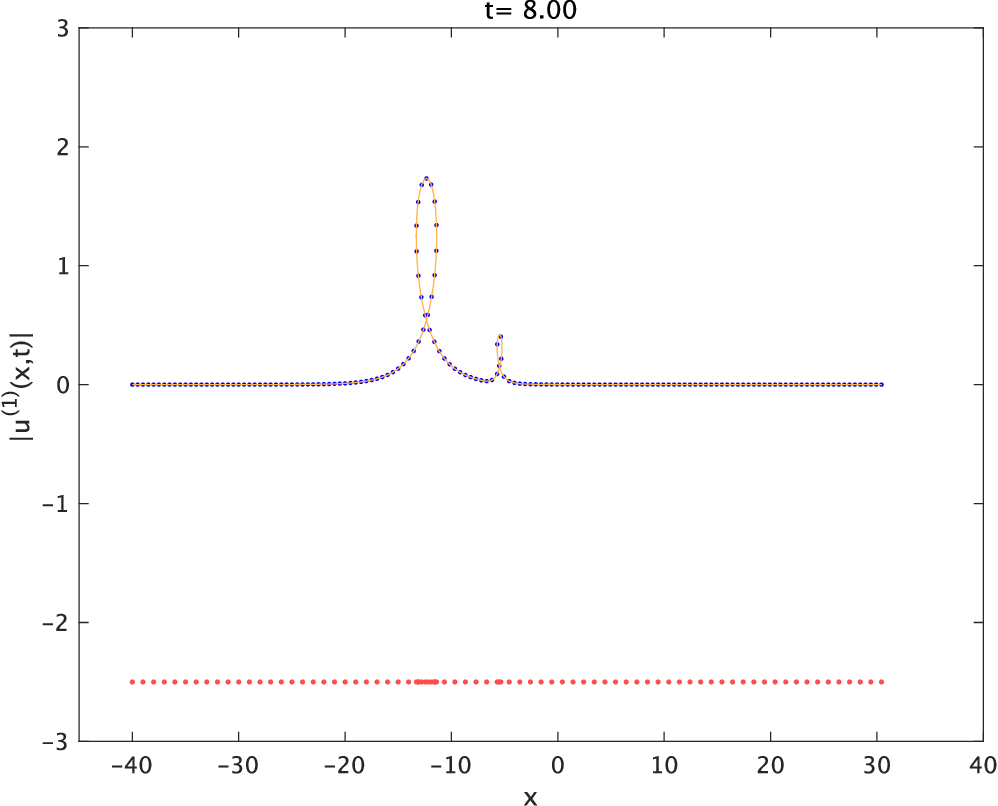}
      \end{minipage} &
      \begin{minipage}[t]{0.47\hsize}
        \centering
        \includegraphics[keepaspectratio, scale=0.33]{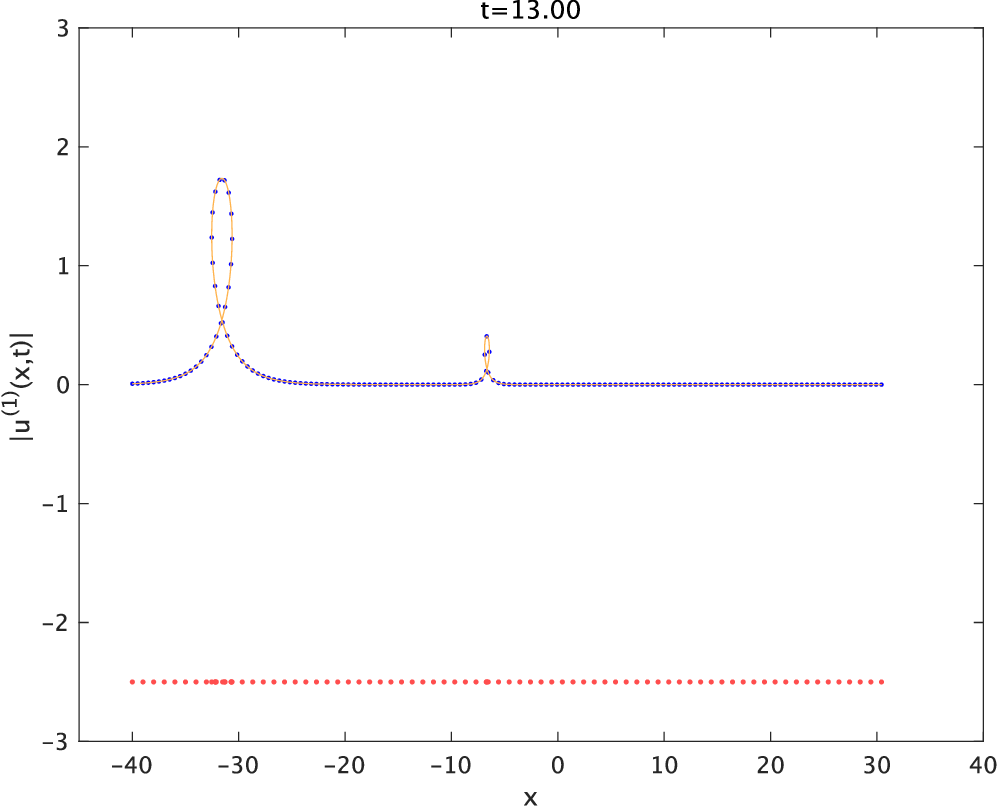}
      \end{minipage} 
       \end{tabular}
     \caption{Numerical simulation of the $|u^{(1)}|$-profile of the two-soliton solution of the 2-CSP equation. ${\rm maxerr}(|u^{(1)}|)=1.63\times 10^{-4}$.}
              \label{2CSP_2_abs_u1}
  \end{figure}

\begin{figure}[htbp]
 \begin{tabular}{cc}
      \begin{minipage}[t]{0.47\hsize}
       \centering
        \includegraphics[keepaspectratio, scale=0.33]{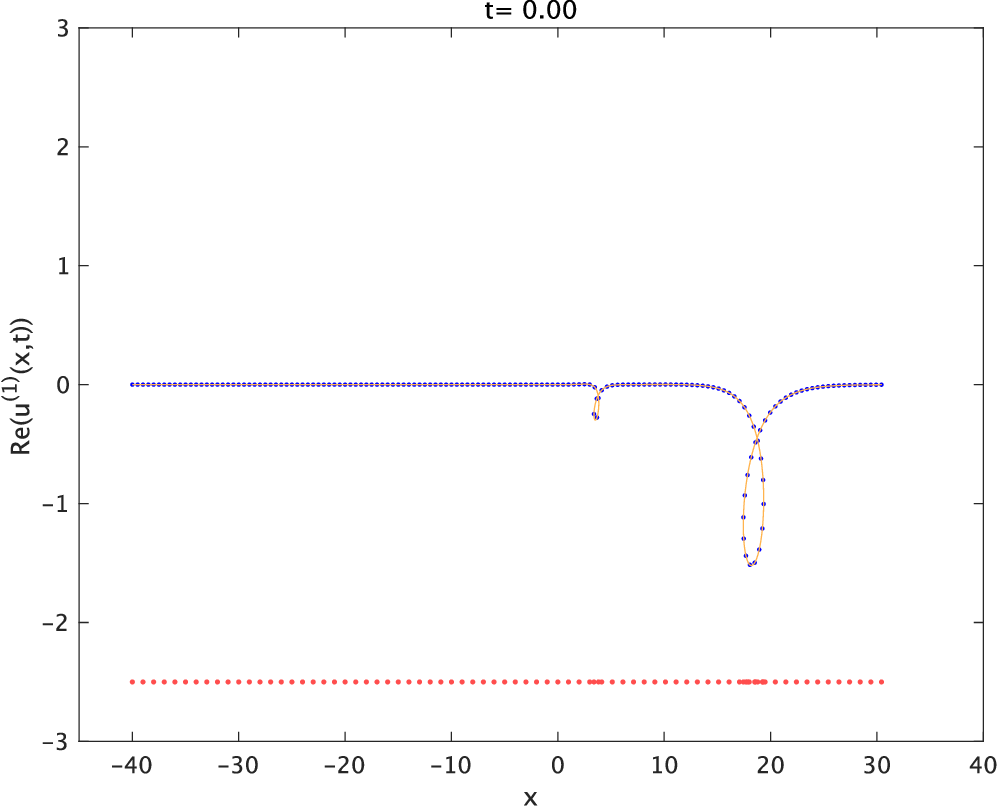}
      \end{minipage} &
      \begin{minipage}[t]{0.47\hsize}
        \centering
        \includegraphics[keepaspectratio, scale=0.33]{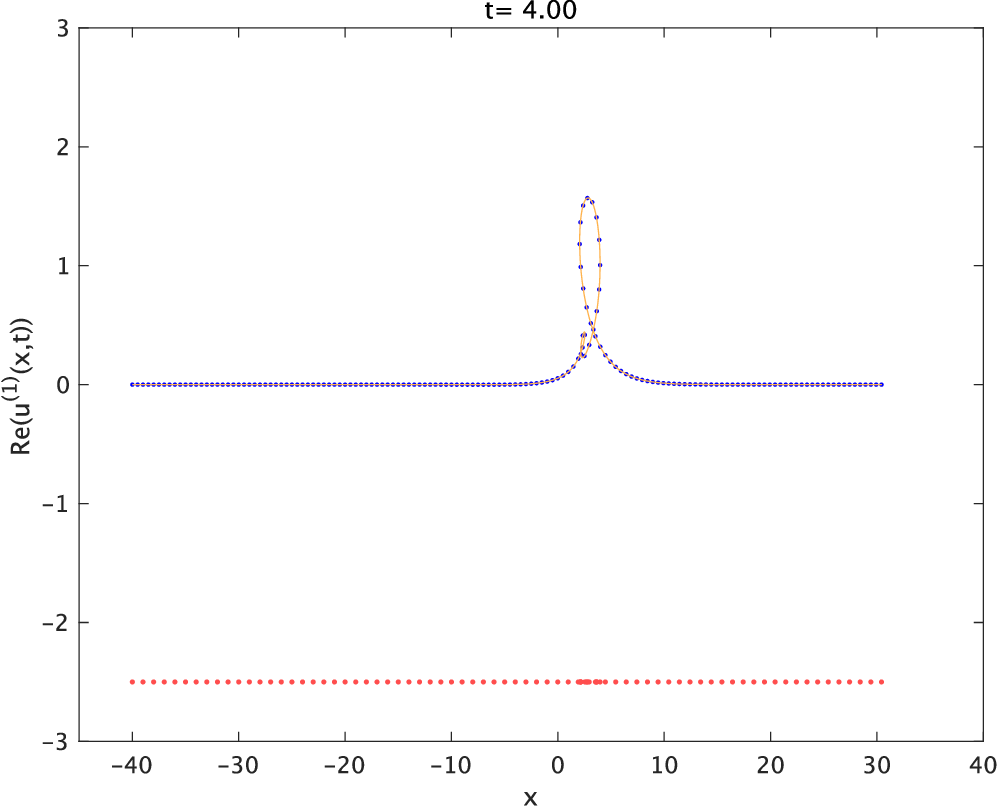}
      \end{minipage}\\ 
  
      \begin{minipage}[t]{0.47\hsize}
        \centering
        \includegraphics[keepaspectratio, scale=0.33]{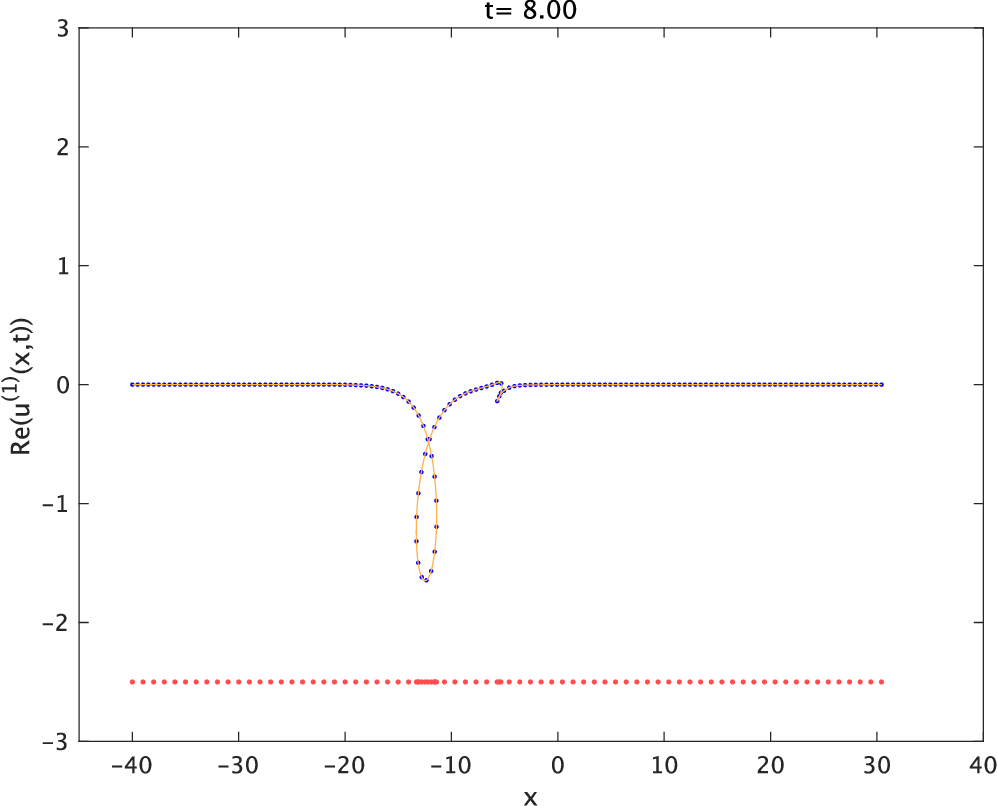}
      \end{minipage} &
      \begin{minipage}[t]{0.47\hsize}
        \centering
        \includegraphics[keepaspectratio, scale=0.33]{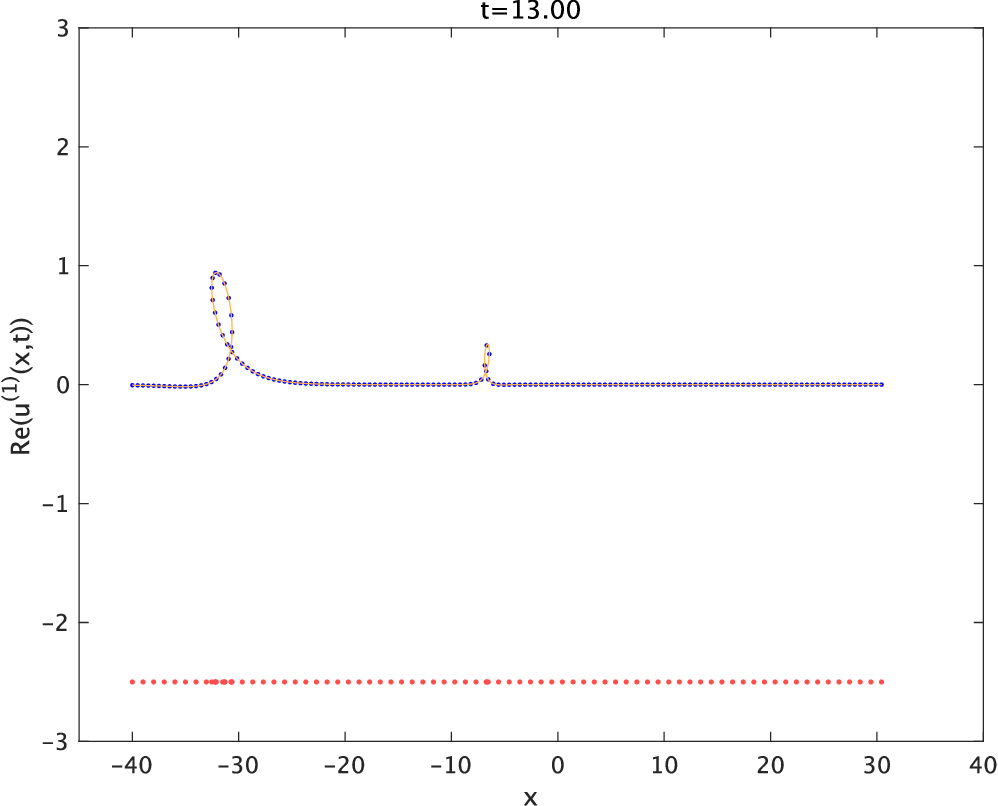}
      \end{minipage} 
       \end{tabular}
     \caption{Numerical simulation of the ${\rm Re}(u^{(1)})$-profile of the two-soliton solution of the 2-CSP equation. ${\rm maxerr}({\rm Re}(u^{(1)}))=2.27\times 10^{-3}$.}
              \label{2CSP_2_re_u1}
  \end{figure}

\begin{figure}[htbp]
 \begin{tabular}{cc}
      \begin{minipage}[t]{0.47\hsize}
       \centering
        \includegraphics[keepaspectratio, scale=0.33]{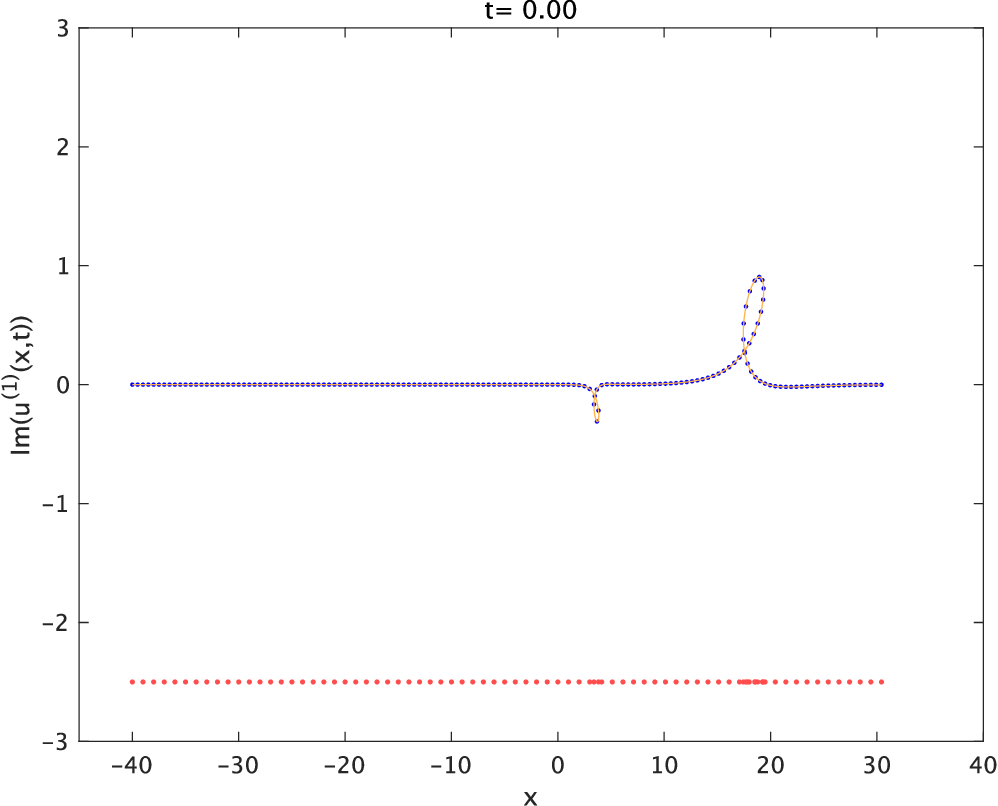}
      \end{minipage} &
      \begin{minipage}[t]{0.47\hsize}
        \centering
        \includegraphics[keepaspectratio, scale=0.33]{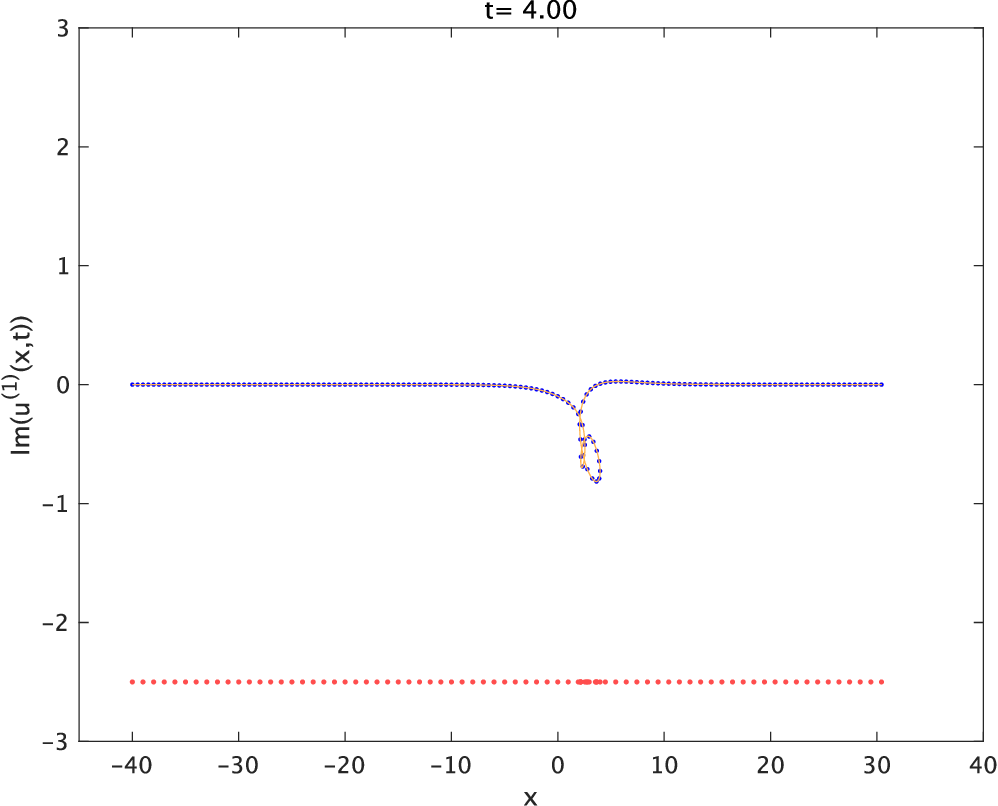}
      \end{minipage}\\ 
  
      \begin{minipage}[t]{0.47\hsize}
        \centering
        \includegraphics[keepaspectratio, scale=0.33]{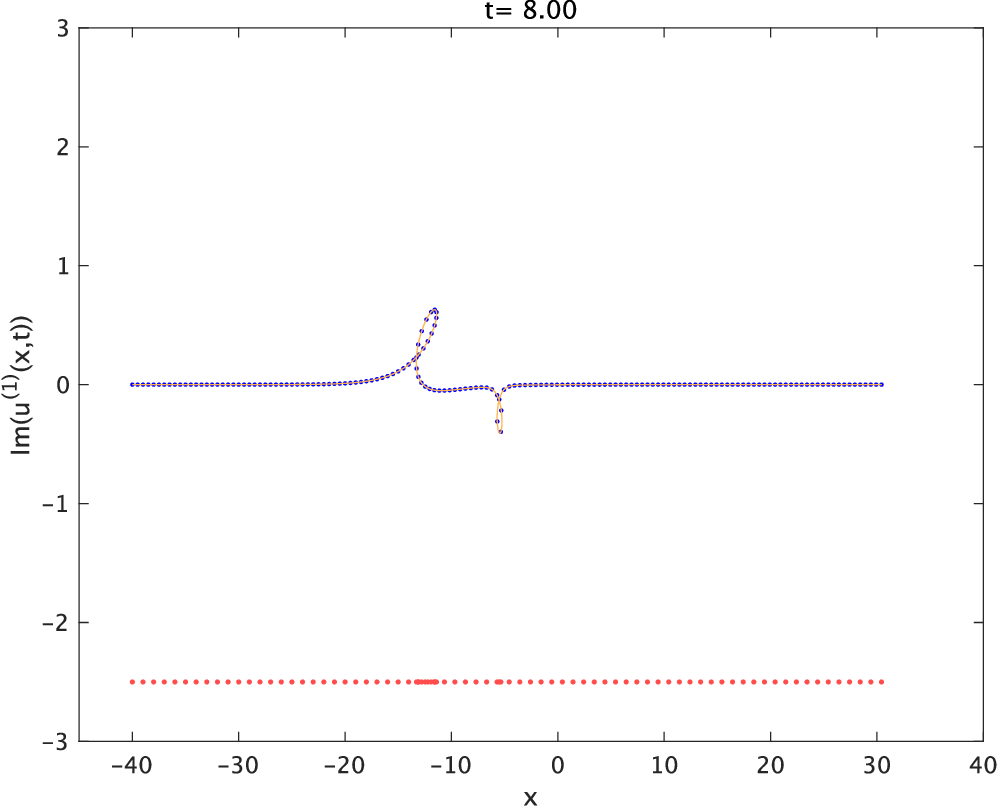}
      \end{minipage} &
      \begin{minipage}[t]{0.47\hsize}
        \centering
        \includegraphics[keepaspectratio, scale=0.33]{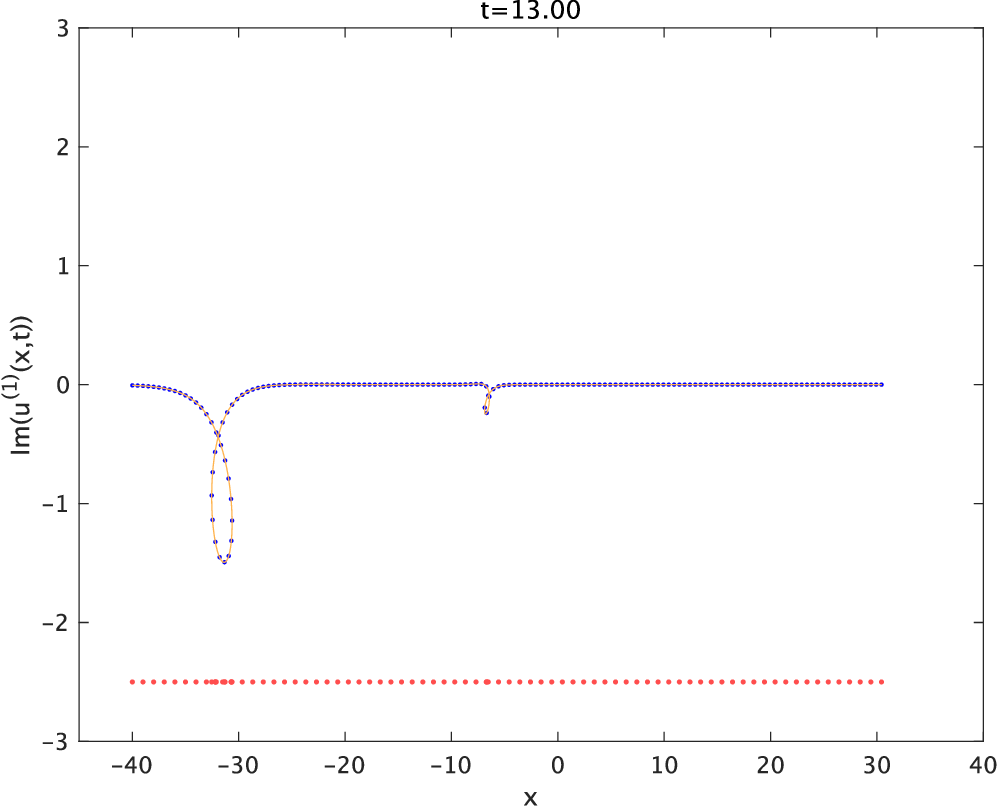}
      \end{minipage} 
       \end{tabular}
     \caption{Numerical simulation of the ${\rm Im}(u^{(1)})$-profile of the two-soliton solution of the 2-CSP equation. ${\rm maxerr}({\rm Im}(u^{(1)}))=2.44\times 10^{-3}$.}
              \label{2CSP_2_im_u1}
  \end{figure}

\begin{figure}[htbp]
 \begin{tabular}{cc}
      \begin{minipage}[t]{0.47\hsize}
       \centering
        \includegraphics[keepaspectratio, scale=0.33]{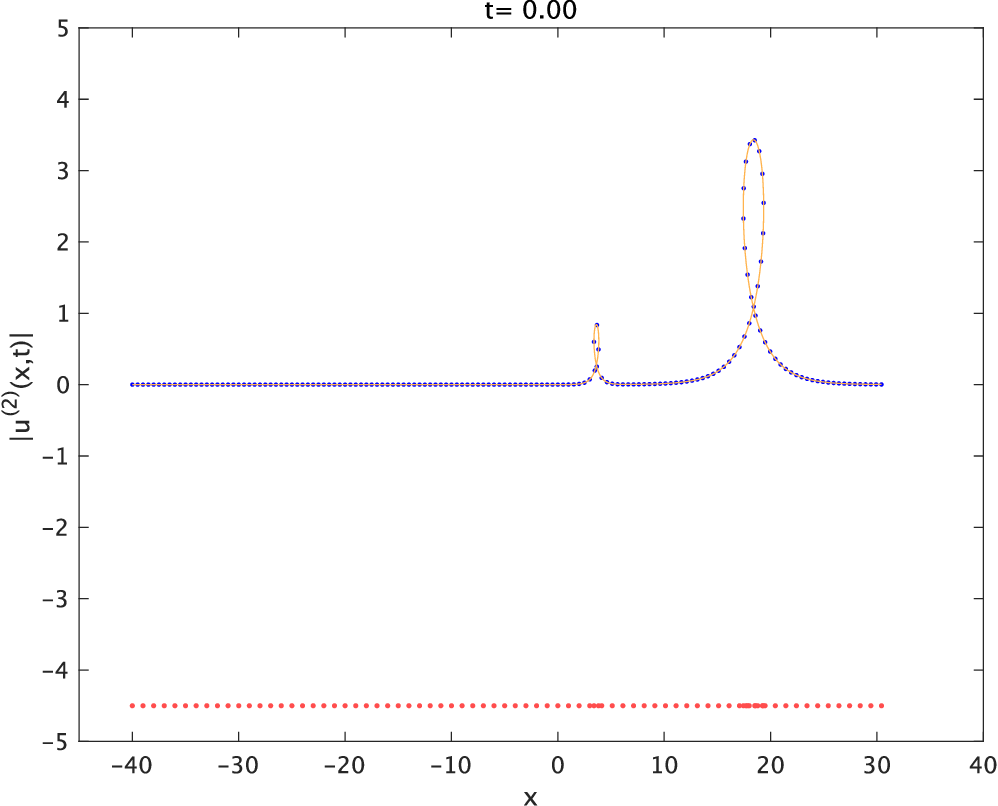}
      \end{minipage} &
      \begin{minipage}[t]{0.47\hsize}
        \centering
        \includegraphics[keepaspectratio, scale=0.33]{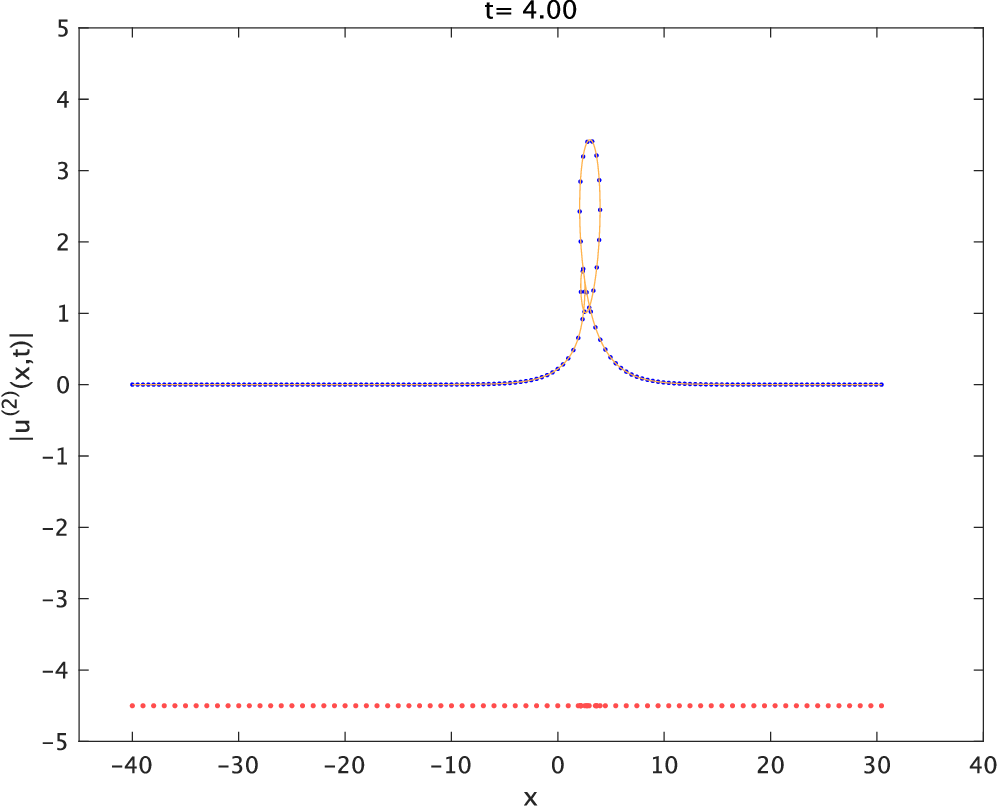}
      \end{minipage}\\ 
  
      \begin{minipage}[t]{0.47\hsize}
        \centering
        \includegraphics[keepaspectratio, scale=0.33]{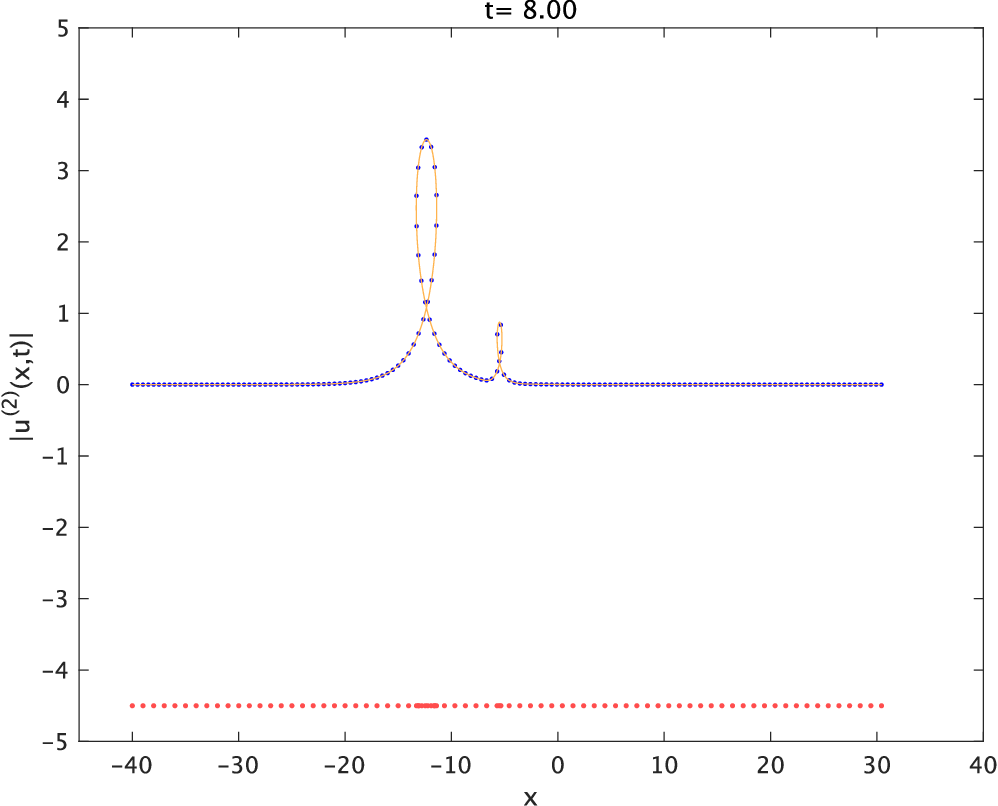}
      \end{minipage} &
      \begin{minipage}[t]{0.47\hsize}
        \centering
        \includegraphics[keepaspectratio, scale=0.33]{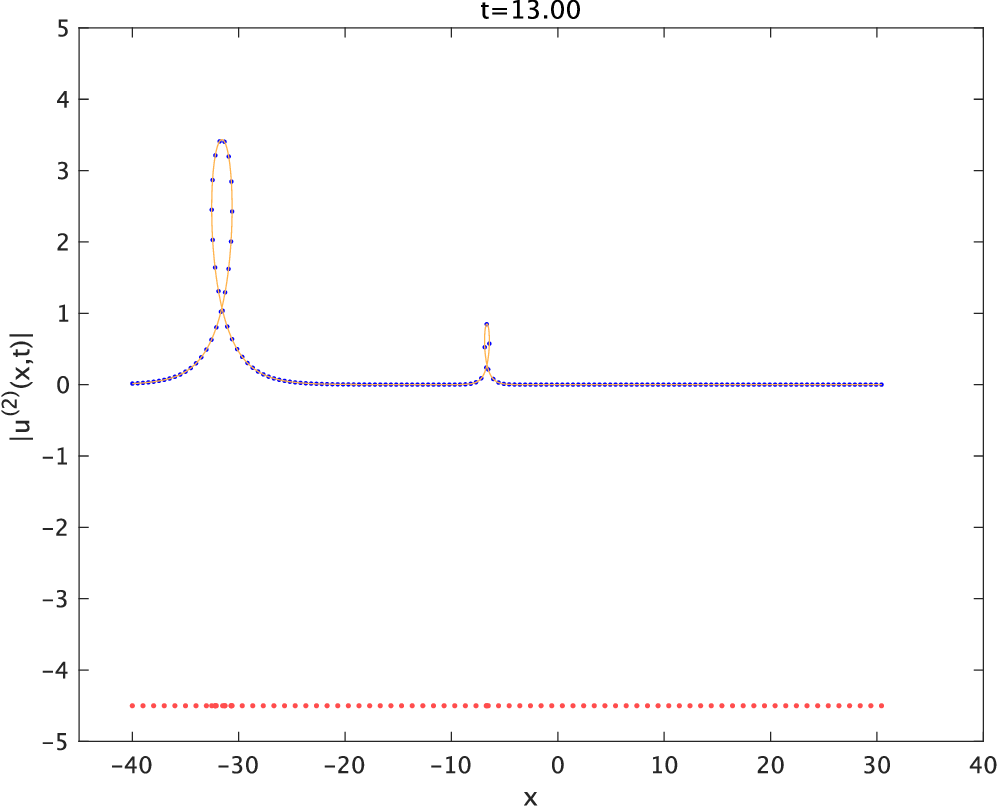}
      \end{minipage} 
       \end{tabular}
     \caption{Numerical simulation of the $|u^{(2)}|$-profile of the two-soliton solution of the 2-CSP equation. ${\rm maxerr}(|u^{(2)}|)=1.02\times 10^{-4}$.}
              \label{2CSP_2_abs_u2}
  \end{figure}

\begin{figure}[htbp]
 \begin{tabular}{cc}
      \begin{minipage}[t]{0.47\hsize}
       \centering
        \includegraphics[keepaspectratio, scale=0.33]{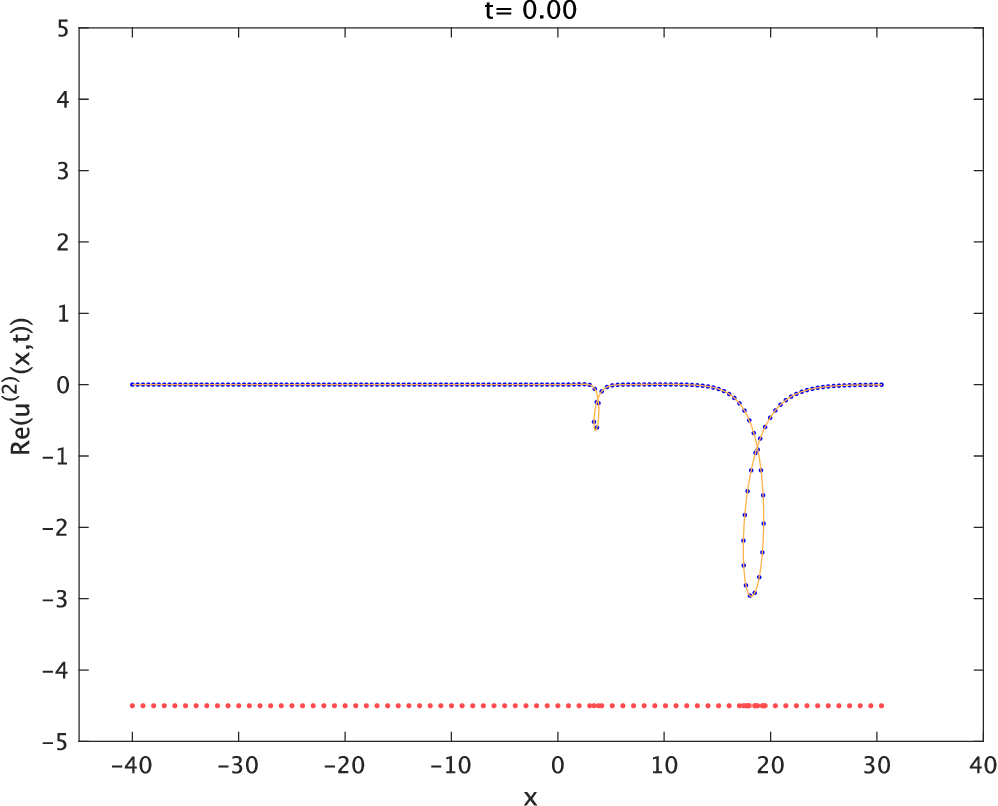}
      \end{minipage} &
      \begin{minipage}[t]{0.47\hsize}
        \centering
        \includegraphics[keepaspectratio, scale=0.33]{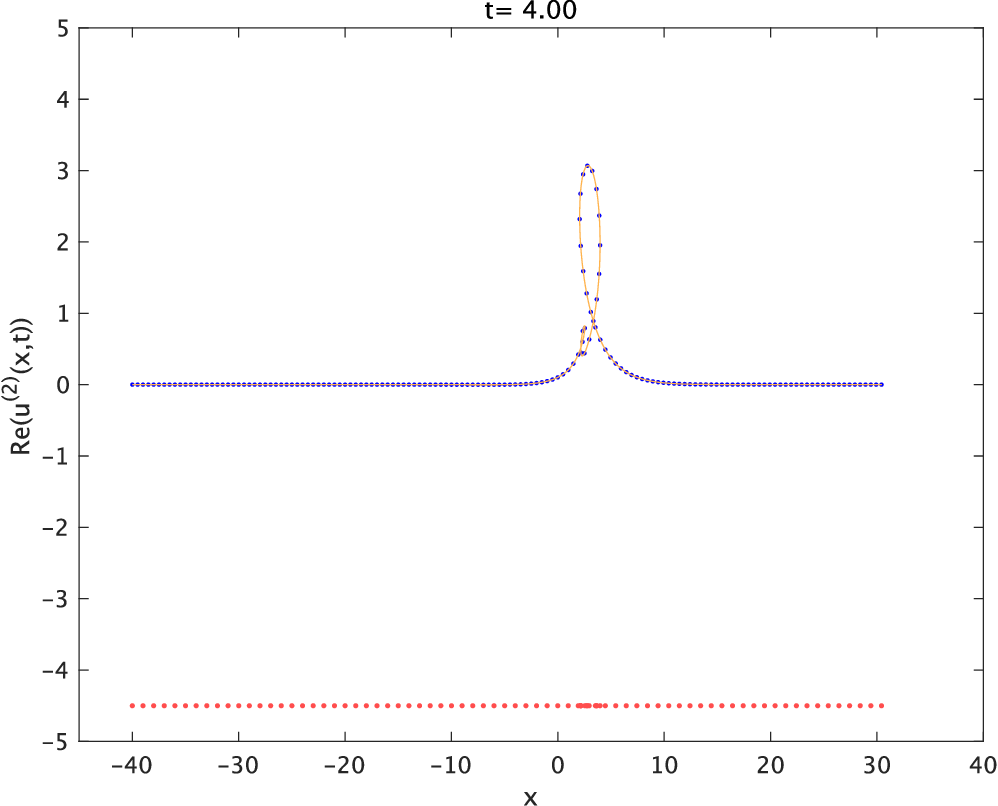}
      \end{minipage}\\ 
  
      \begin{minipage}[t]{0.47\hsize}
        \centering
        \includegraphics[keepaspectratio, scale=0.33]{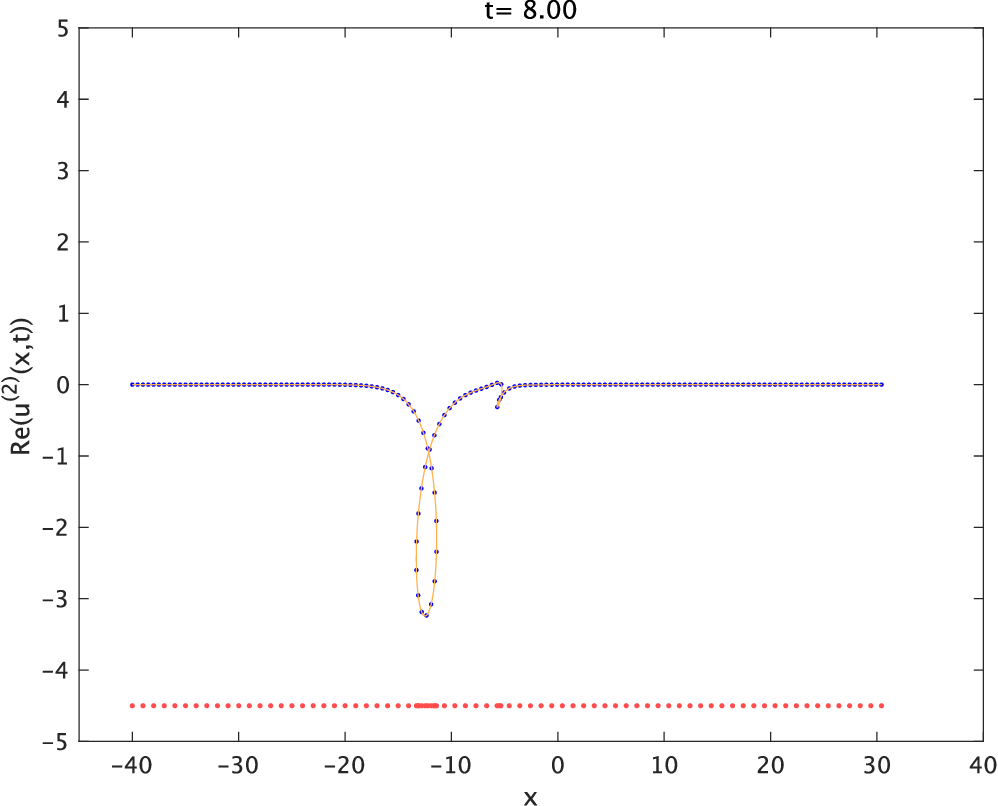}
      \end{minipage} &
      \begin{minipage}[t]{0.47\hsize}
        \centering
        \includegraphics[keepaspectratio, scale=0.33]{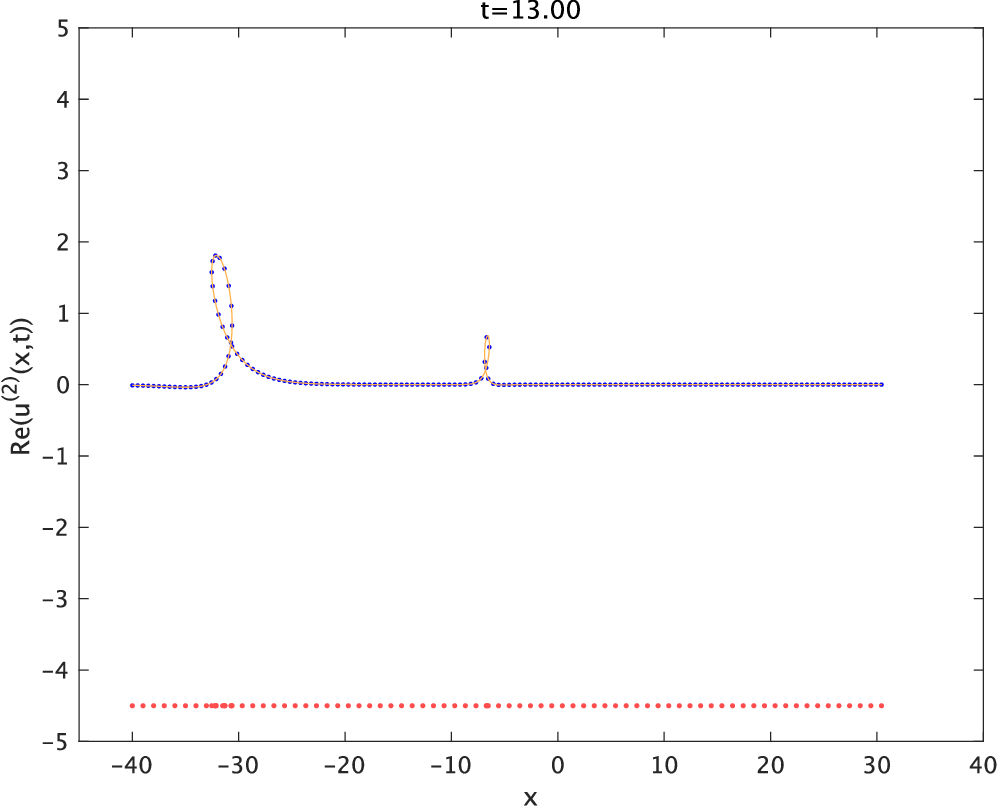}
      \end{minipage} 
       \end{tabular}
     \caption{Numerical simulation of the ${\rm Re}(u^{(2)})$-profile of the two-soliton solution of the 2-CSP equation. ${\rm maxerr}({\rm Re}(u^{(2)}))=2.27\times 10^{-3}$.}
              \label{2CSP_2_re_u2}
  \end{figure}

\begin{figure}[htbp]
 \begin{tabular}{cc}
      \begin{minipage}[t]{0.47\hsize}
       \centering
        \includegraphics[keepaspectratio, scale=0.33]{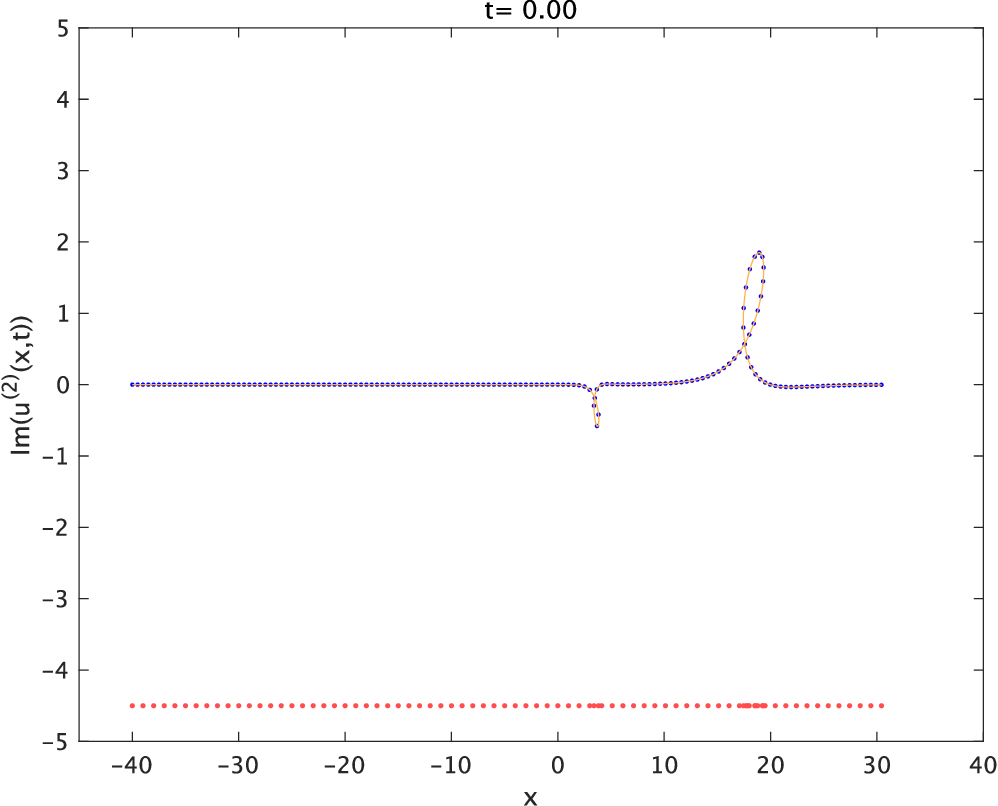}
      \end{minipage} &
      \begin{minipage}[t]{0.47\hsize}
        \centering
        \includegraphics[keepaspectratio, scale=0.33]{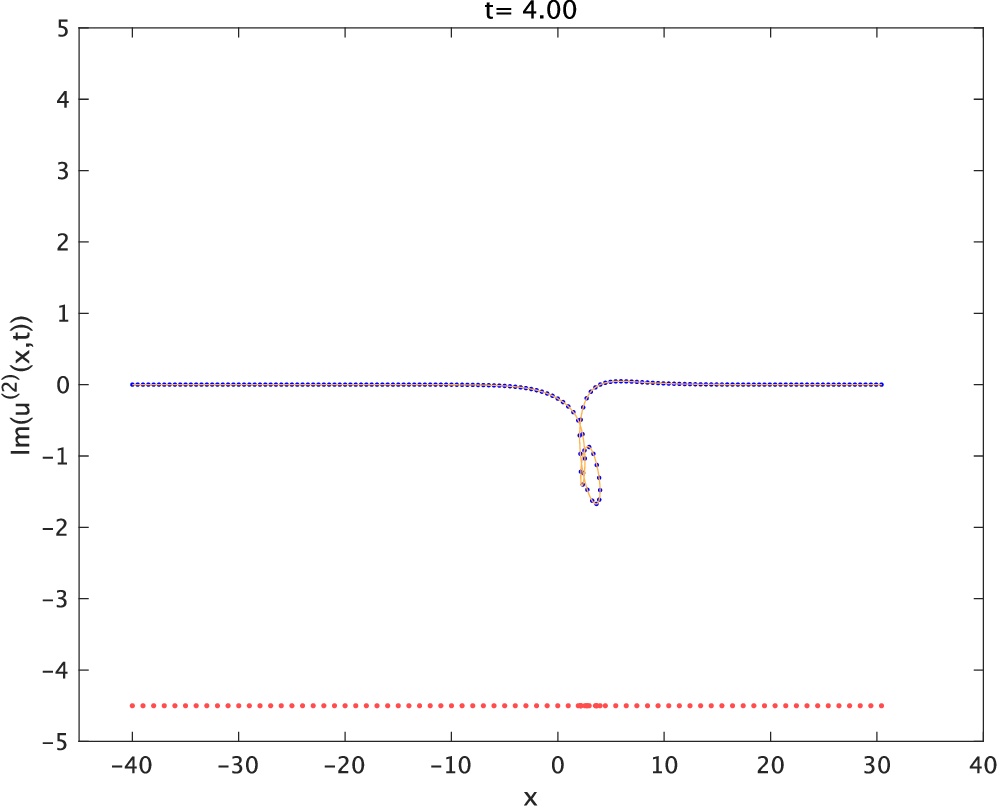}
      \end{minipage}\\ 
  
      \begin{minipage}[t]{0.47\hsize}
        \centering
        \includegraphics[keepaspectratio, scale=0.33]{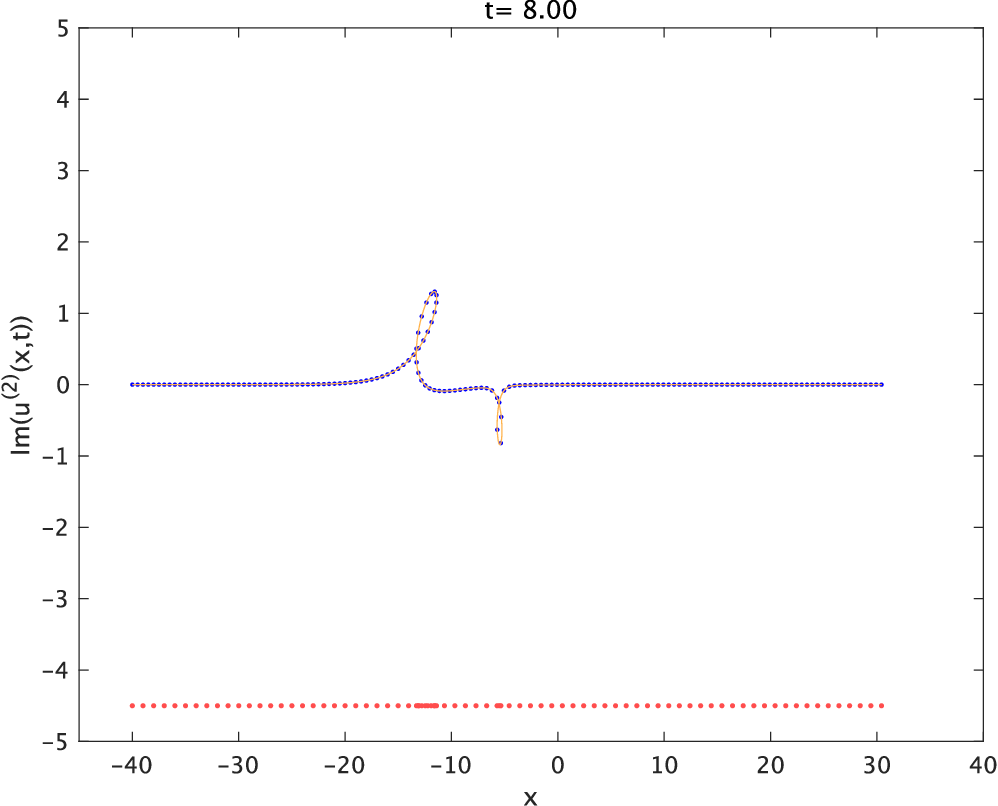}
      \end{minipage} &
      \begin{minipage}[t]{0.47\hsize}
        \centering
        \includegraphics[keepaspectratio, scale=0.33]{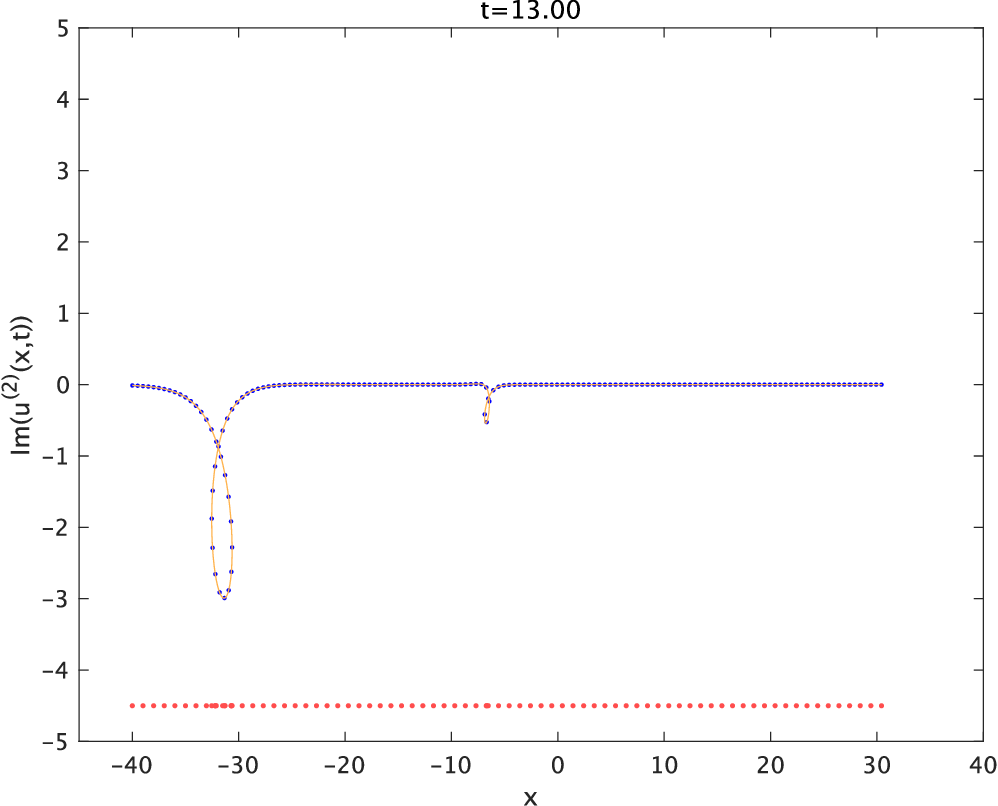}
      \end{minipage} 
       \end{tabular}
     \caption{Numerical simulation of the ${\rm Im}(u^{(2)})$-profile of the two-soliton solution of the 2-CSP equation. ${\rm maxerr}({\rm Im}(u^{(2)}))=2.27\times 10^{-3}$.}
              \label{2CSP_2_im_u2}
  \end{figure}

These numerical experiments show that the mesh automatically refines in
regions where the solutions vary rapidly, confirming that the fully discrete MCSP
system (\ref{fulldis}) provides a self-adaptive moving mesh scheme. For the parameter
sets considered, the numerical and exact solutions agree well near the soliton
peaks. The reported maximum relative errors remain small in these regions, and
two-soliton interactions do not substantially increase the error relative to the
one-soliton cases.

\end{section}

\begin{section}{Conclusion}
\label{sec_con}
We have established a new formulation of the MCSP equation that naturally encapsulates the CCSP equation as a reduction. 
Using Hirota's bilinear method, we derived a bilinear form of the MCSP equation and constructed its exact $N$-soliton solutions 
in Pfaffian form. Based on this bilinear formulation, we obtained integrable semi-discrete and fully discrete analogues of 
the MCSP equation together with their corresponding $N$-soliton solutions. 
By imposing appropriate complex conjugacy reductions, the corresponding semi-discrete and 
fully discrete CCSP equations were also derived.

The resulting fully discrete systems provide self-adaptive moving mesh schemes for the MCSP and CCSP equations. 
Numerical simulations demonstrated excellent agreement between the numerical and exact solutions, showing that 
the mesh points automatically concentrate in regions where the solutions vary rapidly. 
In particular, even for interacting two-soliton solutions, the proposed schemes accurately reproduce 
the exact solutions without any significant increase in the relative error compared with the one-soliton cases. 
These results indicate that the proposed integrable discretizations provide accurate numerical approximations for 
both one- and interacting two-soliton solutions.

\end{section}

\ack
This work was partially supported by JSPS KAKENHI Grant Numbers 
JP22K03441, JP23K22407, JP26K06919 and Waseda
University Grants for Special Research Projects.

\section*{Appendix}
\textbf{A. Detailed Pfaffian calculations for Theorem 4.1}\\
\textit{A.1. Details for the first bilinear equation.}
\setcounter{equation}{0}
\renewcommand{\theequation}{A.\arabic{equation}}
The following computation derives the $m$-shift formulas and equation (\ref{fd_first_relation}) used in Step 1 of the proof. From the definition of $\varphi_{j}(l,m)$, we obtain
\begin{eqnarray}
\fl{{\rm Pf}(j,k)_{l,m+1}={\rm Pf}(j,k)_{l,m}-\varphi_{j}(l,m+1)\varphi_{k}(l,m)+\varphi_{j}(l,m)\varphi_{k}(l,m+1)}\nonumber\\
\fl{\qquad={\rm Pf}(s_{m}^{+},d_{0},j,k)_{l,m},}\nonumber\\
\fl{{\rm Pf}(j,k)_{l,m-1}={\rm Pf}(j,k)_{l,m}-\varphi_{j}(l,m-1)\varphi_{k}(l,m)+\varphi_{j}(l,m)\varphi_{k}(l,m-1)}\nonumber\\
\fl{\qquad={\rm Pf}(s_{m}^{-},d_{0},j,k)_{l,m}.}\nonumber
\end{eqnarray}
Therefore, the $m$-shift relations for $f_{l,m}$ are
\begin{eqnarray}
\fl{f_{l,m+1}={\rm Pf}(s_{m}^{+},d_{0},\cdots)_{l,m},\qquad f_{l,m-1}={\rm Pf}(s_{m}^{-},d_{0},\cdots)_{l,m}.}
\label{shift_relation_ap}
\end{eqnarray}
Using the shift relations (\ref{shift_relation_ap}), the left-hand side of the first bilinear equation in (\ref{fullbilinear}) can be written as
\begin{eqnarray}
\fl{f_{l,m+1}f_{l,m-1}-f_{l,m}^{2}}\nonumber\\
\fl{={\rm Pf}(s_{m}^{+},d_{0},\cdots)_{l,m}{\rm Pf}(s_{m}^{-},d_{0},\cdots)_{l,m}-f_{l,m}^{2}}\nonumber\\
\fl{=\left(f_{l,m}+\displaystyle\sum_{j=1}^{M}(-1)^{j}\varphi_{j}(l,m+1){\rm Pf}(d_{0},\cdots, \hat{j},\cdots)_{l,m}\right)}\nonumber\\
\fl{\qquad\times\left(f_{l,m}+\displaystyle\sum_{j=1}^{M}(-1)^{j}\varphi_{j}(l,m-1){\rm Pf}(d_{0},\cdots, \hat{j},\cdots)_{l,m}\right)-f_{l,m}^{2}}\nonumber\\
\fl{=f_{l,m}\left(\displaystyle\sum_{j=1}^{M}(-1)^{j}\varphi_{j}(l,m-1){\rm Pf}(d_{0},\cdots, \hat{j},\cdots)_{l,m}\right.}\nonumber\\
\fl{\left.\qquad+\displaystyle\sum_{j=1}^{M}(-1)^{j}\varphi_{j}(l,m+1){\rm Pf}(d_{0},\cdots, \hat{j},\cdots)_{l,m}\right)}\nonumber\\
\fl{\qquad+\displaystyle\sum_{j=1}^{M}\displaystyle\sum_{k=1}^{M}(-1)^{j+k}\varphi_{j}(l,m+1){\rm Pf}(d_{0},\cdots, \hat{j},\cdots)_{l,m}}\nonumber\\
\fl{\qquad\quad\times\varphi_{k}(l,m-1){\rm Pf}(d_{0},\cdots, \hat{k},\cdots)_{l,m}}\nonumber\\
\fl{=f_{l,m}\left(\displaystyle\sum_{j=1}^{M}(-1)^{j}\Big(\varphi_{j}(l,m-1)+\varphi_{j}(l,m+1)\Big){\rm Pf}(d_{0},\cdots,\hat{j},\cdots)_{l,m}\right)}\nonumber\\
\fl{\qquad+\displaystyle\sum_{j=1}^{M}\displaystyle\sum_{k=1}^{M}(-1)^{j+k}\varphi_{j}(l,m+1){\rm Pf}(d_{0},\cdots,\hat{j},\cdots)_{l,m}}\nonumber\\
\fl{\qquad\quad\times\varphi_{k}(l,m-1){\rm Pf}(d_{0},\cdots,\hat{k},\cdots)_{l,m}}\nonumber\\
\fl{=f_{l,m}\displaystyle\sum_{j=1}^{M}(-1)^{j}\Big(\varphi_{j}(l,m+1)-2\varphi_{j}(l,m)+\varphi_{j}(l,m-1)\Big)}\nonumber\\
\fl{\qquad\times{\rm Pf}(d_{0},\cdots,\hat{j},\cdots)_{l,m}}\nonumber\\
\fl{\qquad+\displaystyle\sum_{j=1}^{M}\displaystyle\sum_{k=1}^{M}(-1)^{j+k}\Big(\varphi_{j}(l,m+1)\varphi_{k}(l,m-1)-\varphi_{j}(l,m)\varphi_{k}(l,m)\Big)}\nonumber\\
\fl{\qquad\quad\times{\rm Pf}(d_{0},\cdots,\hat{j},\cdots)_{l,m}{\rm Pf}(d_{0},\cdots,\hat{k},\cdots)_{l,m}.}\label{eq_h33}
\end{eqnarray}
We simplify the two terms in (\ref{eq_h33}) separately. The first term of (\ref{eq_h33}) can be written as
\begin{eqnarray}
\fl{\quad f_{l,m}\displaystyle\sum_{j=1}^{M}(-1)^{j}\Big(\varphi_{j}(l,m+1)-2\varphi_{j}(l,m)+\varphi_{j}(l,m-1)\Big){\rm Pf}(d_{0},\cdots,\hat{j},\cdots)_{l,m}}\nonumber\\
=f_{l,m}\displaystyle\sum_{j=1}^{M}(-1)^{j+1}\frac{-4b^{2}}{p_{j}^{2}-b^{2}}\varphi_{j}(l,m){\rm Pf}(d_{0},\cdots,\hat{j},\cdots)_{l,m}.\label{eq_h55}
\end{eqnarray}
Furthermore, the second term of (\ref{eq_h33}) can be written as
\begin{eqnarray}
\fl{\qquad\displaystyle\sum_{j=1}^{M}\displaystyle\sum_{k=1}^{M}(-1)^{j+k}\Big(\varphi_{j}(l,m+1)\varphi_{k}(l,m-1)-\varphi_{j}(l,m)\varphi_{k}(l,m)\Big)}\nonumber\\
\fl{\qquad\quad\times{\rm Pf}(d_{0},\cdots,\hat{j},\cdots)_{l,m}{\rm Pf}(d_{0},\cdots,\hat{k},\cdots)_{l,m}}\nonumber\\
\fl{=\displaystyle\sum_{1\leq j<k\leq M}(-1)^{j+k}\Big(\varphi_{j}(l,m+1)\varphi_{k}(l,m-1)+\varphi_{k}(l,m+1)\varphi_{j}(l,m-1)\Big.}\nonumber\\
\fl{\qquad\Big.\Big.-2\varphi_{j}(l,m)\varphi_{k}(l,m)\Big)}\nonumber\\
\fl{\qquad\quad\times{\rm Pf}(d_{0},\cdots,\hat{j},\cdots)_{l,m}{\rm Pf}(d_{0},\cdots,\hat{k},\cdots)_{l,m}}\nonumber\\
\fl{=\displaystyle\sum_{1\leq j<k\leq M}(-1)^{j+k+1}\Big(2\varphi_{j}(l,m)\varphi_{k}(l,m)-\varphi_{j}(l,m+1)\varphi_{k}(l,m-1)\Big.}\nonumber\\
\fl{\qquad\Big.\Big.-\varphi_{k}(l,m+1)\varphi_{j}(l,m-1)\Big)}\nonumber\\
\fl{\qquad\quad\times{\rm Pf}(d_{0},\cdots,\hat{j},\cdots)_{l,m}{\rm Pf}(d_{0},\cdots,\hat{k},\cdots)_{l,m}}\nonumber\\
\fl{=\displaystyle\sum_{1\leq j<k\leq M}(-1)^{j+k+1}\left(2-\frac{p_{j}+b}{p_{j}-b}\frac{p_{k}-b}{p_{k}+b}-\frac{p_{k}+b}{p_{k}-b}\frac{p_{j}-b}{p_{j}+b}\right)\varphi_{j}(l,m)\varphi_{k}(l,m)}\nonumber\\
\fl{\qquad\quad\times{\rm Pf}(d_{0},\cdots,\hat{j},\cdots)_{l,m}{\rm Pf}(d_{0},\cdots,\hat{k},\cdots)_{l,m}}\nonumber\\
\fl{=\displaystyle\sum_{1\leq j<k\leq M}(-1)^{j+k+1}\left(\frac{p_{j}+b}{p_{j}-b}-\frac{p_{k}+b}{p_{k}-b}\right)\left(\frac{p_{j}-b}{p_{j}+b}-\frac{p_{k}-b}{p_{k}+b}\right)}\nonumber\\
\fl{\qquad\quad\times\varphi_{j}(l,m)\varphi_{k}(l,m){\rm Pf}(d_{0},\cdots,\hat{j},\cdots)_{l,m}{\rm Pf}(d_{0},\cdots,\hat{k},\cdots)_{l,m}}\nonumber\\
\fl{=\displaystyle\sum_{1\leq j<k\leq M}(-1)^{j+k+1}\frac{-2b(p_{j}-p_{k})}{(p_{j}-b)(p_{k}-b)}\frac{2b(p_{j}-p_{k})}{(p_{j}+b)(p_{k}+b)}\varphi_{j}(l,m)\varphi_{k}(l,m)}\nonumber\\
\fl{\qquad\quad\times{\rm Pf}(d_{0},\cdots,\hat{j},\cdots)_{l,m}{\rm Pf}(d_{0},\cdots,\hat{k},\cdots)_{l,m}}\nonumber\\
\fl{=\displaystyle\sum_{1\leq j<k\leq M}(-1)^{j+k}\frac{4b^{2}(p_{j}^{2}-p_{k}^{2})}{(p_{j}^{2}-b^{2})(p_{k}^{2}-b^{2})}{\rm Pf}(j,k)_{l,m}}\nonumber\\
\fl{\qquad\quad\times{\rm Pf}(d_{0},\cdots,\hat{j},\cdots)_{l,m}{\rm Pf}(d_{0},\cdots,\hat{k},\cdots)_{l,m}}\nonumber\\
\fl{=\displaystyle\sum_{1\leq j<k\leq M}(-1)^{j+k+1}\displaystyle\frac{4b^{2}}{p_{j}^{2}-b^{2}}{\rm Pf}(j,k)_{l,m}{\rm Pf}(d_{0},\cdots,\hat{j},\cdots)_{l,m}}\nonumber\\
\fl{\qquad\quad\times{\rm Pf}(d_{0},\cdots,\hat{k},\cdots)_{l,m}}\nonumber\\
\fl{-\displaystyle\sum_{1\leq j<k\leq M}(-1)^{j+k+1}\displaystyle\frac{4b^{2}}{p_{k}^{2}-b^{2}}{\rm Pf}(j,k)_{l,m}{\rm Pf}(d_{0},\cdots,\hat{j},\cdots)_{l,m}}\nonumber\\
\fl{\qquad\quad\times{\rm Pf}(d_{0},\cdots,\hat{k},\cdots)_{l,m}}\nonumber\\
\fl{=\displaystyle\sum_{j=1}^{M}\displaystyle\sum_{k=j+1}^{M}(-1)^{j+k+1}\displaystyle\frac{4b^{2}}{p_{j}^{2}-b^{2}}{\rm Pf}(j,k)_{l,m}{\rm Pf}(d_{0},\cdots,\hat{j},\cdots)_{l,m}}\nonumber\\
\fl{\qquad\quad\times{\rm Pf}(d_{0},\cdots,\hat{k},\cdots)_{l,m}}\nonumber\\
\fl{-\displaystyle\sum_{k=1}^{M}\displaystyle\sum_{j=1}^{k-1}(-1)^{j+k+1}\displaystyle\frac{4b^{2}}{p_{k}^{2}-b^{2}}{\rm Pf}(j,k)_{l,m}{\rm Pf}(d_{0},\cdots,\hat{j},\cdots)_{l,m}}\nonumber\\
\fl{\qquad\quad\times{\rm Pf}(d_{0},\cdots,\hat{k},\cdots)_{l,m}}\nonumber\\
\fl{=\displaystyle\sum_{j=1}^{M}\displaystyle\sum_{k=1}^{M}(-1)^{j+k}\displaystyle\frac{-4b^{2}}{p_{j}^{2}-b^{2}}{\rm Pf}(j,k)_{l,m}{\rm Pf}(d_{0},\cdots,\hat{j},\cdots)_{l,m}}\nonumber\\
\fl{\qquad\quad\times{\rm Pf}(d_{0},\cdots,\hat{k},\cdots)_{l,m}.}\label{eq_h44}
\end{eqnarray}
Substituting (\ref{eq_h55}) and (\ref{eq_h44}) into (\ref{eq_h33}) gives
\begin{eqnarray}
\fl{\quad f_{l,m+1}f_{l,m-1}-f_{l,m}^{2}}\nonumber\\
\fl{=f_{l,m}\displaystyle\sum_{j=1}^{M}(-1)^{j+1}\displaystyle\frac{-4b^{2}}{p_{j}^{2}-b^{2}}\varphi_{j}(l,m){\rm Pf}(d_{0},\cdots,\hat{j},\cdots)_{l,m}}\nonumber\\
\fl{\qquad+\displaystyle\sum_{j=1}^{M}\displaystyle\sum_{k=1}^{M}(-1)^{j+k}\displaystyle\frac{-4b^{2}}{p_{j}^{2}-b^{2}}{\rm Pf}(j,k)_{l,m}{\rm Pf}(d_{0},\cdots,\hat{j},\cdots)_{l,m}}\nonumber\\
\fl{\qquad\quad\times{\rm Pf}(d_{0},\cdots,\hat{k},\cdots)_{l,m}}\nonumber\\
\fl{=b^{2}\displaystyle\sum_{j=1}^{M}(-1)^{j}\displaystyle\frac{-4}{p_{j}^{2}-b^{2}}{\rm Pf}(d_{0},\cdots,\hat{j},\cdots)_{l,m}}\nonumber\\
\fl{\qquad\quad\times\left(\displaystyle\sum_{k=1}^{M}(-1)^{k}{\rm Pf}(j,k)_{l,m}{\rm Pf}(d_{0},\cdots,\hat{k},\cdots)_{l,m}-f_{l,m}\varphi_{j}(l,m)\right).}\label{eq_h66}
\end{eqnarray}
Now, expanding the trivially vanishing Pfaffian ${\rm Pf}(j, d_{0},1,\cdots, j, \cdots M, M^{\prime},\cdots,1^{\prime})_{l,m}$, we obtain
\begin{eqnarray}
{\displaystyle\sum_{k=1}^{M}(-1)^{k}{\rm Pf}(j,k)_{l,m}{\rm Pf}(d_{0},\cdots,\hat{k},\cdots)_{l,m}}\nonumber\\
{\qquad =f_{l,m}\varphi_{j}(l,m)+(-1)^{j}{\rm Pf}(d_{0},\cdots, \hat{j^{\prime}},\cdots)_{l,m}}.\label{eq_h777}
\end{eqnarray}
Substituting (\ref{eq_h777}) into (\ref{eq_h66}) gives
\begin{eqnarray}
\fl{f_{l,m+1}f_{l,m-1}-f_{l,m}^{2}=b^{2}\displaystyle\sum_{j=1}^{M}\displaystyle\frac{-4}{p_{j}^{2}-b^{2}}{\rm Pf}(d_{0},\cdots,\hat{j},\cdots)_{l,m}{\rm Pf}(d_{0},\cdots,\hat{j^{\prime}},\cdots)_{l,m}},\nonumber
\end{eqnarray}
which is precisely equation (\ref{fd_first_relation}).
\\
\\
\textit{A.2. Details for the second and third bilinear equations.}
{\footnotesize
\begin{eqnarray}
\fl{f_{l+1,m+1}-f_{l,m}}\nonumber\\
\fl{\quad ={\rm Pf}(d_{0},s_{l}^{+},\cdots)_{l,m+1}-{\rm Pf}(\cdots)_{l,m}}\nonumber\\
\fl{\quad =\left(\displaystyle\sum_{j=1}^{M}(-1)^{j}{\rm Pf}(d_{0},j)_{l,m+1}{\rm Pf}(s_{l}^{+},\cdots \hat{j},\cdots)_{l,m+1}\right.}\nonumber\\
\fl{\left.\qquad+{\rm Pf}(d_{0},s_{l}^{+})_{l,m+1}{\rm Pf}(\cdots)_{l,m+1}-{\rm Pf}(\cdots)_{l,m}\right)}\nonumber\\
\fl{\quad =\left(\displaystyle\sum_{j=1}^{M}(-1)^{j}{\rm Pf}(s_{m}^{+},j)_{l,m}{\rm Pf}(s_{l}^{+},\cdots \hat{j},\cdots)_{l,m+1}\right.}\nonumber\\
\fl{\left.\qquad-{\rm Pf}(\cdots)_{l,m}+{\rm Pf}(s_{m}^{+},d_{0},\cdots)_{l,m}\right)}\nonumber\\
\fl{\quad =-{\rm Pf}(\cdots)_{l,m}+{\rm Pf}(s_{m}^{+},d_{0},\cdots)_{l,m}}\nonumber\\
\fl{\qquad +\left[\displaystyle\sum_{j=1}^{M}(-1)^{j}{\rm Pf}(s_{m}^{+},j)_{l,m}
\Bigg(\displaystyle\sum_{k=1}^{j-1}(-1)^{k+1}{\rm Pf}(s_{l}^{+},k)_{l,m+1}{\rm Pf}(1,\cdots \hat{k},\cdots \hat{j},\cdots)_{l,m+1}\bigg.\right.}\nonumber\\
\fl{\left.\Bigg.\qquad+\displaystyle\sum_{k=j+1}^{M}(-1)^{k}{\rm Pf}(s_{l}^{+},k)_{l,m+1}{\rm Pf}(1,\cdots \hat{j},\cdots \hat{k},\cdots)_{l,m+1}\Bigg)\right]}\nonumber\\
\fl{\quad =-{\rm Pf}(\cdots)_{l,m}+{\rm Pf}(s_{m}^{+},d_{0},\cdots)_{l,m}}\nonumber\\
\fl{\qquad +\left[\displaystyle\sum_{j=1}^{M}(-1)^{j}\varphi_{j}(l,m+1)\times
\Bigg(\displaystyle\sum_{k=1}^{j-1}(-1)^{k+1}\varphi_{k}(l+1,m+1){\rm Pf}(s_{m}^{+},d_{0},1,\cdots \hat{k},\cdots \hat{j},\cdots)_{l,m}\Bigg.\right.}\nonumber\\
\fl{\left.\Bigg.\qquad+\displaystyle\sum_{k=j+1}^{M}(-1)^{k}\varphi_{k}(l+1,m+1){\rm Pf}(s_{m}^{+},d_{0},1,\cdots \hat{j},\cdots \hat{k},\cdots)_{l,m}\Bigg)\right]}\nonumber\\
\fl{\quad =-{\rm Pf}(\cdots)_{l,m}+{\rm Pf}(s_{m}^{+},d_{0},\cdots)_{l,m}}\nonumber\\
\fl{\qquad +\displaystyle\sum_{k=1}^{M}(-1)^{k}\varphi_{k}(l+1,m+1)\times
\Bigg(\displaystyle\sum_{j=1}^{k-1}(-1)^{j}\varphi_{j}(l,m+1){\rm Pf}(s_{m}^{+},d_{0},1,\cdots \hat{j},\cdots \hat{k},\cdots)_{l,m}\Bigg.}\nonumber\\
\fl{\Bigg.\qquad+\displaystyle\sum_{j=k+1}^{M}(-1)^{j+1}\varphi_{j}(l,m+1){\rm Pf}(s_{m}^{+},d_{0},1,\cdots \hat{k},\cdots \hat{j},\cdots)_{l,m}\Bigg).}\label{eq_h77}
\end{eqnarray}
}
Now, expanding the trivially vanishing Pfaffian
\begin{eqnarray*}
{\rm Pf}(s_{m}^{+}, d_{0}, s_{m}^{+},1, \cdots, \hat{k}, \cdots, M, M^{\prime},\cdots,1^{\prime})_{l,m},
\end{eqnarray*}
we obtain
\begin{eqnarray}
\fl{\sum_{j=1}^{k-1}(-1)^{j}{\rm Pf}(s_{m}^{+},j)_{l,m}{\rm Pf}(s_{m}^{+},d_{0},1,\cdots \hat{j},\cdots \hat{k},\cdots)_{l,m}}\nonumber\\
\fl{\qquad +\displaystyle\sum_{j=k+1}^{M}(-1)^{j+1}{\rm Pf}(s_{m}^{+},j)_{l,m}{\rm Pf}(s_{m}^{+},d_{0},1,\cdots \hat{k},\cdots \hat{j},\cdots)_{l,m}}\nonumber\\
\fl{\qquad ={\rm Pf}(d_{0},s_{m}^{+})_{l,m}{\rm Pf}(s_{m}^{+},\cdots, \hat{k},\cdots)_{l,m}}
\label{identity}
\end{eqnarray}
Multiplying (\ref{eq_h77}) by $(1-ab)/(1+ab)$ and using the identity (\ref{identity}) together with
$\frac{1-ab}{1+ab}\varphi_k(l+1,m+1)
=-\frac{1-ab}{1+ab}\varphi_k(l,m)+\varphi_k(l+1,m)+\varphi_k(l,m+1)$,
we obtain
\begin{eqnarray}
\fl{\frac{1-ab}{1+ab}(f_{l+1,m+1}-f_{l,m})}\nonumber\\
\fl{=-\frac{1-ab}{1+ab}{\rm Pf}(\cdots)_{l,m}+\frac{1-ab}{1+ab}{\rm Pf}(s_{m}^{+},d_{0},\cdots)_{l,m}}\nonumber\\
\fl{\qquad +\displaystyle\sum_{k=1}^{M}(-1)^{k}\Bigg(-\frac{1-ab}{1+ab}\varphi_{k}(l,m)+\varphi_{k}(l+1,m)+\varphi_{k}(l,m+1)\Bigg)}\nonumber\\
\fl{\qquad\times{\rm Pf}(d_{0},s_{m}^{+})_{l,m}{\rm Pf}(s_{m}^{+},\cdots,\hat{k},\cdots)_{l,m}}\nonumber\\
\fl{=-\frac{1-ab}{1+ab}{\rm Pf}(\cdots)_{l,m}+\frac{1-ab}{1+ab}{\rm Pf}(s_{m}^{+},d_{0},\cdots)_{l,m}}\nonumber\\
\fl{\qquad+\frac{1-ab}{1+ab}\displaystyle\sum_{k=1}^{M}(-1)^{k}{\rm Pf}(d_{0},k)_{l,m}{\rm Pf}(s_{m}^{+},\cdots,\hat{k},\cdots)_{l,m}}\nonumber\\
\fl{\qquad-\displaystyle\sum_{k=1}^{M}(-1)^{k}{\rm Pf}(s_{l}^{+},k){\rm Pf}(s_{m}^{+},\cdots,\hat{k},\cdots)_{l,m}}\nonumber\\
\fl{\qquad-\displaystyle\sum_{k=1}^{M}(-1)^{k}{\rm Pf}(s_{m}^{+},k){\rm Pf}(s_{m}^{+},\cdots,\hat{k},\cdots)_{l,m}}\nonumber\\
\fl{=-\frac{1-ab}{1+ab}{\rm Pf}(\cdots)_{l,m}+\frac{1-ab}{1+ab}{\rm Pf}(s_{m}^{+},d_{0},\cdots)_{l,m}}\nonumber\\
\fl{\qquad +\frac{1-ab}{1+ab}\Bigg({\rm Pf}(d_{0},s_{m}^{+},\cdots)_{l,m}+{\rm Pf}(\cdots)_{l,m}\Bigg)}\nonumber\\
\fl{\qquad-{\rm Pf}(s_{l}^{+},s_{m}^{+},\cdots)_{l,m}-{\rm Pf}(s_{m}^{+},s_{m}^{+},\cdots)_{l,m}}\nonumber\\
\fl{=-{\rm Pf}(s_{l}^{+},s_{m}^{+},\cdots)_{l,m},}\nonumber
\end{eqnarray}
which is precisely equation (\ref{fg_shift1}).

We next derive the corresponding shift relations for $g_{l,m}^{(\mu)}$. Expanding the trivially vanishing Pfaffian ${\rm Pf}(s_{m}^{+},s_{m}^{+},d_{0},a^{(\mu)},\cdots)_{l,m}$, we obtain
\begin{eqnarray}
\fl{g_{l,m+1}^{(\mu)}=-{\rm Pf}(d_{0},a^{(\mu)},\cdots)_{l,m+1}}\nonumber\\
\fl{\qquad=-\displaystyle\sum_{j=1}^{M}(-1)^{j}{\rm Pf}(d_{0},j)_{l,m+1}{\rm Pf}(a^{(\mu)},\cdots,\hat{j},\cdots)_{l,m+1}}\nonumber\\
\fl{\qquad =\displaystyle\sum_{j=1}^{M}(-1)^{j+1}{\rm Pf}(s_{m}^{+},j)_{l,m}{\rm Pf}(s_{m}^{+},d_{0},a^{(\mu)},\cdots ,\hat{j},\cdots)_{l,m}}\nonumber\\
\fl{\qquad =-{\rm Pf}(s_{m}^{+},a^{(\mu)},\cdots)_{l,m}.}
\label{g_{l,m+1}}
\end{eqnarray}
Combining (\ref{g_{l,m+1}}) with the corresponding $l$- and $m$-shift relations gives
\begin{eqnarray}
\fl{\frac{1-ab}{1+ab}\big(g^{(\mu)}_{l+1,m+1}+g^{(\mu)}_{l,m}\big)}\nonumber\\
\fl{\quad = \frac{1-ab}{1+ab}\Big(-{\rm Pf}(s_{m}^{+},a^{(\mu)},\cdots)_{l+1,m}+{\rm Pf}(a^{(\mu)},d_{0},\cdots)_{l,m}\Big)}\nonumber\\
\fl{\quad =\frac{1-ab}{1+ab}{\rm Pf}(a^{(\mu)},d_{0},\cdots)_{l,m}}\nonumber\\
\fl{\qquad-\frac{1-ab}{1+ab}\displaystyle\sum_{j=1}^{M}(-1)^{j}{\rm Pf}(s_{m}^{+},j)_{l+1,m}{\rm Pf}(a^{(\mu)},\cdots,\hat{j},\cdots)_{l+1,m}}\nonumber\\
\fl{\quad =\frac{1-ab}{1+ab}{\rm Pf}(a^{(\mu)},d_{0},\cdots)_{l,m}}\nonumber\\
\fl{\qquad-\frac{1-ab}{1+ab}\displaystyle\sum_{j=1}^{M}(-1)^{j}\varphi_{j}(l+1,m+1){\rm Pf}(d_{0},s_{l}^{+},a^{(\mu)},\cdots,\hat{j},\cdots)_{l,m}}\nonumber\\
\fl{\quad =\frac{1-ab}{1+ab}{\rm Pf}(a^{(\mu)},d_{0},\cdots)_{l,m}}\nonumber\\
\fl{\qquad +\displaystyle\sum_{j=1}^{M}(-1)^{j}\left(\frac{1-ab}{1+ab}\varphi_{j}(l,m)-\varphi_{j}(l+1,m)-\varphi_{j}(l,m+1)\right)}\nonumber\\ 
\fl{\quad \qquad\times{\rm Pf}(d_{0},s_{l}^{+},a^{(\mu)},\cdots,\hat{j},\cdots)_{l,m}}\nonumber\\
\fl{\quad =\frac{1-ab}{1+ab}{\rm Pf}(a^{(\mu)},d_{0},\cdots)_{l,m}+\frac{1-ab}{1+ab}{\rm Pf}(d_{0},a^{(\mu)},\cdots)_{l,m}}\nonumber\\
\fl{\qquad-{\rm Pf}(s_{l}^{+},a^{(\mu)},\cdots)_{l,m}+{\rm Pf}(s_{l}^{+},a^{(\mu)},\cdots)_{l,m}
-{\rm Pf}(s_{m}^{+},d_{0},s_{l}^{+},a^{(\mu)},\cdots)_{l,m}}\nonumber\\
\fl{\quad =-{\rm Pf}(s_{m}^{+},d_{0},s_{l}^{+},a^{(\mu)},\cdots)_{l,m},}\nonumber
\end{eqnarray}
which is precisely equation (\ref{fg_shift2}).
\\
\\
\textbf{B. Explicit one- and two-soliton solutions of the 4-SP equation}\\
For completeness, we list the one- and two-soliton solutions of the 4-SP
equation (\ref{4SP}).\\
\noindent
One-soliton:
 \begin{eqnarray}
\fl{u^{(1)}=\frac{g^{(1)}}{f},\quad u^{(2)}=\frac{g^{(2)}}{f},\quad v^{(1)}=\frac{h^{(1)}}{f},\quad v^{(2)}=\frac{h^{(2)}}{f},\quad x=X-2(\log{f})_{T},\quad t=T,}\nonumber\\
\fl{f=1+\displaystyle\frac{1}{4}\left(\displaystyle\frac{p_{1}p_{2}}{p_{1}+p_{2}}\right)^{2}\varphi_{1}\varphi_{2}(c_{11}a^{(1)}_{1}b^{(1)}_{2}+c_{12}a^{(1)}_{1}b^{(2)}_{2}+c_{21}a^{(2)}_{1}b^{(1)}_{2}+c_{22}a^{(2)}_{1}b^{(2)}_{2}),}\nonumber\\
\fl{g^{(1)}=a^{(1)}_{1}\varphi_{1},\quad g^{(2)}=a^{(2)}_{1}\varphi_{1},\quad h^{(1)}=b^{(1)}_{2}\varphi_{2},\quad h^{(2)}=b^{(2)}_{2}\varphi_{2},}\nonumber
\end{eqnarray}
where $\varphi_{i}=B_{i}e^{p_{i}X+p_{i}^{-1}T}$ and $i=1,2$.
\\
\noindent
Two-soliton:
Using the notation in (\ref{con_compact_coefficients}) with $I=J=\{1,2\}$,
the two-soliton solution can be written as
\begin{eqnarray}
\fl u^{(\mu)}&=&\frac{g^{(\mu)}}{f},\quad
v^{(\nu)}=\frac{h^{(\nu)}}{f},\quad
x=X-2(\log f)_T,\quad t=T,\qquad \mu,\nu=1,2,\nonumber\\
\fl f&=&1+\tilde{C}_{14}\varphi_1\varphi_4+\tilde{C}_{13}\varphi_1\varphi_3
+\tilde{C}_{24}\varphi_2\varphi_4+\tilde{C}_{23}\varphi_2\varphi_3\nonumber\\
\fl &&+\Delta\left(\prod_{1\leq i<j\leq4}R_{ij}\right)
\varphi_1\varphi_2\varphi_3\varphi_4,\nonumber\\
\fl g^{(\mu)}&=&a_{1}^{(\mu)}\left(\varphi_1+R_{12}R_{14}\tilde{C}_{24}\varphi_1\varphi_2\varphi_4
+R_{12}R_{13}\tilde{C}_{23}\varphi_1\varphi_2\varphi_3\right)\nonumber\\
\fl &&+a_{2}^{(\mu)}\left(\varphi_2-R_{12}R_{24}\tilde{C}_{14}\varphi_1\varphi_2\varphi_4
-R_{12}R_{23}\tilde{C}_{13}\varphi_1\varphi_2\varphi_3\right),\qquad \mu=1,2,\nonumber\\
\fl h^{(\nu)}&=&b_{3}^{(\nu)}\left(\varphi_3-R_{23}R_{34}\tilde{C}_{24}\varphi_2\varphi_3\varphi_4
-R_{13}R_{34}\tilde{C}_{14}\varphi_1\varphi_3\varphi_4\right)\nonumber\\
\fl &&+b_{4}^{(\nu)}\left(\varphi_4+R_{24}R_{34}\tilde{C}_{23}\varphi_2\varphi_3\varphi_4
+R_{14}R_{34}\tilde{C}_{13}\varphi_1\varphi_3\varphi_4\right),\qquad \nu=1,2.\nonumber
\end{eqnarray}
where $\varphi_{i}=B_{i}e^{p_{i}X+p_{i}^{-1}T}$ and $i=1,2,3,4$.

\end{document}